\documentclass[sigconf]{acmart}

\usepackage{paralist}
\usepackage{multicol}
\usepackage[utf8x]{inputenc}
\usepackage{fullpage}
\usepackage{tikz}
\tikzset{>=latex, 
	point/.style = {circle,draw,thick,minimum size=2mm,inner sep=0pt}
}
\usepackage{xspace}
\usepackage{stmaryrd}
\usepackage{etoolbox}
\usepackage{amsmath,amssymb}
\usepackage{amsthm}

\usepackage{paralist}
\newcommand{\baseschema}{\kw{B}}
\newcommand{\connq}{\kw{Conn}}
\newcommand{\viewschema}{\views}
\newcommand{\configsep}{\kw{;}}
\newcommand{\inputbegin}{\kw{InpBegin}}
\newcommand{\inputend}{\kw{InpEnd}}
\newcommand{\runend}{\kw{RunEnd}}
\newcommand{\proofterm}{t}
\newcommand{\tminput}{\kw{input}}

\newcommand{\badlyshaped}{\kw{badly-shaped}}
\newcommand{\prerun}{\kw{pre-run}}
\newcommand{\tmsucc}{\kw{Succ}}
\newcommand{\tmaccept}{\kw{Accept}}

\newcommand{\ysucc}{\kw{YSucc}}
\newcommand{\xsucc}{\kw{XSucc}}
\newcommand{\xproj}{\kw{XProj}}
\newcommand{\yproj}{\kw{YProj}}
\newcommand{\false}{\kw{False}}
\newcommand{\true}{\kw{True}}
\newcommand{\ha}{\kw{HA}}
\newcommand{\va}{\kw{VA}}
\newcommand{\qstart}{\kw{start}}
\newcommand{\qverify}{\kw{verify}}
\newcommand{\qhelper}{\kw{helper}}
\newcommand{\producttest}{\kw{ProductTest}}

\newcommand{\tw}{\kw{tw}}

\newtheorem{theorem}{Theorem}
\newtheorem{proposition}{Proposition}
\newtheorem{lemma}{Lemma}

\newtheorem{example}{Example}

\newtheorem{claim}{Claim}
\newtheorem{fact}{Fact}

\newcommand{\adom}{\kw{adom}}
\newcommand{\node}{v}

\newcommand{\Nodes}{\kw{Nodes}}

\newcommand{\viewimageclass}{{\mathbb{V}}}

\newcommand{\yend}{\kw{YEnd}}
\newcommand{\xend}{\kw{XEnd}}

\newcommand{\factof}{\kw{FactOf}}
\newcommand{\ruleof}{\kw{RuleOf}}
\newcommand{\goal}{\kw{Goal}}
\newcommand{\CQAppr}{\mathsf{CQAppr}}
\newcommand{\q}{\mathbf{q}}

\newcommand{\goalpred}{\goal}
\newcommand{\code}{\kw{Code}}
\newcommand{\unpred}{\kw{UnPred}(\aschema, k)}
\newcommand{\binpred}{\kw{BinPred}(\aschema, k)}

\newcommand{\myeat}[1]{}

\newcommand{\tspan}{l}
\newcommand{\fgdatalog}{\kw{FGDL}}
\newcommand{\connquery}{\kw{Conn}}
\newcommand{\rewriting}{\kw{R}}

\newcommand{\outcome}[2]{\doutput{#1}{#2}}

\newcommand{\datalogarrow}{\leftarrow}
\newcommand{\datalogwedge}{,}
\newcommand{\datalogprog}{\Pi}

\newcommand{\canondb}{\kw{Canondb}}
\newcommand{\inst}{{\mathcal I}}
\newcommand{\jnst}{{\mathcal J}}
\newcommand{\vinst}{\jnst}

\newcommand{\coderel}{T}
\newcommand{\TD}{{T\kern-1.1mm{}D}}

\newcommand{\Adompred}{Adom}

\newcommand{\vertices}{{\textsc{vertices}}}
\newcommand{\edges}{{\textsc{edges}}}

\newcommand{\automaton}{\mathfrak{A}}
\newcommand{\fpeval}[2]{\mathsf{FPEval}(#1,#2)}
\newcommand{\doutput}[2]{\mathsf{Output}(#1,#2)}

\newcommand{\nb}[1]{\textcolor{red}{\bf!}%
\marginpar
 {\parbox{20mm}{\scriptsize\textcolor{red}{\raggedright #1}}}}
\renewcommand{\vec}[1]{\boldsymbol{#1}}
\newcommand{\aschema}{\mathbf{S}}
\newcommand{\schema}{\aschema}
\newcommand{\views}{\mathbf{V}}
\newcommand{\D}{\mathcal{D}}
\newcommand{\T}{\mathcal{T}}
\newcommand{\TreeAlphabet}{\kw{TreeAlph}}

\def\Q{{\mathcal{Q}}}
\def\F{{\mathcal{F}}}

\def\restr{\!\restriction\!}

\renewcommand{\C}{\mathbb{C}}

\newcommand{\dom}{{\textrm dom}}

\newcommand{\complexity}[1]{\textsc{#1}}
\newcommand{\ptime}{\complexity{PTime}\xspace}
\newcommand{\aczero}{\complexity{AC}^0 \xspace}
\newcommand{\expspace}{\complexity{ExpSpace}\xspace}
\newcommand{\exptime}{\complexity{ExpTime}\xspace}

\newcommand{\twoexptime}{\complexity{2ExpTime}\xspace}
\newcommand{\twoexp}{\twoexptime}
\newcommand{\threeexp}{\complexity{3ExpTime}\xspace}
\newcommand{\np}{\complexity{NP}\xspace}
\newcommand{\conp}{\complexity{co-NP}\xspace}

\newcommand{\decode}[1]{\mathfrak{D}(#1)}

\renewcommand{\varphi}{\phi}
\newcommand{\kw}[1]{\textsc{#1}}

\newcommand{\lemmagfpbwd}[1]{\hyperref[lemma:backwards-gfp]{#1}\xspace}
\newcommand{\lemmagnfpbwd}[1]{\hyperref[lemma:backwards-gnfp]{#1}\xspace}

\newcommand{\view}{V}

\newcommand{\myparagraph}[1]{{\textbf #1.}}

\title{On monotonic determinacy and rewritability \\ for recursive queries and views}
\author{Michael Benedikt, Stanislav Kikot, Piotr Ostropolski-Nalewaja,  and Miguel Romero}

\begin{document}
\begin{abstract}
A query $Q$ is monotonically determined over a set of views $\views$ if $Q$ can be expressed
as a monotonic function of the view image. In the case of relational algebra views and queries, monotonic determinacy
coincides with rewritability as a union of conjunctive queries, and it is decidable in important special cases,
such as for CQ views and queries \cite{NSV,thebook}. We investigate the situation for views and queries
in the recursive query language Datalog. We give both positive and negative results about the ability
to decide monotonic determinacy, and also about the co-incidence of monotonic determinacy with Datalog rewritability.
\end{abstract}

\maketitle

\section{Introduction}
View definitions allow complex queries to be represented by simple relation
symbols. They have many uses, including  as a means to protect access to data,
as a means to raise the level of abstraction available
to data users,   
and as a means to speed up
the evaluation of queries \cite{afrati2019answering}.
Views represent a restricted interface to a dataset, and thus an associated
question is what class of queries can be answered via accessing this interface.
More formally, given a query $Q$ expressed as a logical formula over the base
relations, can the answer to $Q$ be obtained via accessing the views.
There are several different formulations of this computational problem, 
depending on what one means by ``answering a query accessing the views''.
One can ask whether $Q$ is expressible as an arbitrary function of the views,
or
as an arbitrary monotone function of the views. Alternatively,
one can choose a particular
query language $L$ and ask whether $Q$ can be transformed to a query $Q'$ 
over the views, where $Q'$ is in $L$.
The first choice is that $Q$ is \emph{determined over the views},
the second that $Q$ is \emph{monotonically determined over the views},
and the last that $Q$ is \emph{$L$-rewritable over the views}.
Each of these notions can be relativized to finite instances.

These questions were studied initially in the case where
both queries and views are given by conjunctive queries (CQs).
It is known that:

\begin{compactitem}
\item  determinacy of CQ query over a collection of CQ views 
is equivalent to rewritability
of $Q$ over the views in relational algebra
 \cite{NSV}
\item determinacy of a CQ over CQ views does not agree neither with determinacy over finite instances \cite{redspider}
nor with monotone determinacy \cite{afratideterminacy}
\item determinacy of a CQ query over CQ views is undecidable \cite{redspider}, and
the same holds for determinacy over finite instances \cite{rainworm}
\item determinacy is decidable for queries and views given as \emph{path-CQs}  \cite{afratideterminacy}
\item monotonic determinacy of a CQ query over CQ views implies rewritability
of $Q$ as a CQ  \cite{thebook}, 
agrees with monotonic determinacy over finite instances
and is  NP-complete to decide \cite{lmss}
\end{compactitem}

These results have been generalized
to the case of queries and views  built up with more general constructs
of active-domain first-order logic (or equivalently, in relational algebra).
Then monotonic determinacy becomes, like determinacy,  undecidable,
 and monotonic determinacy, like determinacy, disagrees
with its variant over finite instances. But there is still
a relationship between determinacy/monotonic determinacy and
rewritability in a logic: 
determinacy is the same as  rewritability in first-order logic;  monotonic
determinacy is the same as rewritability as a UCQ \cite{NSV,thebook}.

Less is known where queries and views are \emph{recursive},
for example,
when views and queries are in the common recursive query language
Datalog. 
For specialized recursive queries and views over a graph schema, 
the \emph{regular path queries},
both the determinacy and monotonic determinacy problem have been studied.
For \emph{one- and two-way regular path queries and views}
monotonic determinacy (aka ``\emph{losslessness with respect to the sound view assumption}'') is decidable in \expspace 
(\cite{losslessregular} for 1-way,\cite{calvanese2007view} for 2-way),
and implies Datalog rewritability \cite{determinacyregularpath}, while
 plain determinacy is undecidable \cite{redchains}.
It follows from \cite{inverserules} that monotonic determinacy is undecidable
for Datalog queries and CQ views and implies rewritability in Datalog over views.

The status of these questions for more general recursive queries --- e.g., 
 queries and views in Datalog  over higher-arity relations --- is to the best of our knowledge unknown.

\begin{example} \label{ex:run}
Consider a schema with a ternary relation $T$, and binary relation $B$ and unary relations $U_1, U_2$.
Consider the Boolean Datalog query  $Q$ given as:
\begin{align*}
	\goal_{Q} \datalogarrow  ~ U_1(x) \datalogwedge W_1(x) \\
W_1(x) \datalogarrow  T(x,y, z) \datalogwedge B(z,w) \datalogwedge B(y,w) \datalogwedge W_1(w) \\
W_1(x) \datalogarrow U_2(x)
\end{align*}

Consider the following CQ views:
\begin{align*}
V_0(x,w) :=  T(x,y, z) \datalogwedge B(z,w) \datalogwedge B(y,w) \\
 V_1(x) := U_1(x) \quad\quad
V_2(x) := U_2(x) \\
 V_3(y,z) := U_1(x) \datalogwedge T(x,y,z) 
\end{align*}
and the binary Datalog view $V_4$:

\begin{align*}
\goal_{V_4}(y,z) \datalogarrow  T(x,y,z) \datalogwedge B(z,w) \datalogwedge \\
B(y, w) \datalogwedge T(w,q,r) \datalogwedge 
\goal_{V_4}(q,r) \\
\goal_{V_4}(y,z) \datalogarrow B(y,w) \datalogwedge B(z,w) \datalogwedge U_2(w)
\end{align*}

We can see that $Q$ is monotonically determined over the views $V_0$-$V_2$.
In fact there is a Datalog rewriting,  obtained
from $Q$ by first replacing the second rule by $W_1(x) \datalogarrow V_0(x,w) \datalogwedge W_1(w)$
and then replacing each $U_i$ by  $V_i$ in the other rules.
Further $Q$ is monotonically determined using views $V_3$ - $V_4$, since it can be rewritten as the CQ
$\exists y ~ z ~ V_3(y,z) \wedge  V_4(y,z)$.

Note that query $Q$ is not contained in any of the classes considered in past work (e.g. regular path queries).
\end{example}

\myparagraph{Our results} 
We give results on the complexity of deciding
monotonic determinacy and on the ability to rewrite monotonically-determined queries
into suitable languages, for views and queries expressed in Datalog or in sublanguages
such as Monadic Datalog (MDL), or frontier-guarded Datalog ($\fgdatalog$).

We provide new  positive results about rewritability, showing monotonic determinacy implies $L$-rewritability for some natural query languages $L$.
We show that monotonic determinacy implies rewritability in Datalog for Datalog queries and
$\fgdatalog$ views
(Theorem~\ref{thm:frgd-rewriting}), as well as
for MDL queries and a collection of $\fgdatalog$ and CQ views (Theorem~\ref{thm:moncq-rewriting}).
We observe that for CQ $Q$ and Datalog $\views$, monotonic determinacy implies rewritability as
a  CQ, and the same holds if CQ is replaced with UCQ.
Note that an analysis of the ``inverse rules'' algorithm \cite{inverserules}
implies that $\fgdatalog$ 
queries monotonically determined over
CQ views have $\fgdatalog$ 
rewritings.
On the negative side, we show that MDL queries monotonically determined  over CQ views
are not necessarily rewritable in MDL (Theorem ~ \ref{thm:no-monadic}). 
This contrasts with the observation from \cite{inverserules} mentioned above.
In contrast to Theorem ~\ref{thm:moncq-rewriting}, we give an example of an MDL query  monotonically determined over UCQ views
without a Datalog rewriting (Theorem~\ref{thm:no-datalog-mdl-ucq}).  Our results on rewritability are summarized in Figure \ref{fig:table:rewritability}
where ``nn'' stands for ``not necessarily''.

\begin{table*}
\vspace{10pt}
  \centering
  {
    \centering
\begin{tabular}{|c|c|c|c|c|c|}
\hline
Query\,$\backslash$\,Views& CQ   & MDL, FGDL & FGDL + CQ & UCQ & Datalog \\
\hline
CQ &\multicolumn{5}{c|}{CQ [Prop.\,\ref{prop:cqquery}, (a)]}   \\
\hline
UCQ &\multicolumn{5}{c|}{UCQ [Prop.\,\ref{prop:cqquery}, (b)]}   \\
\hline
MDL& FGDL, nn MDL           &  MDL                & Datalog, nn MDL   &\multicolumn{2}{c|}{}\\
   & \cite{inverserules} and [Th.\,\ref{thm:no-monadic}]
                                &  [Th.\,\ref{thm:frgd-rewriting}]        & [Th.\,\ref{thm:moncq-rewriting}] and [Th.\,\ref{thm:no-monadic}] & \multicolumn{2}{c|}{ not necessarily  } \\
\cline{1-4}
FGDL& FGDL  \cite{inverserules}             &  Datalog
                                                  & rewritability in        &\multicolumn{2}{c|}{Datalog [Th.\,\ref{thm:no-datalog-mdl-ucq}]}  \\
\cline{1-2}
Datalog   & Datalog  \cite{inverserules}  &[Th.\,\ref{thm:frgd-rewriting}] & Datalog is open            &\multicolumn{2}{c|}{} \\ 
          &                                       &                  &                            &\multicolumn{2}{c|}{} \\

\hline
\end{tabular}
\caption{Rewritability of Queries Monotonically Determined by the Views } \label{fig:table:rewritability}
}
\end{table*}

We now turn to results about deciding monotonic determinacy.
We show that monotonic determinacy is
\begin{compactitem}
\item[--] decidable in $\twoexp$ for CQ queries and Datalog views (Theorem~\ref{decidability:cq:equivalence}),
\item[--] decidable in $\twoexp$ for queries and views in frontier-guarded Datalog (Theorem~\ref{thm:decidemondetrewriting:frontierguarded}),
\item[--] decidable in $\threeexp$ for MDL queries and a collection of MDL and CQ views (Theorem~\ref{thm:decidemondetrewritingmdlandcqs}), 
\item[--] $\twoexp$-hard for CQ queries and MDL views and for MDL queries and CQ views (Proposition~\ref{prop:easyhardness}) 
\item[--] undecidable for MDL queries and UCQ views (Theorem~\ref{thm:undec})
\end{compactitem}

Known and new results  on decidability of monotonic
determinacy are presented in Figure~\ref{fig:table:decidability} where we use [upper bound]/[lower bound] notation for sources.
%

\begin{table*}
  \centering
  {
    \centering
  \begin{tabular}{|c|c|c|c|c|c|}
\hline
Query $\backslash$ Views & CQ   & MDL, FGDL & FGDL + CQ & UCQ & Datalog \\
\hline
CQ & NP-c       & \multicolumn{2}{c|}{$\twoexp$-c}                    & $\Pi_2^p$-c          & $\twoexp$-c \\
UCQ& \cite{lmss}& \multicolumn{2}{c|}{[Th.\,\ref{decidability:cq:equivalence}]/[Prop.\,\ref{prop:easyhardness}]}
                                                                     & \cite{lutz2018query} & [Th.\,\ref{decidability:cq:equivalence}]/[Prop.\,\ref{prop:easyhardness}]\\
\hline
   & in $\threeexp$ [Th.\,\ref{thm:decidemondetrewritingmdlandcqs}]
                                &                 & in $\threeexp$ [Th.\,\ref{thm:decidemondetrewritingmdlandcqs}] &\multicolumn{2}{c|}{}\\
MDL& $\twoexp$-hard              &                 & $\twoexp$-hard      &\multicolumn{2}{c|}{}\\
   & [Cor.\,\ref{prop:easyhardness}]        & $\twoexp$-c      & [Prop.\,\ref{prop:easyhardness}]&\multicolumn{2}{c|}{Undecidable [Th.\,\ref{thm:undec}]} \\
\cline{1-2}
\cline{4-4}
FGDL& decidability             & [Th.\,\ref{thm:decidemondetrewriting:frontierguarded}]/[Prop.\,\ref{prop:easyhardness}]    
                                                  & decidability       &\multicolumn{2}{c|}{}\\
   & is open                  &                   & is open            &\multicolumn{2}{c|}{}\\
\hline
Datalog &\multicolumn{5}{c|}{ undecidable for a fixed atomic view [Prop.\,\ref{prop:easyhardness}], see also \cite{inverserules}, Th.\,3.1 }\\
\hline
\end{tabular}
\caption{Decidability and Complexity of Monotonic Determinacy } \label{fig:table:decidability}
}
\end{table*}

Alongside with $L$-rewritability we can ask whether there are 
computable functions lying within a certain complexity
class which separate the images of instances where $Q$
is true  from images of those where $Q$ is false.
We call such a function  a \emph{separator} for $Q$ over $\views$.  Note that
Datalog rewritings give rise to $\ptime$ separators, while UCQ-rewritings
produce $\aczero$ separators.
Our additional observations on separators, outside of those that follow from rewritability results,
are:
\begin{inparaenum}
\item for Datalog queries and UCQ views there is always a separator in $\np$ as well as one
in $\conp$;
\item for any primitive recursive function $f$ there are
Datalog queries monotonically determined over Datalog views without a separator in $TIME(f(x))$ (Theorem~\ref{thm:nocomputable}).
\end{inparaenum}

\myparagraph{Techniques} A contribution of the paper is to show how techniques arising from earlier work
can be adapted for the analysis of monotone determinacy.
For our positive results, a key tool is an automata-theoretic technique, involving  bounds
on the treewidth of view images and  the \emph{forward-backward method} developed for analysis of
 guarded
logics  \cite{gho,forbackj}. For our negative results, we show how to adapt some of the coding ideas
used in showing undecidability of determinacy  \cite{redchains,rainworm,redspider} to the setting
of monotonic determinacy, and we also show how tools from constraint satisfaction 
\cite{atserias07:power}
can be used to provide monotonically-determined queries that have no Datalog rewriting.


\myparagraph{Organization}
Section \ref{sec:prelims} contains  preliminaries
about Datalog and monotonic determinacy, while Section \ref{sec:inf} presents key tools
that we  make use of in our positive results.
Section \ref{sec:rewrite} presents our rewritability results, while
Section \ref{sec:decide} gives results on deciding monotonic determinacy.
Section \ref{sec:lower} contains lower bounds on detecting monotonic determinacy, while
Section \ref{sec:nonrewrite} provides non-rewritability results.
The paper ends with  conclusions and some open questions in Section \ref{sec:conc}.
The details of many proofs are deferred to the appendix.

\section{Preliminaries} \label{sec:prelims}
We will work with relational schemas, consisting of a finite set of relations,
with each relation $R$ associated with a number 
the \emph{arity of $R$}.
For $R$ of arity $n$, an $R$-fact is an expression $R(c_1 \ldots c_n)$, where $c_1 \ldots c_n$
are elements. A fact over schema $\schema$ is an $R$ fact for some relation $R$ of $\schema$.
A \emph{database instance} (or simply \emph{instance} when it is clear that
we are discussing data) for a schema is a set of facts over the schema.
The \emph{active domain} of an instance $\inst$, denoted $\adom(\inst)$,  is the set of elements that occur
as $c_i$ in some fact $R(c_1 \ldots c_n)$ of $\inst$.
A \emph{query} of arity $n$ over schema $\aschema$
is a function from instances of $\aschema$ to relations of arity $n$.
A \emph{Boolean query} is a query of arity $0$.
The output of a query $Q$ on instance $\inst$ is  denoted as $\doutput{Q}{\inst}$.
We will also write $\inst \models Q(\vec c)$ or $\inst, \vec c \models Q$
to indicate that $\vec c$ is in the output
of $Q$ on input $\inst$.
A \emph{homomorphism} from  instance  $\inst$ to instance $\inst'$ is a mapping
$h$ such that $R(c_1 \ldots c_n) \in \inst$ implies $R(h(c_1) \ldots h(c_n)) \in \inst'$.
If there is a homomorphism from $\inst$ to $\inst'$ then we write $\inst\to \inst'$. 

The \emph{Gaifman graph} of an instance $\inst$ is the graph whose nodes are
the elements of $\adom(\inst)$ and whose edges connect any $c_i$ and $c_j$
in a $\vec c$ such that $R(\vec c)$ holds. The \emph{radius} of a graph
$G$ is defined as $\min_{u\in\vertices(G)}\max_{v\in\vertices(G)} dist_G(u,v)$
where $dist_G(u,v)$ is the distance between $u$ and $v$ in $G$.

\myparagraph{Conjunctive queries and Datalog}
A \emph{conjunctive query} (CQ) is a logical
formula of the form $\q(\vec x) = \exists \vec{y}\, \phi(\vec x,  \vec y)$, where
$\phi(\vec x,  \vec y)$ is a conjunction of atoms.
Given any CQ $Q$, its \emph{canonical database}, denoted $\canondb(Q)$, is the instance formed
by turning each atom $R(x_1 \ldots x_n)$ into a fact $R(c_{x_1} \ldots c_{x_n})$, where for each variable or
constant $x$ in $Q$ we
have a constant $c_x$.
Each CQ $Q$ with free variables ordered as
$x_1 \ldots x_n$ defines a query of arity $n$ in the obvious way:
a tuple $t_1 \ldots t_n$ is in the output of $Q$
on $\inst$ if there is a homomorphism of $\canondb(Q)$
into $\inst$  mapping each $x_i$ to $t_i$. The radius of a CQ is the radius of 
the Gaifman graph of its canonical database.

Datalog is a language for defining  queries over a relational schema $\aschema$.
Datalog rules are of the form:
\begin{align*}
P(\vec x)    \datalogarrow    \phi(\vec x)
\end{align*}
where $P(\vec x)$ is an atom  over a relation $P$ that is not in $\aschema$, $\phi$ is a conjunctive query
and every variable in $P(\vec x)$ 
occurs in $\phi$. The left side of the rule is the \emph{head}, while the right side is the \emph{body}
of the rule.
In a set of rules,
the relation symbols that occur in the head of a rule are the \emph{intensional
database predicates} (IDBs). The relations in $\aschema$ are called the \emph{extensional
relations} of the rule.  
A \emph{Datalog program} is a finite collection of rules.
For a database instance $\inst$ and a set of Datalog rules $\Pi$ by $\fpeval{\Pi}{\inst}$ we denote the minimal IDB-extension 
of $\inst$ satisfying $\Pi$.  
%
A \emph{Datalog query} $Q = (\Pi, \goal)$ is a Datalog program $\Pi$ together
with a distinguished intensional \emph{goal relation} $\goal$ of arity $k \ge 0$. 
The output of  Datalog query $Q$ on an instance $\inst$ (denoted as $\doutput{Q}{\inst}$ or simply $Q(\inst)$) consists of all tuples $\vec c$ such that
$\goal(\vec c) \in \fpeval{\Pi}{\inst}$. 

For example, consider a signature where there is a binary relation $R$
and unary relation $U$. The formula expressing that $x$ has a path consisting
of $R$ edges to an element in
$U$ would be written in Datalog as the following query $\connquery(x)  = (\Pi, \goal(x))$ where $\Pi$ consists of the following rules:
$$
\begin{array}{rcl}
P(x) &\datalogarrow &U(x)\\
P(x) &\datalogarrow& R(x,y) \datalogwedge P(x) \\
\goal(x) &\datalogarrow & P(x)
\end{array}
$$
Above, 
$P(x)$ and $\goal(x)$  are intensional relations while $R(x,y)$ and $U(x)$ are extensional.
We follow conventions concerning Datalog rules and omit the existential quantifiers
on the variables in the body that do not appear in the head; we also
use ``$\datalogwedge$'' for conjunction.

A Datalog query $Q_1$ is \emph{contained in} a Datalog query
$Q_2$ if $\doutput{Q_1}{\inst} \subseteq \doutput{Q_2}{\inst}$ for every instance $\inst$. 
Datalog containment is known
to be undecidable in general \cite{undeciddatalog}.

\myparagraph{Fragments of Datalog}
\emph{Monadic Datalog} (MDL) is the fragment of Datalog
where all intensional predicates are unary.
\emph{Frontier-guarded Datalog} ($\fgdatalog$) requires that in each rule
all the variables in the head co-occur in a single extensional atom of the body. 
Frontier-guarded Datalog does not contain MDL; for example, in an MDL program
we can have a rule $I_1(x) \datalogarrow I_2(x)$, where $I_1$ and $I_2$
are both intensional. However every MDL program can be rewritten
to be  in $\fgdatalog$, and thus
we declare, as a convention, that any MDL program is Frontier-guarded.
Frontier-Guarded Datalog containment is known to be decidable (e.g. \cite{gnfj}).

\myparagraph{Conjunctive queries and approximating Datalog}
A Datalog query $Q = (\Pi, \goal)$ can be approximated by CQs.
We define collections of CQs $\CQAppr(\Pi, U(\vec x), i)$ with free variables $\vec{x}$ for all atoms $U(\vec x)$ that occur in the head of a rule in 
$\datalogprog$ by induction on $i$. For the base case,
$\CQAppr(\Pi, U(\vec{x}), 1) $ consists of all CQs obtained by taking the body of a rule with the head $U(\vec x)$ in $\Pi$ which contains no intensional predicate.

For the inductive step, $\CQAppr(\Pi, U(\vec{x}), i + 1 )$ consists of all CQs obtained by taking any body of a rule
whose head is $U(\vec{x})$ and replacing all intentional atoms $V(\vec y)$ with 
$\q(\sigma(\vec{z}))$, where
$\q(\vec{z})$ is in $\CQAppr(\datalogprog,V(\vec z),k)$ for $k \le i$ and $\sigma$ unifies $V(\vec z)$ with $V(\vec y)$ by sending $\vec{z}$ to $\vec{y}$.

A \emph{CQ approximation} of a Datalog query $(\Pi, \goal(\vec{x}))$ is any element of $\CQAppr(\Pi, \goal(\vec{x}), i)$ for some $i$.

\begin{proposition} For any Datalog query $Q$, if  $\inst, \vec c \models Q$ then there is 
a CQ approximation  $Q_0$ of $Q$
such that $\inst \models Q_0(\vec c)$.
\end{proposition}

We  often identify an approximation $Q_0$ of a Datalog query $Q$ with
its canonical database; for example, for another Datalog query $Q'$, we can write $\doutput{Q'}{Q_0}$ to indicate
the output of $Q'$ on $\canondb(Q_0)$. 
We can also talk about the approximation of an atom $A$ in a Datalog program, which is defined
by considering the program with $A$ as the goal predicate.

\myparagraph{Views, determinacy, and rewritability} 
A \emph{view} over some relational schema $\aschema$ 
is a tuple $(V, Q_V)$ where $V$ is a view relation and $Q_V$ is an associated 
query over $\aschema$ whose arity  matches that of $V$. $Q_V$ is referred to as the \emph{definition} of view $V$.
By $\views$ we denote a collection of views over a schema
 $\aschema$. We sometimes refer to the vocabulary of the definitions $Q_V$ as the \emph{base schema} for $\views$, denoting
it as $\Sigma_{\baseschema}$, while the predicates  components $V$ are referred to as the \emph{view schema},
denoted $\Sigma_{\viewschema}$.
%
%
%
For an  instance $\inst$ and set of views $\views= \{(V, Q_V) \mid V \in \Sigma_\views \}$, the \emph{view image} of $\inst$, 
denoted by $\views(\inst)$, 
is the instance where each view predicate $V \in \Sigma_\views$ is  interpreted by $\doutput{Q_V}{\inst}$.
A query $Q$ over schema $\aschema$ is \emph{determined over $\views$}  if
\begin{quote}
for any two instances $\inst_1, \inst_2$ such that $\views(\inst_1) = \views(\inst_2)$ we have $\outcome{Q}{\inst_1} = \outcome{Q}{\inst_2}$.
\end{quote}
A query $Q$ over schema $\aschema$ is \emph{monotonically determined over $\views$}  if
\begin{quote}
for any two instances $\inst_1, \inst_2$ such that $\views(\inst_1) \subseteq 
\views(\inst_2)$ we have $\outcome{Q}{\inst_1} \subseteq \outcome{Q}{\inst_2}$.
\end{quote}

\myeat{
Monotone determinacy goes by a number of other names in the database literature,
including ``losslessness under the sound view assumption'' \cite{losslessregular} and 
``strong determinacy'' \cite{perezthesis}.
}

Given views $\views$ and a query $Q$, a query $\rewriting$ over the view schema $\Sigma_\views$ is a 
\emph{separator}
of $Q$ with respect to $\views$ if:
for each $\inst$ over $\aschema$, the output of $\rewriting$ on $\views(\inst)$ 
is the same as the output of $Q$ on $\inst$.
A separator  that can be specified in a particular
language $L$ (e.g. Datalog, CQs) is an \emph{$L$}-rewriting of $Q$ w.r.t. $\views$,
and if this exists we say $Q$ is \emph{$L$-rewritable} over $\views$.

It is clear that if $Q$ has a rewriting in a  language that defines only monotone queries,
like Datalog, then $Q$ must be monotonically determined. We will be concerned with the converse
to this question.
The main questions we will consider, fixing languages  $L_Q$ and $L_\views$  for the queries and views (e.g.
Datalog, fragments of Datalog) are:
\begin{compactitem}
\item can we decide whether a $Q$ in $L_Q$ is monotonically determined over $\views$?
\item fixing another language $L$ for rewritings,
if $Q$ is monotonically determined over $\views$, does it necessarily have a rewriting
in $L$?
\end{compactitem}

\myeat{
For succinctness of presentation, in all Lemmas, Propositions and Proofs (but not Theorems and Corollaries) we assume that the user query $Q$ is Boolean.
The statements and proofs can be easily extended to the general case by replacing instances with \emph{pointed instances} which
are instances together with bindings of answer variables of $Q$ and a reasonable modification 
of the homomorphism definition and coding procedure.
}

In this paper, \emph{for simplicity  we will always consider the determinacy and rewritability 
problems restricting to the case when the query $Q$ is Boolean.} But all of our 
results extend to the non-Boolean case. In addition, we allow our
instances to be finite or infinite, but \emph{all of the results extend when
the instances are assumed to be finite}. Ssee the appendix for details.

\section{Forward and backward between Datalog and automata} \label{sec:inf}
We overview an automata-theoretic technique that will prove useful in 
rewriting results. It involves
\emph{treewidth bounds}, along with the idea of combining
\emph{forward mappings} from Datalog to automata, projection of an automata onto a subvocabulary, and \emph{backward mappings}
from an automaton to Datalog. The approach derives
from work on guarded logics \cite{forbackj,gho}.

\myparagraph{Treewidth and tree codes}
For a number $k$
a \emph{tree decomposition of width $k$} for an
instance $\inst$ is a pair $\TD = (\tau, \lambda )$ consisting
of a rooted directed tree  $\tau = (V,E)$  and a 
map $\lambda$ associating a tuple
of distinct elements $\lambda(v)$ of length at most $k$ 
(called a \emph{bag}) to each vertex $v$ in $V$ 
such that the following conditions hold:
\begin{compactitem}
	\item[--] for any atom $R(\vec{c})$ in $\inst$, there is a vertex $v\in V$ with 
$\vec{c} \subseteq \lambda(v)$;
	\item[--] for any element $c$ in $\inst$, the set  $\{\, v\in V \mid c \in \lambda(v) \,\} $ is connected in~$\tau$.
\end{compactitem}

Above we abuse notation slightly by using $\lambda(v)$ also to refer to the underlying set of elements as
well as the tuple. Also in the literature the width associated to such a decomposition is $k-1$, but this distinction
will not be important for any of our results.
Will also talk about a tree decomposition of width $k$ for
a pair  $(\inst, \vec{a})$ consisting of an instance and a tuple. In this case
we add to the requirements above that
$\vec{a}$ is an initial segment of  $\lambda(r)$ for $r$ the root of the tree.

The \emph{treewidth} of an instance $\inst$, $\tw(\inst)$, is the minimum width
of a tree decomposition of $\inst$.
For a tree decomposition $\TD$ of data instance $\inst$ 
let $\tspan(\TD)$ be the maximum over elements $e$ of $\inst$ of the number of bags containing $e$.

We will now discuss how to represent tree decompositions by labeled
trees called codes. In this context, we will always  assume that in tree decompositions, \emph{all  vertices $v \in V$ have
outdegree at most $2$.} It is easy to show that if an instance has any tree decomposition
of width $k$, it has one with this property.

We represent such tree decompositions as instances
in a signature $\code(\aschema, k)$ which contains the following relations:
\begin{compactitem}
\item for every relation $R \in \aschema$ of arity $m$ and every sequence $\vec{n} = n_1, \dots, n_m$ of 
numbers of size at most $k$
there is a unary relation $\coderel^R_{\vec{n}}$ in $\code(\aschema, k)$ to mark the nodes $v$ in $\tau$ such that the atom 
$R(b_{n_1}, \dots, b_{n_m})$  is in $\inst$, 
where $\lambda(v)=(b_1, \dots, b_k)$.
\item for every partial 1-1 map $s$ from $\{1, \dots, k\}$ to $\{1, \dots, k\}$, there is a binary relation
$\coderel_s$ to indicate the ``same as'' relation between positions in neighboring bags.
For example, if $(u,v) \in \coderel_s$ and $s(3) = 1$, then the position $3$ in $u$ and the position $1$ in $v$ stand for 
the same element. All relations $\coderel_s$ are directed from a parent to a child.
\end{compactitem}
We use $\unpred$ and $\binpred$ to denote the sets of all unary and binary predicates in $\code(\aschema, k)$ respectively.
A tree over this signature will be referred to as a \emph{tree code of width $k$ for $\aschema$}.

It should be clear how each tree decomposition of $(\inst,\vec{a})$ of width
$k$ gives rise to a tree code of width $k$ for $\aschema$; 
if there are bags with less than $k$ elements, we fill them up with dummy elements to the length $k$. 
 We now show
how to \emph{decode} an instance from such a code $\T$. 
For nodes $u,v$ in a code $\T$, we write $(u,i) \equiv_0 (v,j)$ if 
 $(u,v) \in \coderel_s$ holds in $\T$ and $s(i)=j$.
For a node $u$ and position $i$ we let  $[u,i]$ be the equivalence class of $(u,i)$ in the equivalence
relation generated by $\equiv_0$.
In words, the position $i$ in the node $u$ corresponds to 
the position $j$ in the node $v$ if there is an undirected path leading from $u$ to $v$ with the edge labels that 
in a step-by-step manner establish a match between $i$ in $u$ and $j$ in $v$. 
The \emph{decoding} of $\T$, denoted $\inst = \decode\T$,
is the $\aschema$ database instance $\inst$ consisting 
of atoms  $R([v_1,i_1], \dots, [v_r, i_r])$ where each $R$ from $\aschema$ is applied to
\myeat{the universe is the set $\{[v,i] : v \in \dom(\T), i \in \{1, \dots, k\}\}$ and } 
exactly those tuples $([v_1,i_1], \dots, [v_r, i_r])$ for which there is some  node $w \in \dom(\T)$ such that $w \in \coderel^R_{j_1\dots j_r}$ and $[w,j_m] = [v_m,i_m]$
for all $m \in \{1, \dots, r\}$. In this case we also say the $\T$ \emph{is a code of} $\inst$.

\myparagraph{Monadic Datalog Normalisation}
A Monadic Datalog query is said to be \emph{normalized} if the body of any recursive rule
does not contain IDB atoms with the head variable.  A well-known and simple fact
is that any MDL query can be transformed into a normalized one.

\begin{proposition}[\cite{chaudhuri1997equivalence}] \label{prop:normalise} For each MDL query $Q$ 
there exists a normalized MDL query  $Q'$
which is equivalent to $Q$.
\end{proposition}

\myeat{
\begin{example}
The MDL program \nb{remove example?} 
$$
\begin{array}{rcl}
\goal(x) & \datalogarrow & P(x) \datalogwedge D(x)\\
P(x) & \datalogarrow & Q(x) \datalogwedge A(x,w) \\
Q(x) & \datalogarrow & P(x) \datalogwedge B(x,y) \\
Q(x) & \datalogarrow & C(x,y) \datalogwedge P(y) \\
P(x) & \datalogarrow & T(x) \\
\end{array}
$$
is first transformed into
$$
\begin{array}{rcl}
\goal(x) & \datalogarrow & P(x) \datalogwedge D(x) \\
P(x) & \datalogarrow & T(x) \datalogwedge B(x,y) \datalogwedge A(x,w) \\
P(x) & \datalogarrow & C(x,y) \datalogwedge P(y) \datalogwedge A(x,w) \\
P(x) & \datalogarrow & T(x) \\
\end{array}
$$
and then into
$$
\begin{array}{rcl}
\goal'(x) & \datalogarrow & D(x) \datalogwedge T(x) \datalogwedge B(x,y) \datalogwedge A(x,w) \\
\goal'(x) & \datalogarrow & D(x) \datalogwedge C(x,y) \datalogwedge P(y) \datalogwedge A(x,w) \\
\goal'(x) & \datalogarrow & D(x) \datalogwedge T(x)  \\
P(x) & \datalogarrow & T(x) \datalogwedge B(x,y) \datalogwedge A(x,w) \\
P(x) & \datalogarrow & C(x,y) \datalogwedge P(y) \datalogwedge A(x,w) \\
P(x) & \datalogarrow & T(x). \\
\end{array}
$$

\end{example}
}

Normalization is useful in connection with tree codes, since it is easy to see that the CQ approximations
of normalized queries have decompositions with small ``treespan'':
\begin{lemma}\label{lem:normalised:treedecomp}
Let $Q$ be a normalized Monadic Datalog query. Then there is a number $k = O(|Q|)$ such that all CQ-approximations
of $Q$ have tree decomposition $\TD$ of width $k$ with $\tspan(\TD) \le 2$.
\end{lemma}

\myparagraph{Bounding the treewidth of view images}
We will present  results showing that, for certain classes
of sets of views $\views$ and Datalog queries $Q$,  we can find
a uniform bound on the treewidth of the $\views$-image of
the approximations of $Q$.

It is easy to see that expanding an instance with the evaluation of
all intensional predicates of a frontier-guarded program does not blow-up treewidth:

\begin{lemma}\label{lem:boundingtw:1}
If ~ $\Pi \in \fgdatalog$ and $\inst$ is an instance of treewidth $k$, then 
$\fpeval{\Pi}{\inst}$ is of treewidth $k$.
\end{lemma}
A locality argument shows that applying connected CQ views preserves bounded treewidth:

\begin{lemma}\label{lem:boundingtw:2}
Let $\TD$ be a tree decomposition of a data instance $\inst$ of width $k$ with $\tspan( \TD) \le 2$.
Let $\views$ be a set of connected CQ views, and $\views(\inst)$ the view image of $\inst$ under $\views$.
Let $r$ be the greatest radius of a CQ in $\views$.
Then the treewidth of $\views(\inst)$ is at most $k' = \frac{k(k^{r+1}-1)}{k-1}$.
\end{lemma}

\myparagraph{Tree automata}
We describe our variant of tree automata that accept binary  trees $\T$ with 
edges labelled by binary relations from the set
$\binpred$ and nodes labelled with unary predicates from a set $\unpred$.
We consolidate node- and edge-labels
by considering a tree alphabet $\TreeAlphabet$.
$\TreeAlphabet$ contains  labels for internal nodes 
$\sigma_L^{s_1, s_2}$ indexed by  \emph{sets} 
of unary predicates $L \subseteq \unpred$ and pairs  of binary predicates $s_1, s_2 \in \binpred$. It also
contains leaf labels
$\sigma_L$ indexed by $L \subseteq \unpred$. 
We sometimes treat trees as terms over this alphabet: a tree with root
labeled $\sigma_L^{s_1, s_2}$ with children
$t_1$ and $t_2$ would be written as $\sigma_L^{s_1, s_2}(t_1, t_2)$.

A \emph{nondeterministic finite tree automaton} (NTA) \emph{over} $\TreeAlphabet$ is a tuple $\automaton = (\Q,\Q_f,\Delta_0, \Delta_2)$, where
\begin{compactitem}
\item $\Q$ is a finite set of \emph{states} 
\item $\Q_f \subseteq \Q$ is a set of \emph{final} states, 
\item $\Delta_0$ is a set of \emph{initial transitions} of the form $\sigma_L \to q$, and
\item $\Delta_2$ is a set of \emph{transitions} of the form $q_1,q_{2},\sigma_L^{s_1, s_2} \to q$.
\end{compactitem}
A run of $\automaton$ on a tree $\T$ is a label function $f: \Nodes(\T) \to \Q$ satisfying the following:
if $t_v = \sigma^{s_1, s_2}_L (t_{v_1}, t_{v_2})$ for $\sigma^{s_1, s_2}_L \in \TreeAlphabet$, then $(f(v_1), f(v_2), \sigma^{s_1, s_2}_L \to  f(t)) \in \Delta_2$
{\textbf and } if $f(v) = q$ for a leaf $v$ of $\T$ with  $t_v = \sigma_L$, then $\sigma_L \to q \in \Delta_0$. 
We say that $\T$ is \emph{accepted} by $\automaton$ if there is a run of $\automaton$ on $\T$ that labels the root of $\T$ with a final state.

\myparagraph{Forward from Datalog to NTA}
We now show how to create a tree automaton accepting the view images of approximations of a given
Datalog query.
We say that a class $\C$ of instances  is $k$-\emph{regular}  if  the treewidth of instances in $\C$ is at most $k$, and
there is an automaton $\automaton$ such that 
\begin{compactitem}
\item for codes $\T$ of width $k$, $\automaton$ accepts $\T$ implies $\decode\T \in \C$.
\item for each instance $\F \in \C$ there is a code $\T$ such that $\decode\T = \F$ and $\automaton$ accepts $\T$.
\end{compactitem}

In this case we say that $\automaton$ \emph{captures} $\C$. If the stronger condition ``for all codes $\T$,
$\automaton$ accepts $\T$ iff $\decode\T \in \C$'' holds, we say that $\automaton$ \emph{recognizes} $\C$. 

The following simple ``forward mapping'' proposition shows that we can capture the approximations of Datalog queries with an automaton:

\begin{proposition}\label{prop:cqappr:regular} 
 For any Datalog query $Q = (\Pi, \goal)$, there is an {\exptime} function that
outputs an NTA $\automaton_Q$ that captures the  set of canonical databases of CQ
approximations of $Q$. 
\end{proposition}

If we restrict to instances of a fixed treewidth, we can do better, obtaining
 an NTA that recognizes all trees that satisfy the Datalog program considered
as a set of Horn clauses:

\begin{proposition} \label{prop:datalog} 
For any Datalog program $\datalogprog$, the class $\{\F \mid \F \models \datalogprog, \tw(\F) \le k \}$ (here $\F$ are finite instances
 which contain both EDBs and IDBs of $\Pi$) is $k$-regular and is recognized by an NTA  at most doubly-exponential sized in $k$ and singly-exponential in $|\Pi|$.
\end{proposition}

We also note that if we have captured a class of codes of instances with an automaton, we can project away
some of the signature and still capture:

\begin{proposition} \label{prop:restriction} If $\C$ is a $k$-regular class in $\Sigma$ captured by NTA $\automaton$ 
and $\Sigma' \subseteq \Sigma$, then 
the class $$\C\restr \Sigma' = \{\F\restr \Sigma' \mid \F \in \C\}$$ is also $k$-regular, captured by an automaton of size at most $|\automaton|$. The same holds with ``captured'' replaced by ``recognized''.
\end{proposition}

Our next ``forward mapping'' result shows that we can recognize the 
 set of codes of small treewidth
which fail to satisfy clauses of a frontier-guarded program:

\begin{proposition} \label{prop:datalog:notaccept}
For a $\fgdatalog$ query $Q=(\Pi, \goal)$ the set $\{(\inst, \vec a) \mid \inst \not\models Q(\vec a), \tw(\inst) \le k \} $ is $k$-regular and recognized by
an NTA of size at most doubly-exponential in $k$.
\end{proposition}

\begin{proof}
Follows from Propositions \ref{prop:datalog} and \ref{prop:restriction} since for a frontier-guarded $Q=(\Pi, \goal)$ we have, using Lemma \ref{lem:boundingtw:1},
\begin{multline*}
\{(\inst, \vec a) \mid \inst \not\models Q(\vec a), \ \tw(\inst) \le k\} = \\
= \{(\F\restr\Sigma, \vec a) \mid \F \models \Pi, \ \F \not \models \goal(\vec{a}),\ \tw(\F) \le k\}
\end{multline*}
where $\Sigma$ is the signature of the EBDs in $\Pi$.

\end{proof}
In applying these results, we will sometimes use implicitly that \emph{if 
$\C_1$ is captured by  $\automaton$ and $\C_2$ is recognized by $\automaton'$, then
$\C_1 \cap \C_2$ is captured by the product of $\automaton$ and $\automaton'$}. Note that, in contrast,
classes of instances that are captured are \emph{not} closed under intersection.

\myparagraph{Homomorphic determinacy}
A query $Q$ is said to be \emph{homomorphically determined} by views $\views$ if:

\begin{quote}
Whenever we have two instances $\inst_1$ and  $\inst_2$ and a homomorphism $h$ from $\views(\inst_1)$ to $\views(\inst_2)$, then
for each tuple $(c_1, \dots, c_k) \in Q(\inst_1)$ we also have $(h(c_1), \dots, h(c_k)) \in Q(\inst_2)$. 
\end{quote}

Note that if $Q$ is rewritable over $\views$ in Datalog, or any other homomorphism-invariant
query language, then $Q$ must be homomorphically determined by $\views$.

Homomorphic determinacy of $Q$ over $\views$ always implies monotonic determinacy
of $Q$ over $\views$; monotonic determinacy is simply the case where $h$ is the identity.
Surprisingly, for Datalog queries and views the converse also holds:

\begin{lemma} \label{lem:homdet} For any Datalog query $Q$ and Datalog views $\views$,  if $Q$ is monotonically determined
 over $\views$ then it is homomorphically determined over $\views$.
\end{lemma}

\myparagraph{Backwards from NTAs to Datalog}
Consider arbitrary NTA $\automaton$ that works on tree codes of width $k$. 
From $\automaton$ we construct a Datalog program.
For every transition of the form $q_1, q_2, \sigma_{L}^{s_1,s_2} \to q$ with
$L = \{\coderel^{R^1}_{\vec{n}_1}, \dots, \coderel^{R^m}_{\vec{n}_m}\}$ we create
a rule
\begin{multline*}
P_q(x_1, \dots, x_k) \datalogarrow \\
\bigwedge_{i = 1}^{k} \Adompred(x_i) \land P_{q_1}(x_1^1, \dots, x_k^1)\land P_{q_2}(x_1^2, \dots, x_k^2)   \\ \land \!\!\!
\bigwedge_{i \in \dom(s_1)} \!\!\! x_i = x^1_{s_1(i)}\land \!\!\! \bigwedge_{i \in \dom(s_2)} \!\!\! x_i = x^2_{s_2(i)} \land  \bigwedge_{l = 1}^m R^l(\vec{x}_{\vec{n}_l})
\end{multline*}
where $j$ ranges over  $1$ and $2$,  $x_i^j$ are fresh variables for indices $j \in \{1,2\}$ and $i \in \{1, \dots, k\}$, and for $\vec{n}_l = (n_l^1, \dots, n_l^d)$
we have $\vec{x}_{\vec{n}_l} = (x_{n_l^1}, \dots, x_{n_l^d})$.
For initial transitions of the form $\sigma_{L} \to q$ with $L = \{\coderel^{R^1}_{\vec{n}_1}, \dots, \coderel^{R^m}_{\vec{n}_m}\}$ we have rules
\begin{equation*}
P_q(x_1, \dots, x_k) \datalogarrow\bigwedge_{i = 1}^{k} \Adompred(x_i) \land \bigwedge_{l = 1}^m R^l(\vec{x}_{\vec{n}_l}).
\end{equation*}

For accepting states $q$ we add the rules $\goal_{\automaton} \datalogarrow P_q(x_1, \dots, x_k)$
for the goal predicate $\goal_{\automaton}$; recall that we are assuming here that
the original query $Q$ is Boolean, so we are looking for a Boolean Datalog rewriting.
 We also add a standard set of rules which, when evaluated on any data instance $\inst$
under fixed-point semantics, guarantee that the interpretation of the IDB $\Adompred(x)$ is the active domain of $\inst$.
Denote the resulting \emph{backward map Datalog query} by $Q_{\automaton}$.


We now get to the main result of this section, which states that if
we assume homomorphic determinacy and begin with an automaton representing view images
of approximations of $Q$, then applying the backward mapping produces a Datalog rewriting
of $Q$ over $\views$. The proof is mostly a matter of working with the definitions.
Homomorphic determinacy is used in the direction from right to left.

\begin{proposition} \label{prop:backward}
Let $Q$ be homomorphically determined over $\views$ and $\automaton$ be any automaton 
working on $k$-codes such that
$\{\views(Q_i) \mid i \in \omega\} \subseteq \decode{L(\automaton)} \subseteq \{\jnst \mid \views(Q_i)\text{\textrm{} maps into } \jnst \text{\textrm{} for some } i \in \omega\}$. 
That is to say, we require that
\begin{compactitem}
\item[(1)] for each CQ approximation $Q_i$ of $Q$ there is a code $\T$ such that $\decode\T = \views(Q_i)$ and $\T$ is accepted by $\automaton$ (first inclusion);
\item[(2)] for each  $\T$ accepted by $\automaton$ there is a CQ approximation $Q_i$ of $Q$ and a homomorphism from
$\views(Q_i)$ into $\decode\T$ (second inclusion).
\end{compactitem}

Then
for each data instance $\inst$ we have $\inst \models Q$ iff $\views(\inst)
\models Q_{\automaton}(\vec{a})$ for some $\vec{a} \in \adom(\inst)^k$.
\end{proposition}

\section{Rewritability} \label{sec:rewrite}
We are now ready to present our main results about rewritings of queries
that are monotonically determined over views.
The following result exhibits  how the forward and backward mappings help us
obtain Datalog rewritings.

\begin{theorem} \label{thm:frgd-rewriting} Suppose $Q$ is a Datalog
query and $\views$ is a collection of $\fgdatalog$ views.
If $Q$ is monotonically determined by $\views$, then
$Q$ is rewritable over $\views$ in Datalog. The size of the rewriting 
is at most double-exponential in $|Q|$ and exponential in $|\views|$.
If $Q$ is MDL such a rewriting exists in MDL as well.
\end{theorem}

\begin{proof}
Consider the class $\C$ of canonical databases of CQ approximations of $Q$. By Proposition~\ref{prop:cqappr:regular},
$\C$ is $k$-regular for some $k = O(|Q|)$ and is captured by an NTA $\automaton'$ of at most exponential size in $|Q|$.
By Lemma \ref{lem:boundingtw:1}, the treewidth of the class 
of view images of $\C$ is also bounded by $k$.
We claim that there is an automaton $\automaton$ that captures
$\viewimageclass = \{ \jnst \mid \tw(\jnst) \le k, \views(Q_i) \subseteq \jnst \text{ for some } i \in \omega\}$ of size
at most double-exponential in $|Q|$ and single-exponential in $|\views|$ (``of required size'' below)
and argue that it satisfies the conditions of  Proposition \ref{prop:backward}.

Without loss of generality we assume that the sets of IDBs of programs for different views are disjoint, and that their goal predicates are
identical with the view predicates. Denote by $\Pi_{\views}$ the union of all rules in Datalog queries in $\views$. Note that by definition, for any instance $\inst$,
the restriction of $\fpeval{\Pi_{\views}}{\inst}$ on the view signature is exactly $\views(\inst)$.

By Proposition~\ref{prop:datalog}, there is an NTA $\automaton^{\Pi_\views}$ of required size which
recognizes all codes of $\{\F \mid \F \models \Pi_{\views}, \tw(\F) \le k \}$. Therefore the class
$\mathbb{F} = \{\F \mid \F \restr \Sigma_{\baseschema} \in \C, \F \models \Pi_{\views}, \tw(\F) \le k \}$
is captured by the intersection of $\automaton'$ and $\automaton^{\Pi_\views}$, which is also of required size.
Observe that $\viewimageclass$ is the projection of $\mathbb{F}$ on the signature of view predicates
and so $\viewimageclass$ is captured by some NTA $\automaton$ of required size by Proposition~\ref{prop:restriction}.

By Lemma \ref{lem:homdet}, $Q$ is homomorphically determined
over $\views$.  
Since $\{\views(Q_i) \mid Q_i \text{ is a CQ approximation of } Q\}$  is the projection of  $\{ \fpeval{\Pi_{\views}}{\inst} \mid \inst \in \C \}$
on the signature of view predicates,  we have:
\begin{multline*}
 \{\fpeval{\Pi_{\views}}{\inst} \mid \inst \in \C  \}
\subseteq \\
\{\F \mid \F \restr \Sigma \in \C, \F \models \Pi_{\views}, \tw(\F) \le k \} \subseteq \\
\{\F \mid \F \restr \Sigma \in \C, \F \models \Pi_{\views} \}, 
 \end{multline*}
The inclusions above are preserved when projecting to the signature of view predicates. 
From this we can verify
that the condition of Proposition \ref{prop:backward} holds for $\automaton$.
Now applying  Proposition \ref{prop:backward}, we conclude that
$Q$ is Datalog rewritable, and that the rewriting is of required size.

If $Q$ is MDL the construction can be refined to produce an MDL rewriting;
see the appendix for details.
\end{proof}

We can use the same technique in the setting where the views are combinations of
Monadic Datalog and CQs, while the query is Monadic Datalog, using 
normalization  (Lemma \ref{lem:normalised:treedecomp}) and the bound of
Lemma \ref{lem:boundingtw:2}. Normalization is used to enforce the
bound on $\tspan$ required
 in Lemma \ref{lem:boundingtw:2}. Although Lemma \ref{lem:boundingtw:2} requires connectivity, 
we can show that disconnected views can be replaced by connected ones.

\begin{theorem} \label{thm:moncq-rewriting} Suppose $Q$ is a normalized Monadic Datalog
query and $\views$ is a collection of Monadic Datalog and  CQ views. 
If $Q$ is monotonically determined by $\views$, then
$Q$ is rewritable over $\views$ in Datalog. The size of the rewriting is at most double-exponential in $K = O(|Q|^{|\views|})$.
\end{theorem}

The previous rewriting results involved restricting the views.
We now note that if we restrict the query to be a UCQ, monotonic determinacy
implies not only Datalog rewritability, but even UCQ rewritability, for arbitrary Datalog views:

\begin{proposition} \label{prop:cqquery}
For views $\views$ in arbitrary Datalog we have:\\
(1) if a CQ $Q$ is monotonically determined by $\views$, then there 
is a $CQ$-rewriting of $Q$ in terms of $\views$;\\
(2) if a UCQ $Q$ is monotonically determined by $\views$, then there 
is a UCQ-rewriting of $Q$ in terms of $\views$.\\
In both cases the rewritings are  polynomial size in $|Q|$ and $|\views|$.
\end{proposition}

\begin{proof}
This can be seen as a  ``degenerate'' variant of the forward-backward technique, which
is well-known in the DB and KR literature \cite{NSV,afratideterminacy,lutz2018query}.
Let  $Q$ be the disjunction of $Q_i: i \in S$.
Let $Q'=\bigvee_{i \in S} \views(Q_i)$ denote the query  that
holds on an instance $\inst'$ of the view schema exactly when for some
$i \in S$, there is a homomorphism of $\views(Q_i)$ into  $\inst'$.
Equivalently, this is the query obtained by applying the views to each canonical database
of a disjunct of $Q$, and then interpreting the resulting facts as a query.

We claim that if $Q$ is monotonically determined by $\views$, then $Q'$ is a rewriting of $Q$. 
In particular, if $Q$ is a CQ, then $Q'$ is just a CQ.
We need to show that for each instance $\inst$, $\inst \models Q$ iff $\views(\inst) \models Q'$. 

($\Rightarrow$) If  some  $Q_k$ maps into $\inst$, then $\views(Q_k)$ maps into $\views(\inst)$.

($\Leftarrow$) Suppose  some $\views(Q_k)$ maps into $\views(\inst)$. 
Monotonic determinacy implies homomorphic determinacy by Lemma \ref{lem:homdet}
In the definition of homomorphic determinacy, take $\inst_1 = Q_k$ and $\inst_2 = \inst$. It is easy to check that $\views(\inst_1)$ maps into $\views(\inst_2)$ and $\inst_1 \models Q$.
It follows that $\inst_2 \models Q$, in other words, that $\inst \models Q$.

\end{proof}

\section{Decidability} \label{sec:decide}

We move from rewritability results to decision procedures
for monotonic determinacy.

\myparagraph{Monotonic determinacy testing procedure}
Our decidability results will depend
upon an characterization of monotonic determinacy, which we  review here.
Given a Datalog query $Q$ and Datalog views $\views$, a \emph{canonical test for Monotonic
Determinacy} is a tuple $(Q_i, D')$ that consists of:

\begin{compactitem}
\item A CQ $Q_i$ that is a CQ-approximation of $Q$
\item An instance $D'$ of the input schema formed by taking each
fact $F=V(\vec c)$ in $\views(Q_i)$, choosing a CQ approximation  $Q'$ 
of $Q_V$, and replacing $F$ with fresh elements and facts from $Q'$
that witness $V(\vec c)$. That is, firing the rule
$\forall \vec x~ V(\vec x) \rightarrow Q'(\vec x)$. In this case we say that
\emph{$D'$ is obtained from $\views(Q_i)$ by applying inverses of view definitions.}
\end{compactitem}

Such a test \emph{succeeds} if $D'$ satisfies $Q$. It is easy to see that
monotonic determinacy is characterized using tests:

\begin{lemma} \label{lem:pipeline} $Q$ is monotonically determined over $\views$ if and only if
every test succeeds.
\end{lemma}

We show that monotonic determinacy is decidable for some classes of views
by bounding the treewidth of all instances $D'$ that are the second component of some test.

\begin{theorem} \label{thm:decidemondetrewriting:frontierguarded} Suppose $Q$ and
$\views$ are Frontier-guarded Datalog queries.
Then there is an algorithm that decides if  $Q$ is monotonically determined by $\views$ in $\twoexp$.
\end{theorem}

\begin{proof} 
In this proof the words ``of required size'' mean ``doubly-exponential in $Q$ and single-exponential in $\views$'',
$\C$ stands for the class of all CQ approximations of $Q$, $\Sigma_\views$ is the view signature
and $\Sigma_{\baseschema}$ is the initial signature.

We must check whether $Q$ holds on all tests. As observed in the proof of Theorem~\ref{thm:frgd-rewriting}, there is an integer $k = O(|Q|, |\views|)$
bounding the treewidth of all CQ approximations of $Q$ and views in $\views$.
Let $\mathbb{V} = \{ \F \restr \Sigma_\views \mid \F \restr \Sigma_{\baseschema}  \in \C, 
\tw(\F) \le k, \F \models \Pi_\views\}$.
As argued in the proof of Theorem~\ref{thm:frgd-rewriting}, $\mathbb{V}$ is $k$-regular and captured by an NTA $\automaton_\views$ of required size.

Since $Q$ is a monotone query, instead of checking whether all tests succeed, we will check an equivalent condition that $Q$ holds on the class
$ETEST(Q, \views)$ which consists of all instances $D'$ which can be obtained from an instance in $\mathbb{V}$ by applying inverses of view definitions
while keeping the atoms of the view signature. Note that the treewidth of all instances in $ETEST(Q, \views)$ is also bounded by $k$. By Proposition~\ref{prop:cqappr:regular},
for each view $\view$ with definition $Q_\view$ there exists an automaton $\automaton'_{\view}$ running
on codes $\T$ which for each atom $\view(\vec c)$ at a node $n$  checks whether
$n$  has a descendant $n'$ such that $n'$ contains $\vec{c}$ and the subtree of $\T$ rooted at $n'$ is a code of some CQ approximation of $Q_\view$.
It should be clear that the automaton $\automaton_{ETEST}$  obtained as the product of $\automaton_\views$ and $\automaton'_{\view}$ for all $\view \in \views$
captures $ETEST(Q, \views)$.

By Proposition~\ref{prop:datalog:notaccept}, there is an NTA $\automaton''$ of required size 
which recognizes those codes which do not satisfy $Q$. 
So to check if $Q$ is monotonically determined by $V$ we construct the intersection 
of $\automaton_{ETEST}$ and $\automaton''$ (which is of required size) and check if it is empty. The latter check is linear in the size of the automaton.
\end{proof}

Using MDL normalization and the treewidth bounds of 
Lemma \ref{lem:boundingtw:2} we can
use the same proof technique to 
extend this to a mix of CQ and Frontier-guarded Datalog views, provided that $Q$ is in Monadic Datalog.

\begin{theorem} \label{thm:decidemondetrewritingmdlandcqs} Suppose $Q$ is in Monadic Datalog,
and $\views$ is a collection of CQ and Frontier-guarded Datalog views.
Then there is an algorithm that decides if  $Q$ is monotonically determined by $\views$ in $\threeexp$.
\end{theorem}

The previous cases of decidability required restricting
the views. We now observe that if we only restrict $Q$ to be a CQ, then we can reduce  monotonic determinacy
to checking equivalence between a recursive and a non-recursive query, the one created by
the ``simple forward backward method'' of Proposition \ref{prop:cqquery}.

\begin{theorem}\label{decidability:cq:equivalence}
If $Q$ is a CQ and $\views$ is a collection of Datalog views, then  the problem of monotonic determinacy 
of $Q$ over $\views$ is decidable in $\twoexp$.
\end{theorem}

\section{Lower bounds on testing monotonic determinacy} \label{sec:lower}
We now begin our negative results, starting
with lower bounds for testing  monotonic determinacy.
We first note some lower bounds on monotonic determinacy that can be obtained
through straightforward reductions from containment or equivalence:
\begin{proposition}\label{prop:easyhardness}
Monotonic determinacy is 
\begin{compactitem}
\item $\np$-hard for CQ queries and views \cite{lmss,thebook}
\item $\Pi_2^p$-hard for UCQ queries and UCQ views
\item $\twoexp$-hard for CQ queries and MDL views
\item $\twoexp$-hard for MDL queries and a fixed atomic view
\item undecidable for Datalog queries and a fixed atomic view (cf \cite{inverserules})
\end{compactitem}
\end{proposition}

It is more challenging to get undecidability results in settings
where the equivalence  problem for the views and queries is decidable,
as is the case for UCQs and Monadic Datalog \cite{cosmadakis1988decidable}.
The remainder of this section will be devoted to developing techniques for this case.

A \emph{tiling problem} is a tuple $TP =(Tiles, HC, VC, IT, FT)$ where $Tiles = \{T_1, \dots, T_k\}$,
  $HC$ and $VC$ are binary relations (``horizontal and vertical compatibility''),
and $IT$ and $FT$ are subsets of tiles that must be placed at the bottom left and top right corner respectively.

A \emph{solution} to a tiling problem consists of   numbers
 $n$ and $m$, and  map $\tau: \{1, \dots, n\} \times \{1, \dots, m\} \to Tiles$ such that 
\begin{compactitem}
\item[(T1)] $(\tau (i,j),\tau (i+1,j)) \in HC$ for $1 \le j \leq m$ and  $1 \le i < n $ ;
\item[(T2)] $(\tau (i,j),\tau (i,j + 1 )) \in VC$ for  $1 \le j < m  $ and  $1 \le i \le n$.
\item[(T3)] $\tau (1,1) \in IT$ and (T4) $\tau (n,m) \in FT$.
\end{compactitem}
By a standard reduction from the halting problem for Turing machines,
it is easy to show that the problem ``given a tiling problem TP, tell if it has a solution''
is undecidable. 
By reducing this tiling problem to the problem of monotonic determinacy for MDL queries
and UCQ views we obtain
\begin{theorem} \label{thm:undec}
\label{thm:undec-mdl-ucq}
The problem of monotonic determinacy for 
MDL queries and UCQ views is undecidable.
\end{theorem}

The idea of the reduction is, given TP, to construct $Q_{TP}$ and $\views_{TP}$ which generate 
tests for monotonic determinacy that look like $(n,m)$-grids with assignments of tiles. 
The query $Q_{TP}$ will have disjuncts that return ``true'' when they detect violations of conditions 
(T1)--(T4). Thus  $Q_{TP}$ and $\views_{TP}$ will have a failing test for 
monotonic determinacy iff the tiling problem
$TP$ has a solution.

\begin{figure}
\scalebox{0.7}{\includegraphics{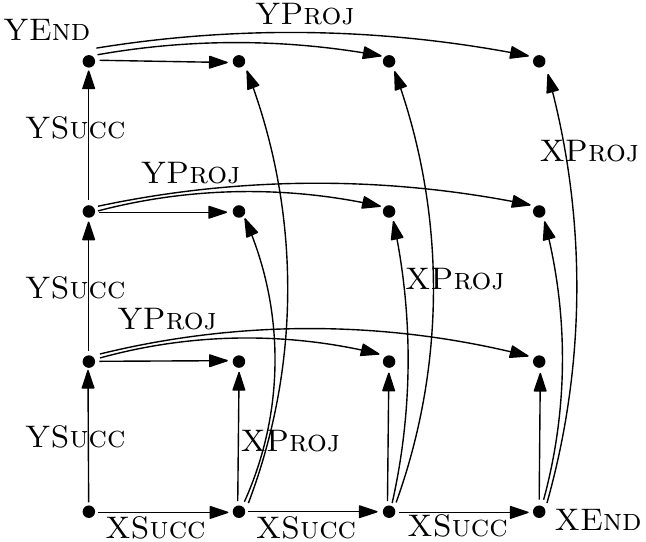}}\\[.1cm]\scalebox{0.6}{\includegraphics{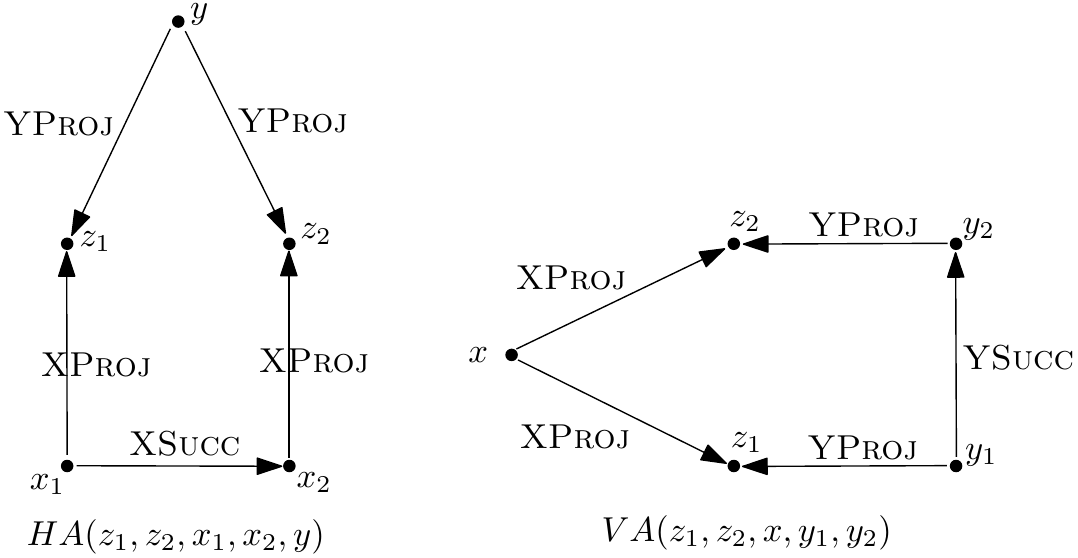}}
\caption{A grid-like test for monotonic determinacy (a) and CQs for checking horizontal and vertical adjacency between grid points (b)}
\label{figure:grid}
\end{figure}

Figure~\ref{figure:grid}, (a) shows such a test. 
We code the grid using four binary relations $\ysucc$, $\xsucc$, $\xproj$, $\yproj$ and unary markers
$\xend$ and $\yend$. Vertical and horizontal axes are represented as chains of $\ysucc$- and $\xsucc$-atoms 
respectively.  
The ``grid points'' are linked via $\xproj$- and $\yproj$-edges to their
projections on the axes. The unary predicates $\xend$ and $\yend$ mark the 
ends of the axes.

Note how CQs
$\ha(z_1, z_2,x_1,x_2,y) = 
\yproj(y, z_1)\land \yproj(y,z_2) \land \xproj(x_1, z_1) \land \xproj(x_2,z_2) \land \xsucc(x_1, x_2)$
and
$\va(z_1,z_2,x, y_1,y_2) = \yproj(y_1, z_1)\land \yproj(y_2,z_2) \land \xproj(x, z_1) \land \xproj(x,z_2) \land \xsucc(y_1, y_2)$
(see Figure~\ref{figure:grid}, (b)) can be used to check vertical and horizontal adjacency between grid points.
For example, $\ha(z_1, z_2,x_1,x_2,y)$ says  $z_1$ and $z_2$ have the same $y$-projection, while the $x$-projection of $z_2$ is 
next to the $x$-projection of $z_1$. 
Query $H(z_1, z_2) = \exists y \exists x_1 \exists x_2 \,\ha(z_1, z_2, x_1, x_2, y)$ 
holds of grid points $z^0_1$ and $z^0_2$ iff $z^0_2$ is the right neighbour of $z^0_1$.

Given a tiling problem $TP$, we define the query $Q_{TP}$
as a disjunction $Q_{\qstart} \lor Q_{\qhelper} \lor Q_{\qverify}$ where Monadic Datalog query $Q_{\qstart}$ and UCQs 
$Q_{\qhelper}$ and $Q_{\qverify}$ are defined by the following programs:

\begin{compactenum}
\item $Q_{\qstart} ~ \datalogarrow ~ A(x) \datalogwedge B(x)$  
\item $A(x) ~ \datalogarrow ~ \xsucc(x,x') \datalogwedge A(x') \datalogwedge C(x')$
\item $A(x) ~ \datalogarrow ~ \xend(x)$
\item  $B(y) ~ \datalogarrow ~ \ysucc(y,y') \datalogwedge B(y') \datalogwedge D(y')$
\item  $B(y) ~ \datalogarrow ~ \yend(y)$
\item $Q_{\qhelper} ~ \datalogarrow ~ C(u) \datalogwedge \yproj(y,z) \datalogwedge \xproj(x,z)$
\item $Q_{\qhelper} ~ \datalogarrow ~ D(u) \datalogwedge \yproj(y,z) \datalogwedge \xproj(x,z)$\\
\item $Q_{\qverify} ~ \datalogarrow ~ \ha(z_1,z_2,y,x_1,x_2) \datalogwedge T_i(z_1) \datalogwedge T_j(z_2)$ \\
 $\mbox{ for all pairs}$ $(T_i, T_j)\notin HC$ 
\item  $Q_{\qverify} ~ \datalogarrow ~ \va(z_1,z_2,y_1,y_2,x) \datalogwedge T_i(z_1) \datalogwedge T_j(z_2)$ \\
$\mbox{ for all pairs}$ $(T_i, T_j)\notin VC$
\item  $Q_{\qverify} ~ \datalogarrow ~ \ysucc(o,y) \datalogwedge \ysucc(y,z) \datalogwedge \xsucc(o,x) \datalogwedge \\
\xproj(x,z) \datalogwedge T_i(z)$ 
$\mbox{ for all }$ $T_i \notin IT$
\item $Q_{\qverify} ~ \datalogarrow ~ \yend(y) \datalogwedge \yproj(y,z) \datalogwedge  \\
T_i(z) \datalogwedge \xproj(x,z) \datalogwedge \xend(x)$ 
$\mbox{ for all }$ $T_i \notin FT$
\end{compactenum}

The set of views $\views_{TP}$ consists of
\begin{compactitem}
\item[--] the \emph{grid-generating
view}
$$
\begin{array}{rcl}
S(x,y) & \datalogarrow & C(x) \datalogwedge D(y)\\
S(x,y) & \datalogarrow & \xproj(x,z) \datalogwedge T_i(z) \datalogwedge \yproj(y,z) \mbox{ for all } T_i \mbox{ in } Tiles;\\
\end{array}
$$
\item[--] 
the \emph{atomic views} $V_\ysucc$, $V_\xsucc$, $V_\yend$, $V_\xend$ and $V_{T_i}$ for EDBs $\ysucc,\xsucc, \yend$, $\xend$ and each
$T_i$ in $Tiles$;
\item[--] the following \emph{special} views
$$
\begin{array}{crcl}
& V^{\qhelper}_C(u,x,y,z) & \datalogarrow & C(u) \datalogwedge \xproj(x,z) \datalogwedge \yproj(y,z) \\
& V^{\qhelper}_D(u,x,y,z) & \datalogarrow & D(u) \datalogwedge \xproj(x,z) \datalogwedge \yproj(y,z) \\
&  V_{\ha}(z_1, z_2, y,x_1,x_2) & \datalogarrow & \ha (z_1, z_2, y,x_1,x_2)\\
& V_{\va}(z_1, z_2, y_1,y_2,x) & \datalogarrow & \va (z_1, z_2, y_1,y_2,x)\\
& V_{I}(o,x,y,z) & \datalogarrow & \xsucc(o,x)\datalogwedge \xproj(x,z) \datalogwedge \\
&  & &  \ysucc(o,y) \datalogwedge \yproj(y,z) \\
& V_{F}(x,y,z) & \datalogarrow & \xproj(x,z) \datalogwedge \xend(x) \datalogwedge \\
& & & \yend(y) \datalogwedge \yproj(y,z) 
\end{array}
$$
\end{compactitem}

A typical CQ-approximation of $Q_{\qstart}$ is shown in Figure~\ref{fig:unfold} (a), and it generates the axes of the grid
which are marked with unary predicates $C$ and $D$. The view-image of such CQ is shown in Figure~\ref{fig:unfold} (b).
This view image for each grid-point contains an $S$-atom, and so a grid-like test as in Figure~\ref{figure:grid} (a)
can be constructed out of this view image by replacing each of these $S$-atoms with any of the disjuncts other than the first disjunct
in the definition of the grid-generating view. 

When we run $Q$ on the tests, $Q_{\qverify}$ comes into play. Note the correspondence between rules 8) -- 11) for $Q_{\qverify}$ and 
the negations of conditions (1) -- (4) in the definition of a solution of a tiling problem.
 Thus, when executed on a grid-test from Figure~\ref{figure:grid}~(a),
$Q_{\qverify}$ returns $\false$ iff a grid test is a solution to TP. 
The query $Q_{\qhelper}$ ensures that we are not harmed
in the case where the grid-generating views are applied with the first rule.

\begin{figure}
  \includegraphics[width=.5 \textwidth]{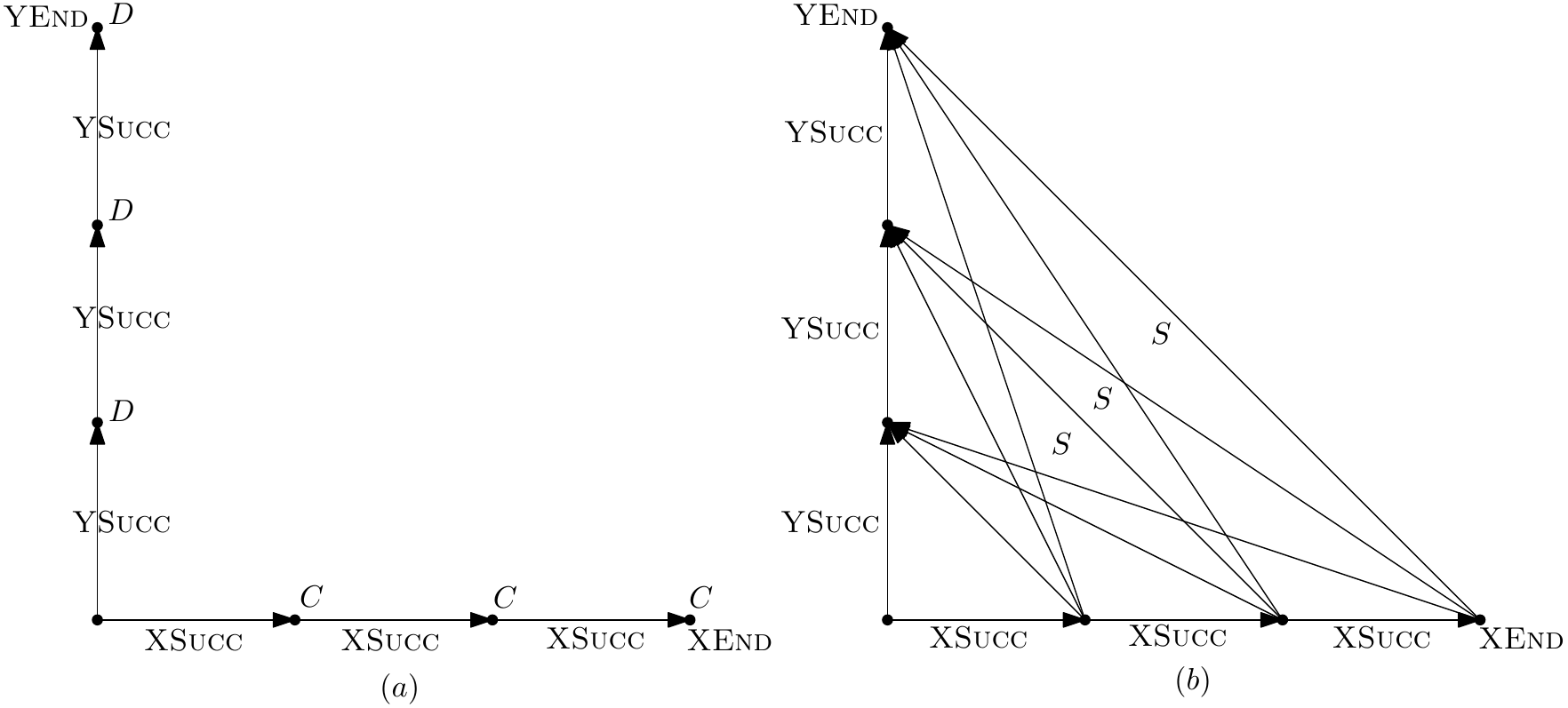}
  \caption{A typical approximation of a $\qstart$ atom (a) and its view image (b). (b) is obtained from (a) by replacing $C$ and $D$
  with their cross-product $S = C \times D$ }
  \label{fig:unfold}
\end{figure}

We can verify that a solution of our tiling problem corresponds to
monotonic determinacy, which will prove useful in both
our undecidability and non-rewritability results:

\begin{proposition} \label{prop:tpnotmondet}
$Q_{TP}$ is not monotonically determined by $\views_{TP}$ iff TP has a solution.
\end{proposition}

\section{Non-rewritability} \label{sec:nonrewrite}
We now turn to negative results concerning rewritability. 

\myparagraph{Pebble games}
In order to prove non-definability in Datalog and Monadic Datalog, we use the well-known tool of \emph{existential pebble games}.
A \emph{partial homomorphism} from $\inst$ to $\inst'$ is a mapping $h$ from a subset $D\subseteq \adom(\inst)$ to $\adom(\inst')$ such that           
$R(c_1 \ldots c_n) \in \inst$ implies $R(h(c_1) \ldots h(c_n)) \in \inst'$, provided each $c_i\in D$.
Let $k\geq 2$. In the \emph{existential $k$-pebble game}  we have two players, the \emph{Spoiler} and the \emph{Duplicator}, each having a set of pebbles $\{{\texttt p}_1,\dots,{\texttt p}_k\}$ and 
$\{{\texttt q}_1,\dots,{\texttt q}_k\}$, respectively. 
The game is played on two instances $\inst$ and $\inst'$ over the same schema. In each round, 
the Spoiler either places a pebble ${\texttt p}_i$ on some element of $\inst$ or removes ${\texttt p}_i$ from 
$\inst$, to which the Duplicator responds by placing its corresponding pebble ${\texttt q}_i$ on some element of $\inst'$ or by removing 
${\texttt q}_i$ from $\inst'$, respectively. 
The Duplicator wins the game if he has a \emph{winning strategy}, i.e., if he can 
indefinitely continue playing the game in such a way that after each round, 
if $a_1,\dots,a_k$ are the elements in $\inst$ marked by the 
Spoiler's pebbles $\{{\texttt p}_1,\dots,{\texttt p}_k\}$, and $a'_1,\dots,a'_k$ are the 
elements in $\inst'$ marked by the Duplicator's pebbles $\{{\texttt q}_1,\dots,{\texttt q}_k\}$, then the 
relation $\{(a_1,a'_1),\dots,(a_k,a'_k)\}$ is a partial homomorphism from $\inst$ to $\inst'$.

Recall that if there is a homomorphism from $\inst$ to $\inst'$, we write $\inst\to \inst'$. 
Similarly, if Duplicator wins the game on $\inst$ and $\inst'$, then we write $\inst\to_k \inst'$. 
Observe that $\inst\to \inst'$ implies $\inst \to_k \inst'$, for every $k\geq 2$.

The following property relates the game to homomorphisms from structures of bounded treewidth:


\begin{fact}\cite{AKV-cp04}
\label{fact:pebble-tw}
Let $k\geq 2$. Let $\inst$ and $\inst'$ be two instances over the same schema. Then the following are equivalent:
\begin{compactenum}
\item $\inst \to_k \inst'$, 
\item for every instance $\inst''$ of treewidth $\leq k-1$, if $\inst''\to \inst$, then $\inst''\to \inst'$. 
\end{compactenum}
\end{fact}

Existential pebble games with $k$ pebbles preserve
truth of Boolean Datalog queries with rule bodies of size at most $k$.
Thus games can be used to show non-definability in Datalog:

\begin{fact}\cite{KolaitisVardi95}
\label{fact:pebble-datalog}
Let $Q$ be a Boolean query. 
Suppose there exists two instances $\inst_k$ and $\inst'_k$ such that $Q(\inst_k)=\true$, $Q(\inst'_k)=\false$ and $\inst_k\to_k \inst'_k$, for infinitely many $k$'s. 
Then $Q$ is not definable in Datalog. 
\end{fact}

Let $\inst$ be an instance and $k\geq 2$ be an integer. An instance $\inst'$ is a \emph{$k$-unravelling} of $\inst$ if there is a homomorphism $\Phi$ from $\inst'$ to $\inst$ and a tree decomposition $(\tau,(\lambda(u))_{u\in \vertices(\tau)})$ of $\inst'$ of
width at most $k$, such that:

\begin{compactenum}
\item For each $u\in \vertices(\tau)$, the mapping $\Phi|_{\lambda(u)}$ is a partial isomorphism from $\inst'$ to $\inst$. 
\item For $u\in \vertices(\tau)$ with children $u_1,\dots,u_\ell$,
 the set \\
$\{\Phi(\lambda({u_1})),\dots,\Phi(\lambda({u_\ell}))\}$ contains the collection of all non-empty subsets of $\inst$ of size $\leq k$.
\end{compactenum}

If, in addition we have $|\lambda(u)\cap \lambda(v)| \le 1$ for all non-equal $u$ and $v$ in $\vertices(\tau)$, then we say that $\inst'$ is $(1,k)$-unravelling of $\inst$. 
Duplicator has a winning strategy between an instance
and its $(1,k)$-unravelling in a variation
of the $k$-pebble games in which at most  one pebble can remain in place 
in each move.
Such games preserve Boolean Monadic Datalog queries with bodies of size $k$, and
hence each Boolean Monadic Datalog query is preserved under $(1,k)$-unravellings for sufficiently large $k$. So we have the following variant of 
Fact \ref{fact:pebble-datalog}:
\begin{fact}
\label{fact:pebble-datalog:1}
Let $Q$ be a Boolean query. Suppose there exists two instances $\inst_k$ and $\inst'$ such that $Q(\inst_k)=\true$, $Q(\inst'_k)=\false$ and $\inst'_k$ is a (1,$k$)-unravelling of $\inst_k$, for infinitely many $k$'s. 
Then $Q$ is not definable in Monadic Datalog. 
\end{fact}

Note that the treewidth of any $k$-unravelling is at most $k-1$. Observe also that all $k$-unravellings of an
instance $\inst$ are homomorphically equivalent. The following facts about unravellings will  be
useful (see the appendix):

\begin{fact}
\label{fact:pebble-unravellings}
Let $k\geq 2$. Let $\inst$ be an instance and $U$ be any $k$-unravelling of $\inst$. Then the following hold:
\begin{compactenum}
\item $U\to \inst$ and $\inst \to_k U$.
\item For every instance $\inst'$, we have $\inst\to_k \inst'$ iff $U\to \inst'$. 
\end{compactenum}
\end{fact}

\myparagraph{Non-rewritability in Monadic Datalog}
We recall that Monadic Datalog queries monotonically determined
over CQ views always have $\fgdatalog$ rewritings (e.g. \cite{inverserules},
or Thm \ref{thm:moncq-rewriting} ).
We show that they may not be rewritable in MDL:
\begin{theorem}
\label{thm:no-monadic}
There exists a Monadic Datalog query $Q$ and a set of CQ views $\views$ such that
$Q$ is rewritable with respect to $\views$ in Datalog, but not in Monadic Datalog.
\end{theorem}

\begin{proof}
Consider the following Monadic Datalog query $Q$
$$
\begin{array}{rcl}
W(x) & \datalogarrow & A(x,y) \datalogwedge B(y,v) \datalogwedge C(x,z) \datalogwedge D(z,v) \datalogwedge U(v)\\
W(x) & \datalogarrow & A(x,y) \datalogwedge B(y,v) \datalogwedge C(x,z) \datalogwedge D(z,v) \datalogwedge W(v)\\
Goal    & \datalogarrow & W(x) \datalogwedge  M(x)
\end{array}
$$ and a set of views $\views$
$$
\begin{array}{rcl}
S(x,y,z) & \datalogarrow & M(x) \datalogwedge A(x,y) \datalogwedge C(x,z)\\
R(y,z,y',z') & \datalogarrow & B(y,v) \datalogwedge D(z,v) \datalogwedge A(v,y') \datalogwedge C(v,z') \\
T(y,z,v)    & \datalogarrow & U(v) \datalogwedge B(y,v) \datalogwedge D(z,v).\\
\end{array}
$$

$Q$ checks whether the instance $\inst$ contains the points $s \in M^\inst$ and $t \in U^\inst$
which are connected by a sequence of ``diamonds'' (see Figure~\ref{fig:diamond}, (a)).

\begin{figure}
\begin{center}
\scalebox{0.5}{\includegraphics{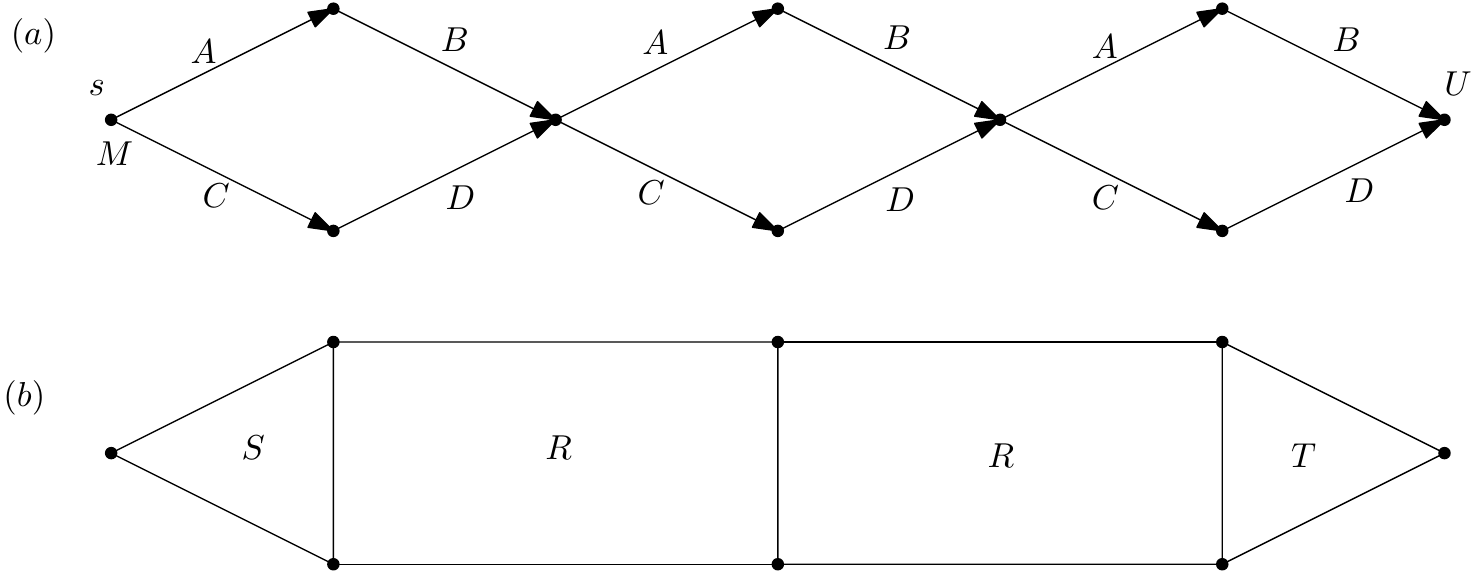}}\\
\end{center}
\caption{An unravelling of Q (a) and its view image (b)}
\label{fig:diamond}
\end{figure} 

\myeat{
Relying on Proposition~\ref{prop:infinity-rewriting} and Figure~\ref{fig:diamond}, (b) and  
it is easy to see that the query $Q'$
$$
\begin{array}{rcl}
P(y,z) & \datalogarrow & T(y,z,v)\\
P(y,z) & \datalogarrow & R(y,z,y',z') \datalogwedge P(y',z')\\
Goal    & \datalogarrow & S(x,y,z) \datalogwedge  P(y,z)
\end{array}
$$
is a Datalog rewriting of $Q$ with respect to the views. 
}

We claim that there is no Monadic Datalog rewriting of $Q$ in terms of these views.
To prove this, given an integer $k$, we construct two instances $\inst_k$ and $\inst'_k$ such that
$\inst_k \models Q$, $\inst'_k \not \models Q$, but Duplicator wins in the $(1,k)$-game for the view images 
$\views(\inst_k)$ and $\views(\inst'_k)$.

Let $\inst_k$ be a sequence of $k+1$ diamonds  from Figure~\ref{fig:diamond}, (a) and $\jnst_k$ be the view image of $\inst_k$
from Figure~\ref{fig:diamond}, (b). Let $\jnst'_k$ be the (infinite) $(1,k)$-unravelling of $\jnst_k$. 
Let $\inst'_k$ be the result of applying ``inverse rules'': 
$$
\begin{array}{rcl}
S(x,y,z) & \to & M(x) \land A(x,y) \land C(x,z)\\
R(y,z,y',z') & \to & \exists v \, B(y,v) \land D(z,v) \land A(v,y') \land C(v,z') \\
T(y,z,v)    & \to & U(v) \land B(y,v) \land D(z,v).\\
\end{array}
$$
to $\jnst'_k$ and removing the view predicates. There are two types of 
elements  in $\inst'_k$,
those that were present in $\jnst'_k$  and those introduced by the existential
quantifier over $v$ in the second rule which are called \emph{anonymous}.

We first claim that the view image of $\inst'_k$ is $\jnst'_k$. To see this, consider a homomorphism $h$ from
the body of 
$q(y,z,y',z') =  B(y,v), D(z,v), A(v,y'), C(v,z')$ 
into $\inst'_k$. Note that $h$ must map $v$ into an anonymous point, 
and so $y,z,y'$ and $z'$ must
be mapped to the points $y_0,z_0, y'_0, z'_0$ for some $R(y_0,z_0,y'_0,z'_0) \in \jnst'_k$.
If follows that any $R$-atom in $V(\inst'_k)$ is in $\jnst'_k$. This also holds for $S$-atoms
thanks to the unary predicate $M$; and  also for $R$-atoms, thanks to the unary predicate $U$.

Our next claim is that $\inst'_k \models Q$ iff $\jnst'_k \models Q'$. Indeed, $\inst_n$ maps into $\inst'_k$ iff 
$\jnst_n$ maps into $\jnst'_k$.

Finally, we claim that $\jnst'_k \not\models Q'$. We show that for any $n$ there is no homomorphism from 
$\jnst_n$ to $\jnst'_k$. Indeed, as $\jnst'_k$ maps homomorphically onto $\jnst_k$, for any points
$s, t$ in $\jnst'_k$ with $s \in \pi_1(S^{\jnst'_k})$, $t \in \pi_3(T^{\jnst'_k})$, the distance between $s$ and $t$
(measured in the Gaifman graph for $\jnst'_k$) cannot be less then $k+1$. This implies the claim
for $1 \le n \le k$. For $n \ge k+1$ the claim follows from the observation
that there is no homomorphism from the query describing the pattern in Figure \ref{fig:patt}
into $\jnst'_k$.  Indeed, it can be easily shown by induction that if such a homomorphism
existed, then there would be a single bag in $\jnst'_k$ containing its whole image.
This  is a contradiction, as all bags are of size $k$, and the query in question 
has $2k+2$ variables. 
\begin{figure}
\begin{center}
\scalebox{0.65}{\includegraphics{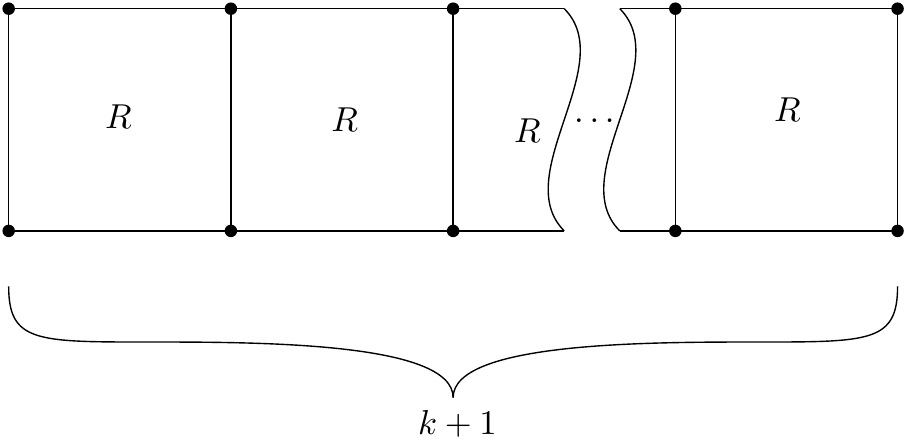}}\\
\end{center}
\caption{A long row of $R$ rectangles}
\label{fig:patt}
\end{figure}

It follows that $\inst_k$ and $\inst'_k$ are as required.
\end{proof}

\myparagraph{Non-rewritability in Datalog}
Now we show that monotonic determinacy does not imply Datalog rewritability, even for 
Monadic Datalog queries and UCQ views:

\begin{theorem}
\label{thm:no-datalog-mdl-ucq}
There exists a Monadic Datalog query $Q$ and a set of UCQ views $\views$ such that $Q$ is monotonically determined by $\views$ but there is no Datalog rewriting  of $Q$ over $\views$.
\end{theorem}

The query $Q$ and views $\views$ we will use 
have the form $Q_{TP^*}$ and $\views_{TP^*}$ for a particular tiling problem $TP^*$ as defined in Section~\ref{sec:lower}. 
We define a schema $\delta:=\{\texttt H, V, I, F\}$  where $\texttt H, V$ are binary and $\texttt I, F$ are unary relations.
Let $TP=(Tiles, HC,VC, IT, FT)$ be a tiling problem. 
Given a database instance $\inst$ for schema $\delta$, we say that $\inst$ \emph{can be  tiled
by  $TP$} if there is an assignment of each element of $\inst$ to a tile where $\texttt H, V, I, F$
satisfy the horizontal, vertical, initial and final constraints. 
We denote by $\inst_{TP}$ the tiling problem $TP$ viewed as a relational structure for $\delta$, with $Tiles$ the domain.
Then an instance for $\delta$
can be tiled by $TP$ exactly when it has a homomorphism into $\inst_{TP}$.
For  $n,m\geq 1$, we denote by $\inst_{n,m}^{grid}$ the database instance with 
domain $\{(i,j):\text{$1\leq i\leq n$ and $1\leq j\leq m$\}}$ and facts ${\texttt I}((1,1))$, ${\texttt F}((n,m))$, 
${\texttt H}((i,j), (i+1,j))$, for every $1\leq i<n$ and $1\leq j\leq m$, and ${\texttt V}((i,j),(i,j+1))$, for every $1\leq i\leq n$ and $1\leq j< m$. Then $TP$ has a solution in the usual sense 
if  $\inst_{n,m}^{grid}$ can be tiled with $TP$.

We can adapt techniques of \cite{atserias07:power} to show that
there is a tiling problem for which no $n \times m$ rectangular grid can be tiled, but where for each
$k$ large enough grids can be ``$k$-approximately tiled'', in the sense of having $k-$unravellings
that can be tiled.
\begin{lemma}
\label{lemma:special-tiling}
There is a tiling instance $TP^*$ such that $\inst_{n,m}^{grid}$ can not be tiled with $TP^*$ for each $n,m \geq 1$
but for each $n,m \geq 3$ and each $k$ with $2 \leq k < \min\{n,m\}$ any  $k$-unravelling
of $\inst_{n,m}^{grid}$  can be tiled with ${TP^*}$.
\end{lemma}

\begin{proof}[Proof of Theorem~\ref{thm:no-datalog-mdl-ucq}]
Let $TP^*$ be the tiling instance from Lemma~\ref{lemma:special-tiling} and let $Q_{TP^*}$ and $\views_{TP^*}$ be the MDL query and UCQ views from Theorem~\ref{thm:undec-mdl-ucq}. 
Recall that $Q_{TP^*}$ and $\views_{TP^*}$ are defined over the schema 
\[
\sigma:=\{\xsucc,\ysucc,C,D,\xend,\yend, \xproj,\yproj, T_1,\dots, T_p\}
\]
where $\{T_1,\dots,T_p\}$ is the tile set of $TP^*$. 
Since $\inst_{n,m}^{grid}$ cannot be tiled with $TP^*$, for each $n,m\geq 1$, the tiling instance $TP^*$ has no solution, and hence $Q_{TP^*}$ is monotonically determined by $\views_{TP^*}$. 

Fix $\ell\geq 10$. We shall define instances $\inst_\ell$ and $\inst'_\ell$ over $\sigma$ such that $\views_{TP^*}(\inst_\ell)\to_{\lfloor\sqrt{\ell-1}\rfloor}\views_{TP^*}(\inst'_\ell)$, $Q_{TP^*}(\inst_\ell)=\true$ and $Q_{TP^*}(\inst'_\ell)=\false$. 
By applying Fact~\ref{fact:pebble-datalog}, this implies that $Q_{TP^*}$ has no Datalog rewriting over $\views_{TP^*}$, as required.  
The instance $\inst_\ell$ has domain $z_0\cup X\cup Y$, where $X:=\{x_1,\dots,x_\ell\}$ and $Y:=\{y_1,\dots,y_\ell\}$, and facts $D(x_i), C(y_i)$, for all $1\leq i\leq \ell$, along with
\[
\xsucc(x_i, x_{i+1}), \ysucc(y_i,y_{i+1}) \mbox{ for all } 1\leq i< \ell
\]
along with $\xend(x_\ell)$, $\yend(y_\ell)$, $\ysucc(z_0,y_1)$ and $\xsucc(z_0,x_1)$.
Figure~\ref{fig:unfold} (a) depicts $\inst_3$. Informally,
$\inst_\ell$ is the expansion of $Q_{TP^*}$ (more precisely of $Q_{\qstart}$) representing the $(\ell\times\ell)$-grid. In particular, $Q_{TP^*}(\inst_\ell)=\true$.

Intuitively, we would now like to define $\inst'_\ell$ so that its view
image contains a $\lfloor\sqrt{\ell-1}\rfloor$ unravelling of the view image of $\inst_\ell$.
By Fact \ref{fact:pebble-unravellings} 
we would have  $\views_{TP^*}(\inst_\ell)\to_{\lfloor\sqrt{\ell-1}\rfloor}\views_{TP^*}(\inst'_\ell)$ as required.
But using Lemma \ref{lemma:special-tiling} and the definition of $Q_{TP^*}$ we hope to show
$Q_{TP^*}(\inst'_\ell)=\false$. We will follow this intuition, but to define
the appropriate  $\inst'_\ell$ we will need to construct several auxiliary instances.
Let $E_\ell:=\views_{TP^*}(\inst_\ell)$. Figure~\ref{fig:unfold} (b) depicts $E_{3}$. 
Recall that view images are defined over schema  $\tau$:
\begin{flalign*}
& \{V_\xsucc,V_\ysucc, V_\xend,V_\yend, V_{T_1},\dots, V_{T_p},  V^\qhelper_C,V^\qhelper_D, & \\
&V_{HA}, V_{VA}, V_I, V_F, S\} &
\end{flalign*}
Intuitively, $E_\ell$ copies the $\xsucc,\ysucc,\xend$ and $\yend$-facts from $\inst_\ell$, while the $S$-facts correspond to the product $Y\times X$.  
Let $U_\ell$ be a $\lfloor\sqrt{\ell-1}\rfloor$-unravelling of $E_\ell$, which is witnessed by a 
homomorphism $\Phi:U_\ell\to E_\ell$ and a tree decomposition $(\tau,(\lambda(u))_{u\in \vertices(\tau)})$ of $U_\ell$.

In order to exploit Lemma~\ref{lemma:special-tiling}, we need to interpret $U_\ell$ as an unravelling of the grid $\inst_{\ell,\ell}^{grid}$.
The idea is to define a new instance $W_\ell$ over schema $\delta=\{\texttt H, V, I, F\}$ (recall that $\delta$ is the schema of $\inst_{\ell,\ell}^{grid}$) whose domain consists of all the $S$-facts
of $U_\ell$ and the horizontal and vertical successor relations are interpreted in the natural way.
Thus we can think of $W_\ell$ as an unravelling of the $S$-facts of $E_\ell$, which in turn correspond to grid points of $\inst_{\ell,\ell}^{grid}$ (the fact $S(y_j,x_i)$ corresponds to the point $(i,j)$).
Formally, $W_\ell$ is defined as follows:
\begin{compactenum}
\item The domain of $W_\ell$ contains all pairs $(w,z)$ such that $S(w,z)$ is a fact in $U_\ell$.
\item ${\texttt I}((w,z))$ is a fact iff $\Phi(w)=y_1$ and $\Phi(z)=x_1$. Similarly, ${\texttt F}((w,z))$ is a fact iff $\Phi(w)=y_\ell$ and $\Phi(z)=x_\ell$.
\item ${\texttt H}((w,z),(w',z'))$ is a fact iff $w=w'$ and
$V_\xsucc(z,z')$ is a fact in $U_\ell$.
Similarly, ${\texttt V}((w,z),(w',z'))$ is a fact iff $z=z'$ and $V_\ysucc(w,w')$ is a fact in $U_\ell$.
\end{compactenum}

\begin{claim}
 \label{claim:wl}
$W_\ell$ can be tiled by $TP^*$.
\end{claim}
\begin{proof}
We use the characterization for tilings of $W_\ell$  as homomorphisms into $\inst_{TP^*}$.
Then by Lemma \ref{lemma:special-tiling} and Fact~\ref{fact:pebble-tw}, it suffices
to show (a) $W_\ell\to \inst_{\ell,\ell}^{grid}$ and (b)  $\tw(W_\ell) \leq \ell-2$. 

For (a), we can take the homomorphism $\psi$ such that for every $(w,z)$ in $W_\ell$, we have $\psi((w,z))=(i,j)$ iff $\Phi(w)=y_j$ and $\Phi(z)=x_i$, for $1\leq i,j\leq \ell$. 
Let us argue that $\psi$ is a homomorphism. If ${\texttt I}((w,z))$ is a fact in $W_\ell$, by definition we have $\Phi(w)=y_1$ and $\Phi(z)=x_1$, 
and then ${\texttt I}(\psi((w,z)))={\texttt I}((1,1))$, which is a fact in $\inst_{\ell,\ell}^{grid}$. 
If ${\texttt H}((w,z),(w',z'))$ is a fact in $W_\ell$, then $w=w'$ and $S(w,z), S(w',z'), V_\xsucc(z,z')$ are facts in $U_\ell$. 
It follows that $\Phi(w)=\Phi(w')=y_j$, $\Phi(z)=x_i$ and $\Phi(z')=x_{i+1}$, for some $1\leq j\leq \ell$ and $1\leq i<\ell$. 
Hence ${\texttt H}(\psi((w,z)),\psi((w',z')))={\texttt H}((i,j),(i+1,j))$, which is a fact in $\inst_{\ell,\ell}^{grid}$. The argument is analogous for $\texttt F$ and $\texttt V$-facts. 

For condition (b), recall that $(\tau,\lambda)$
is a  decomposition of 
$U_\ell$ with $|\lambda(u)|\leq \lfloor\sqrt{\ell-1}\rfloor$, for all $u\in \vertices(\tau)$. 
We define a decomposition 
$(\tau',\lambda')$
 for $W_\ell$ with
$\tau':=\tau$ and, for each $u\in \vertices(\tau')$, we have 
\[
\lambda'(u):=\{(w,z): \text{$\{w,z\}\subseteq \lambda(u)$ and $S(w,z)$ is a fact in $U_\ell$}
\]
The connectedness condition is  inherited from $\tau$.
Suppose that we have a fact ${\texttt H}((w,z),(w',z'))$ in $W_\ell$. Then $w=w'$, and 
$S(w,z), S(w',z'), V_\xsucc(z,z')$ are facts in $U_\ell$. There must exist $u\in \vertices(\tau)=\vertices(\tau')$ such that  
$\{w=w',z,z'\}\subseteq \lambda(u)$ (as every clique is always contained in a bag). 
It follows that $\{(w,z),(w',z')\}\subseteq \lambda'(u)$. 
The argument for $\texttt V$-facts is analogous. 
Finally, note that $|\lambda'(u)|\leq |\lambda(u)|^2\leq \ell-1$, for all $u\in \vertices(\tau')$. We conclude that the treewidth of $W_\ell$ is $\leq \ell-2$ as required.
\renewcommand{\qedsymbol}{$\blacksquare$}
\end{proof}

Using the tiling solution $\chi$ for $W_\ell$ given by Claim~ \ref{claim:wl} and ``chasing with the inverse
rules of the view definitions''
we  can move to the desired instance $\inst'_\ell$ for the base schema $\sigma$.
The instance $\inst'_\ell$ is obtained from $U_\ell$ by replacing each fact $V_\xsucc(w,z)$, $V_\ysucc(w,z)$, $V_\xend(w)$ and $V_\yend(w)$, 
by facts $\xsucc(w,z)$, $\ysucc(w,z)$, $\xend(w)$ and $\yend(w)$, respectively; and by replacing each fact $S(w,z)$ by facts $\xproj(z,s_{w,z})$, $\yproj(w,s_{w,z})$,  
and 
$T_i(s_{w,z})$, where $s_{w,z}$ is a fresh element and $\chi((w,z))=T_i$. By construction, all facts of $U_\ell$ are contained in 
those of $\views_{TP^*}(\inst'_\ell)$ and 
hence $U_\ell\to \views_{TP^*}(\inst'_\ell)$. 
By Fact~\ref{fact:pebble-unravellings} (2), we have $\views_{TP^*}(\inst_\ell)\to_{\lfloor\sqrt{\ell-1}\rfloor} \views_{TP^*}(\inst'_\ell)$. 

It remains to show that $Q_{TP^*}(\inst'_\ell)=\false$. Since there are no $C$ or $D$-facts in $\inst'_\ell$, $Q_{\qstart}$ and $Q_{\qverify}$ cannot hold in $\inst'_\ell$. 
Towards a contradiction, suppose some rule (8)--(11) holds in $\inst'_\ell$. 
If rule (8) holds then there are elements $w,z,z'$ in $\inst'_\ell$ and facts 
\begin{align*}
\yproj(w,s_{w,z}),\yproj(w,s_{w,z'}),\xproj(z,s_{w,z}), \\
\xproj(z',s_{w,z'}), \xsucc(z,z')
\end{align*}
along with $T_i(s_{w,z})$, $T_j(s_{w,z'})$, for tiles $(T_i,T_j)\not\in HC$ for $TP^*$. By construction of 
$\inst'_\ell$, 
we know $S(w,z)$, $S(w,z')$ and $V_\xsucc(z,z')$ are in $U_\ell$. 
In particular, ${\texttt H}((w,z),(w,z'))$ is a fact in $W_\ell$. On the other hand, 
by definition of $\inst'_\ell$, we know  $\chi((w,z))=T_i$ and $\chi((w,z'))=T_j$. 
Since $\chi$ is a valid tiling of  $W_\ell$ for ${TP^*}$, 
 $(T_i,T_j)\in HC$ in $TP^*$; a contradiction. 
The case of rule (9) is symmetric. 
If rule (10) holds,  there are  $u,w,z$ in $\inst'_\ell$ and facts 
\[
\ysucc(u,w), \xsucc(u,z), X(z,s_{w,z}), Y(w,s_{w,z})
\] 
along with $T_i(s_{w,z})$, for some tile $T_i$ not an initial tile of $TP^*$ . 
It follows that $V_\ysucc(u,w), V_\xsucc(u,z)$ and $S(w,z)$ are facts in $U_\ell$. 
Note that ${\texttt I}((w,z))$ is a fact in $W_\ell$ since $\Phi$ is a homomorphism from the unravelling $U_\ell$ to $E_\ell$ and then we must have $\Phi(u)=z_0$, $\Phi(w)=y_1$ and $\Phi(z)=x_1$. 
Now by definition of $\inst'_\ell$, we know that $\chi((w,z))=T_i$. Since $\chi$ is a valid
tiling of  $W_\ell$ for  ${TP^*}$, $T_i$ is an initial tile of
$TP^*$, which is a contradiction.
The argument for rule (11) is analogous. We conclude that $Q_{TP^*}(\inst'_\ell)=\false$. 
\end{proof}

\myparagraph{Complexity of separators}
Thus far we have seen that there may be no Datalog rewriting even in the case of UCQ views.
What about separators, which are like rewritings, but not required
to be in a logic?
It is easy to see that for UCQ queries and views, there is always
a rewriting in $\conp$ and a rewriting in $\np$. This is true because every view image
is the view image of a small instance; basically the same observation was made for regular path
queries in \cite{determinacyregularpath}. Thus if we want really strong lower bounds, we need
to deal with recursive queries, and we need to look beyond regular path queries.

We show that when we turn to general Datalog queries and views, there may be no separator
in $\ptime$. In fact, 
we can find monotonically determined examples  with no separator that can be performed within any given computable time bound. 
\begin{theorem}  \label{thm:nocomputable}
There is no function $F$ such that
for all $Q, \views$ such that $\views$ and $Q$ are in Datalog and
$Q$ is monotonically determined over $\views$, there is a separator
of $Q$ over $\views$
that runs in time
$F(\views(I))$.
\end{theorem}

The proof  is 
inspired by  a construction in  \cite{determinacyregularpath} which obtained
Datalog views and queries where the certain answers are difficult to 
compute.  
Roughly speaking, we modify this by considering a query verifying that the base data
represent an input and a valid computation of a high-complexity deterministic
Turing Machine, while the views verify that the computation is halting and return the input.
Determinism of the machine will imply monotonic determinacy of the query over the views.
An efficient separator will contradict the high-complexity of the machine.
Details are in the appendix.

\section{Conclusion} \label{sec:conc}
We have taken some basic steps in understanding monotonic determinacy
for recursive queries. We leave quite a number of gaps in both the understanding of
rewritability and decidability/complexity of testing monotonic determinacy, as one can see
from  Figures  \ref{fig:table:rewritability} and  \ref{fig:table:decidability}.
 To highlight
just one, while we have shown that monotonic determinacy of a Datalog query over  Datalog
views does not imply a rewriting in any reasonable  complexity class, we do not know what
can be said when the query is restricted; e.g. to be in Frontier-guarded Datalog.
While  we have shown that when a Monadic Datalog query is monotonically
determined over UCQ views, it may
not be Datalog rewritable, we do not know whether a rewriting can be obtained by expanding
the language, e.g. to stratified Datalog; this is true of the particular rewritings constructed
in Theorem \ref{thm:no-datalog-mdl-ucq}. See the appendix.

\bibliographystyle{abbrv}
\bibliography{detarxiv}
\newpage
\appendix
\onecolumn
\noindent {\large\textbf{APPENDIX}}
\section*{Finite variants}
In the body of the paper we used a semantics
in terms of arbitrary instances, but we claimed that
all the results hold when instances are restricted to be finite.
One can relativize all of the definitions to finite instances.
We say $Q$ is monotonically determined over $\views$ \emph{with respect to finite
instances} if whenever two finite instances agree on $\views$ they must agree on $Q$.
Similarly we can talk about  $Q$ being rewritable in logic $L$ with
respect to views $\views$ over finite instances
if there is a query $R \in L$ such that for every finite instance $\inst$
evaluating $R$ on the $\views$-image of $\inst$ gives the same result
as evaluating $Q$ on $\inst$.

For the results about deciding monotonic determinacy and the positive
results about  rewriting monotonically determined queries in Datalog,
the equivalence of the finite and unrestricted
variants follows from the following well-known fact:

\begin{proposition}
If Datalog query $Q$ is monotonically determined over Datalog views $\views$ for
finite instances, then $Q$ is monotonically determined over $\views$  (over all instances). 

For all the languages $L$ we consider
here (Datalog, MDL, etc.) if   $Q$ is $L$-rewritable over $\views$
with respect to finite instances, it is $L$-rewritable with respect
to all instances.
\end{proposition}

\begin{proof}
We prove the first statement.
Assume $Q$ is monotonically determined over Datalog views $\views$ for
finite instances, and suppose we have two instances $\inst$ and $\inst'$, perhaps
infinite
with $\views(\inst) \subseteq \views(\inst')$ and $Q(\inst)$ not contained
in $Q(\inst'$). Fix $\vec t \in Q(\inst)-Q(\inst')$. There is a finite
subinstance $\inst_0$ of $\inst$ with $\vec t \in Q(\inst_0)$.
$\views(\inst_0) \subseteq \views(\inst')$ and $\views(\inst_0)$ is finite, so there
is a finite subinstance $\inst'_0$ of $\inst'$ with $\views(\inst_0) \subseteq \views(\inst'_0)$. 
Now $\inst_0$ and $\inst'_0$ contradict the hypothesis on $Q$ and $\views$.

For the the second statement, we use the fact that equivalence of Datalog
queries over finite instances implies equivalence over all instances.
\end{proof}

It remains to  consider our negative results about rewritings and separators.
An easy case is Theorem \ref{thm:nocomputable}, showing that no computable
function bounds the time of a separator for Datalog queries
monotonically-determined over 
Datalog views.
The query and views in the example are monotonically
determined over all instances, hence over finite instances;
 the argument that the query does not have separators growing at a given
time bound does makes no use of infinitary methods, and hence holds to show
that there is no such separator over finite instances.

The proofs of Theorems \ref{thm:no-monadic} and \ref{thm:no-datalog-mdl-ucq} both make use
of unravellings, which can be infinite. For Theorem
\ref{thm:no-monadic}, we can argue just by looking at  the statement: if
the query $Q$ can be rewritten to a Monadic Datalog query $R$ over the
views $\views$ with respect to finite instances, then consider the Datalog query
$R_\views$ formed by composing the rules for $\views$ and $R$, treating each
view predicate as an intensional predicate of $R_\views$. Then $Q$ is equivalent
to $R'$ over all finite instances. But then, using the fact that
a witness to non-containment of Datalog must be finite, we see that
$Q$ is equivalent to $R'$ over all instances, a contradiction of
the theorem. The same argument holds for Theorem \ref{thm:no-datalog-mdl-ucq}.

\newpage
\section*{Rewritability  results inherited from prior work}

In the body of the paper we claimed
that
by simply applying the ``inverse rules'' algorithm \cite{inverserules}
we can show that frontier-guarded Datalog queries monotonically determined over
CQ views have frontier-guarded Datalog rewritings
 over CQ views.
We now explain why this is the case.

We recall some basics about the inverse rules algorithm, which works
by first constructing a logic program and then ``de-functionalizing the program'': mimicking the
function symbols with annotated predicates.
Logic programs are generalizations of Datalog programs that
allow function symbols in the head of rules. The semantics is via fixed point
as with Datalog. If we have in addition a distinguished Boolean intensional predicate,
the goal predicate, we can talk about a logic program query, projecting the output
of the fixpoint onto the goal predicate.

Consider a collection of CQ views $\views = \{(V, Q_V) \mid V \in \Sigma_\views\} $ over
a base schema $\schema$.
we associate a
set of TGDs:
$\Gamma_\views = \{V(\vec{x}) \to \exists \vec{y} \, Q_V(\vec{x},\vec{y})\}$ consisting of \emph{inverse rules}.
By replacing the existential quantifiers  with skolem functions,
we can consider  $\Gamma_{\views}$ as a logic program with input signature
the view schema and the base schema as intensional relations.
\begin{example} \label{ex:inv}
Suppose we have just one view, $V(x,y,z) = \exists u  ~ S(x,y,u) \wedge S(u,y,z)$

Then the corresponding inverse rules are:
\begin{align*}
S(x, y, f(x,y,z)) \datalogarrow  V(x,y,z) \\
S(f(x,y,z), y, z) \datalogarrow V(x,y,z)
\end{align*}
where $f$ is a skolem function.
\end{example}
Note that these rules have a single atom in each rule body. We call such rules \emph{atomic}.
We refer to the intensional predicates of the rules as \emph{extensionally-based
IDBs}.

Let $Q$ be a Boolean Datalog query over the base signature, with goal predicate $\goal_Q$.
We write $Q \cup \Gamma_\views$, the \emph{inverse rules logic program} to indicate the query
that unions the rules of $Q$ and those of $\Gamma_\views$, using the goal predicate $\goal_Q$.
This is a
a Boolean logic program query over the base signature.

\begin{example} \label{ex:invruleslogicprogram}
Consider the frontier-guarded Datalog query $Q$ with goal predicate $\goal$
and rules:
$$
\begin{array}{rcl}
\goal & \datalogarrow & \connq(x,x) \\
\connq(x,y) & \datalogarrow & S(x,y, z) \datalogwedge  \connq(x,z) \datalogwedge \connq(z,y) \\
\connq(x,y) & \datalogarrow & S(x,y,z)
\end{array}
$$
Then the inverse rules logic program $Q \cup \Gamma_{\views}$ for query $Q$
and $\views$ the single view in Example \ref{ex:inv} would have all the rules for $Q$
along with the two inverse rules  for $V$ coming from Example \ref{ex:inv}.
Observe that $S$ is now an intensional predicate along with  $\goal$
and $\connq$.
\end{example}

When $Q$ is frontier-guarded,
the inverse-rules logic program is not necessarily frontier-guarded. Indeed,
the rules in $\Gamma_{\views}$ contain only intensional predicates.
However, it is immediate that the head variables of each rule are contained
in some atom with an extensionally-based intensional predicate.


One can characterize the output of the inverse rules logic program on an arbitrary
 instance of the view schema, not just those which are view images of an instance
of the base schema.
Given a query $Q$ over the  input schema and an instance $\vinst$ of the view
schema, the  \emph{certain answers} of $Q$ with respect to $\views$ over $\vinst$ is
the intersection of $Q(\inst)$ over all $\inst$ such that $\views(\inst) \subseteq \vinst$.
It is easy to see that:

\begin{theorem} \cite{inverserules}  For any such $\vinst$, $Q \cup \Gamma_\views$ evaluated
on $\vinst$ gives the certain answers of $Q$ with respect to $\views$ over $\vinst$.
\end{theorem}

It follows from the theorem that
if $Q$ is monotonically determined by $\views$ and $\inst$ is any instance of the base
schema,
$Q \cup \Gamma_\views$ evaluated over $\views(\inst)$ is the same as $Q(\inst)$. 
Using
the terminology from the body of the paper this can be restated as:
$Q \cup \Gamma_\views$ is 
 a separator for $Q$ with respect to $\views$.

We now turn to the de-functionalization step of the inverse rules algorithm.
This replaces the inverse rules logic program with an ordinary Datalog program over
a different set of intensional predicates. These predicates are obtained by \emph{annotating} the intensional
predicates of the logic program. The idea is that 
an IDB atom $R(x,f(x,y))$ containing skolem term $f$ created during the construction
of the fixpoint of the logic program 
will be mimicked by an  atom $R^{1, f(1,2)}(x,y)$ created during the fixpoint
of the de-functionalized program.
We refer to the original paper \cite{inverserules} for the details,
but  illustrate
the idea with an example.

\begin{example} \label{ex:defunct}
We give part of the de-functionaliziation of the inverse rules logic program
from Example \ref{ex:invruleslogicprogram}.
The inverse rules themselves will translate to the rules:
\begin{align*}
S^{1,2, f(1,2,3)} (x, y, z) \datalogarrow  V(x,y,z) \\
S^{f(1,2,3), 2, 3} (x, y, z) \datalogarrow V(x,y,z)
\end{align*}
The following rule in the inverse rules logic program:
\begin{align*}
\connq(x,y) \datalogarrow S(x,y, z) \datalogwedge  \connq(x,z) \datalogwedge \connq(z,y) 
\end{align*}
 will generate many annotated rules.
One can consider the substitution $x=f(x_1, y_1, z_1),  y=y_1 z=f(x_2,y_2,z_2)$
 into the rule above, which gives the rule:
\begin{align*}
\connq(f(x_1, y_1, z_1), y_1) \datalogarrow 
S(f(x_1,y_1, z_1), y_1, z_1) \datalogwedge \connq(f(x_1,y_1,z_1), f(x_2,y_2,z_2) ) 
\datalogwedge \connq( f(x_2,y_2,z_2), y_1 )
\end{align*}

The corresponding annotated rule would be:
\begin{multline*}
\connq^{f(1,2, 3), 2} (x_1,y_1, z_1)  \datalogarrow  \\
S^{f(1,2,3), 2, 3} (x_1, y_1, z_1)  \datalogwedge
\connq^{f(1,2,3), f(4,5, 6)} (x_1, y_1, z_1, x_2, y_2,z_2) \datalogwedge
\connq^{f(1,2,3), 4 } (x_2, y_2,z_2, y_1)
\end{multline*}
\end{example}

We can see that annotated rules as produced by the standard inverse-rules algorithm
are not frontier-guarded. However, we note that:
\begin{itemize}
\item  each IDB  that is an annotation of an extensionally-based IDB appears in exactly one
rule, and that rule is atomic. Hence in particular the unique
such rule has a view atom as a frontier-guard.
\item the head variables of each  rule
co-occur in an atom that  is an annotation of an extensionally-based relation.
\end{itemize}
From this it follows that we can conjoin to each rule a view  atom that makes the rule
frontier-guarded:  namely, we can conjoin the view atom
corresponding to the annotated  extensionally-based relation in the second item.

\begin{example} \label{ex:betterdefunct}
Continuing the example above, the annotated rule:
\begin{align*}
\connq^{f(1,2, 3), 2} (x_1,y_1, z_1)  \datalogarrow \\
S^{f(1,2,3), 2, 3} (x_1, y_1, z_1)  \datalogwedge
\connq^{f(1,2,3), f(4,5, 6)} (x_1, y_1, z_1, x_2, y_2,z_2) \datalogwedge
\connq^{f(1,2,3), 4 } (x_2, y_2,z_2, y_1)
\end{align*}
is converted to the frontier-guarded rule:

\begin{align*}
\connq^{f(1,2, 3), 2} (x_1,y_1, z_1)  \datalogarrow  \\
V(x_1, y_1, z_1) \datalogwedge
S^{f(1,2,3), 2, 3} (x_1, y_1, z_1)  \datalogwedge
\connq^{f(1,2,3), f(4,5, 6)} (x_1, y_1, z_1, x_2, y_2,z_2) \datalogwedge
\connq^{f(1,2,3), 4 } (x_2, y_2,z_2, y_1)
\end{align*}
\end{example}

\newpage
\section*{Proofs for Section \ref{sec:inf}: treewidth bounds and the forward-backward method}
\subsection*{Proof of Lemma \ref{lem:normalised:treedecomp}}
Recall the statement:

\medskip

Let $Q$ be a normalized Monadic Datalog query. Then there is a number $k = O(|Q|)$ such that all CQ-approximations
of $Q$ have tree decomposition $\TD$ of width $k$ with $\tspan(\TD) \le 2$.

\medskip

\begin{proof}
In fact, $k$ is the maximal number of variables in a body of $Q$. Then the definition of a CQ-approximation
gives rise to a tree decomposition $\TD$ of width $k$. The property $\tspan(\TD) \le 2$ follows from the fact that $Q$
is normalized.
\end{proof}

\subsection*{Proof of Lemma \ref{lem:boundingtw:1} }
Recall the statement:

\medskip

If $\Pi$ is a Datalog program such that all its rules are frontier-guarded, and $\inst$ is an instance of treewidth $k$, then
$\fpeval{\Pi}{\inst}$ is of treewidth $k$.

\medskip

\begin{proof}
Monadic rules introduce only monadic predicates which do not increase the treewidth of the instance.
Guarded rules introduce atoms which are wholly inside the EDB guards. So treewidth does not increase
when the rules are fired.
\end{proof}

\section*{Proof of Lemma \ref{lem:boundingtw:2}}
Recall the statement:

\medskip

Let $\TD$ be a tree decomposition of a data instance $\inst$ of treewidth $k$ with $\tspan(\TD) \le 2$.
Let $\views$ be a set of connected CQ views, and $\views(\inst)$ the view image of $\inst$ under $\views$.
Let $r$ be the greatest radius of a CQ in $\views$.
Then the treewidth of $\views(\inst)$ is at most $k' = \frac{k(k^{r+1}-1)}{k-1}$.

\medskip

\begin{proof}
For a bag $b$ of $\TD$ and an integer $n$ define recursively its \emph{$n$-extension} by setting $ext(b,0) = b$ and $ext(b,n) = \{u \mid \exists v \in ext(b,n-1) \text{ such that }
u \text{ and } v \text{ belong to a same bag of } \TD\}$. Since $\tspan(\TD) \le 2$, it is easy to see by induction that $|ext(b,n)| \le k + k^2 + \dots + k^{n+1} =
\frac{k(k^{n+1}-1)}{k-1}$. Let $\TD'$ be a tree of bags whose set of nodes
is $ext(b,r)$, with an edge between $ext(b,r)$ and $ext(b',r)$ exactly when
there is an edge from $b$ to $b'$ in $T$.
We claim that $\TD'$ is a tree decomposition of $\views(\inst)$.

First we show that for any element $v$ the set of all bags in $\TD'$ containing $v$ is connected. Suppose that two nodes $n_1$ and $n_2$ of $\TD'$ contain $v$.
Then there are bags $b_1$ and $b_2$ in $\TD$ such that $n_1 = ext(b_1, r)$ and $n_2 = ext(b_2, r)$. Thus $v$ belong in some bags $b_1'$ and $b_2'$ which are
at most $r$ steps away from $b_1$ and $b_2$ respectively. Let $\pi$ be a unique simple path connecting $b_1$ and $b_2$ in $\TD$. Now we have a number cases depending 
on the length of $\pi$ and relative positions of $b_1$ and $b_2$ with respect to $\pi$ (see Figure~\ref{fig:cases}).
 
\begin{figure}
\scalebox{0.95}{\includegraphics{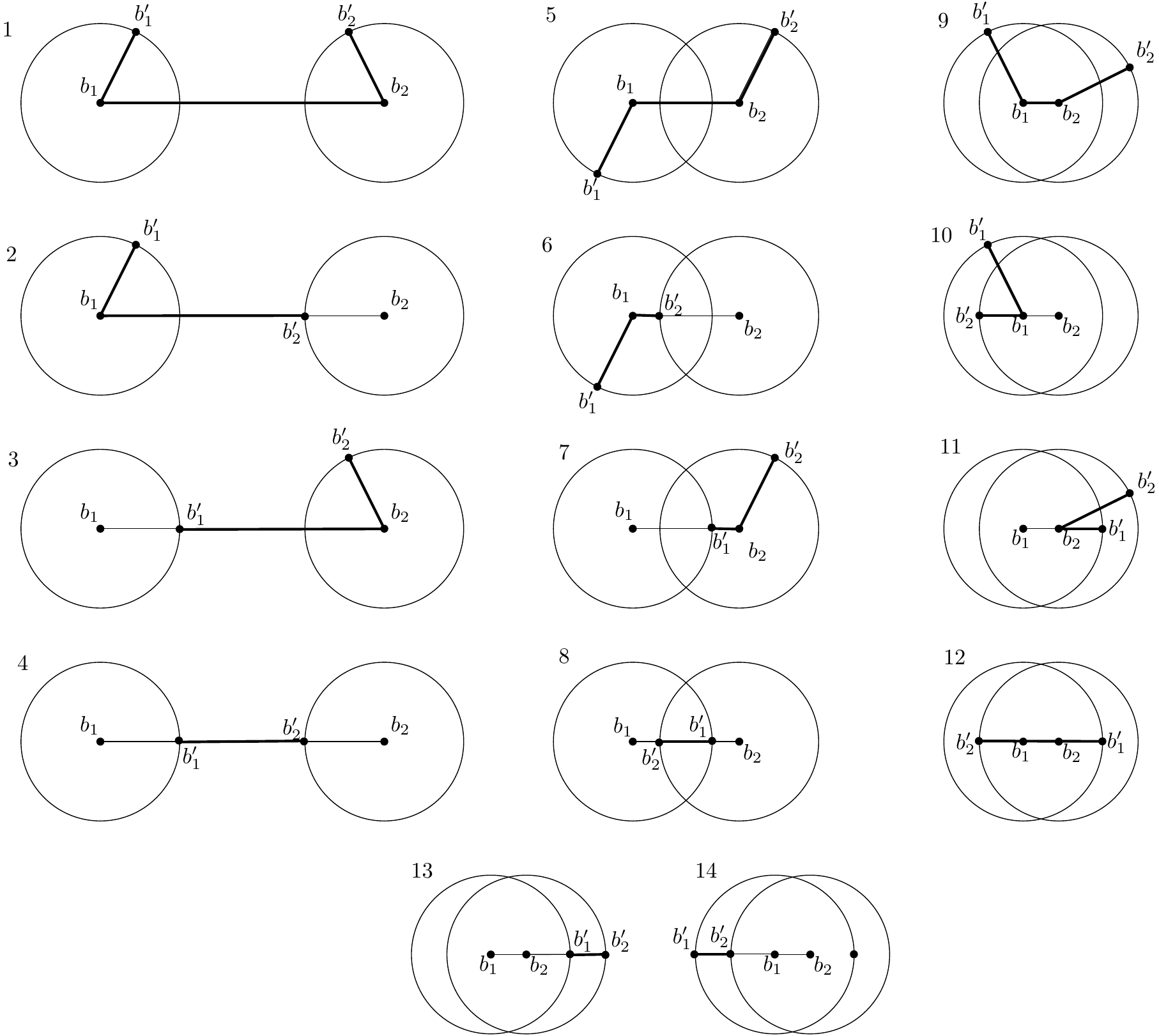}}
\caption{}
\label{fig:cases}
\end{figure}

In each of the cases we use the fact that $v$ must belong to all bags on a unique simple path between $b_1'$ and $b_2'$ (highlighted by bold lines) to conclude that
$v$ must also belong to all $r$-extensions of bags on $\pi$. For example, in Case 4, $v$ belongs to all bags between $b'_1$ and $b'_2$, and so to their extensions.
But also $v$ belongs to all $r$-extensions of bags between $b_1$ and $b'_1$, because all such bags are within distance $r$ from $b'_1$ and $v$ belongs to $b'_1$.
Similarly $v$ belongs to all $r$-extensions of all bags between $b'_2$ and $b_2$. In Case 14 since $v$ belongs to $b'_2$, by the same argument 
it follows that it belongs to  all $r$-extension of all bags between $b_1$ and $b_2$. Other cases are similar.

Secondly we show that for each atom $S(\vec{c})$ from $\views(\inst)$ there is a node in $\TD'$ containing $\vec{c}$. 
Suppose that $S(\vec{c})$ was 
generated by the view definition
$S(\vec{x}) \leftarrow \phi(\vec{x}, \vec{y})$ for a connected CQ $\phi(\vec{x}, \vec{y})$ with free variables $\vec{x}$ and quantified variables $\vec{y}$ under some assignment $\eta$ defined on both $\vec{x}$ and $\vec{y}$.
As $\phi$ is of radius at most $r$, 
it should have a variable $z \in \vec{x} \cup \vec{y}$ such that all other variables are at distance at most $r$ from $z$
in the Gaifman graph of $\phi(\vec{x}, \vec{y})$. Therefore the range of $\eta$ lies 
within distance $r$ from $\eta(z)$.  
Let $b$ be any bag of $\TD$ containing $\eta(z)$. If follows that $\vec{c}$ is contained in the $r$-extension of $b$. 

Finally, it is  easy to see that the sizes of bags of $\TD'$ are as required.
\end{proof}

\subsection*{Proof of Proposition \ref{prop:cqappr:regular} }
Recall the statement:

\medskip

 For any Datalog query $Q = (\Pi, \goal)$, there is an {\exptime} function that
outputs an NTA $\automaton_Q$ that captures the  set of canonical databases of CQ
approximations of $Q$. 

\medskip

\begin{proof}
Without any loss of generality we assume that all rules of $\Pi$ have
either $0$ or $2$ IDB atoms.

The states of $\automaton_Q$  will be rule heads of $\Pi$ paired
with an injective mapping $m$ from the head variables to $\{1, \dots, k\}$.
For example, if our state in node $v$ is $( P(x,y), \{x \mapsto 1, y \mapsto 3 \} )$
this means that we are looking for witnesses to the fact $P([v,1],[v,3] )$.

In a state $(U(\vec x), m )$ we non-deterministically choose a rule
body with the head $U(\vec x)$ and a consistent extension $m'$ of $m$ to all of the variables in the body.
Consistent here means that for every EDB atom $R(\vec{y})$ in the body of the rule, the unary predicate
$\coderel^R_{m(\vec x)}$ is in the label of the current node, and conversely each atom in the label
of the current node corresponds to some EDB atom.

Now let's turn to intensional predicates. If there are no intensional atoms in the body, we accept.
If $F_1(\vec{y}_1)$ and $F_2(\vec{y}_2)$ are the intensional atoms in
the rule body in some canonical order, then we have a transition which for $i = 1, 2$
goes to the  $i$-{th} child of the current node and switches the state into $(F_i(\vec{y}_i), m_i)$, where
$m_i$ is the restriction of $m$ to $\vec{y}_i$. We also check that the edge label leading to the $i$-th child is the restriction of the identity
map on $\{1, \dots, k\}$ to the image of $m_i$.

To see that the conditions for capturing hold, note that all CQ approximations of $Q$ have a standard tree
decomposition of width $k$, where there is one-to-one correspondence between bags and rule bodies.
Therefore we have a standard $k$-code, where variables in each rule body are ordered in such a way that
common variables in two adjacent bags occur in exactly same positions. 
This code is accepted by $\automaton_Q$, which gives the second requirement for capturing.  And all codes
that are accepted by $\automaton_Q$ are one of these standard codes, which gives the first required
property of capturing.
\end{proof}

\subsection*{Proof of Proposition \ref{prop:datalog} }

Recall the statement

\medskip

For any Datalog program $\datalogprog$, the class $\{\F \mid \F \models \datalogprog, \tw(\F) \le k \}$ (here $\F$ are finite instances
 which contain both EDBs and IDBs of $\Pi$) is $k$-regular and is recognized by an NTA of at most double-exponential size in $k$ and single-exponential size in $|\Pi|$.

\medskip

\begin{proof}
First we construct a \emph{two-way alternating tree automaton} which, for each of the rules of $\Pi$ of the form $R(\vec x) \datalogarrow \phi(\vec x, \vec y)$,
guesses non-deterministically moving in both directions the valuations
$\vec a$ and $\vec b$ of $\vec x$ and $\vec y$, respectively, and then checks whether 
$\neg \phi(\vec a, \vec b) \lor R(\vec a)$ holds. Note that its size is linear in $\Pi$ and single-exponential in $k$. Then using Theorem A.2
of \cite{cosmadakis1988decidable} we convert it into an NTA with an  exponential blow-up.
\end{proof}

\subsection*{Proof of Proposition \ref{prop:restriction}}
Recall the statement:

\medskip

If $\C$ is a $k$-regular class in $\Sigma$ captured by NTA $\automaton$ 
and $\Sigma' \subseteq \Sigma$, then 
the class $$\C\restr \Sigma' = \{\F\restr \Sigma' \mid \F \in \C\}$$ is also $k$-regular and is captured by an automaton of size at most $|\automaton|$. The same holds with ``capture'' replaced by ``recognize''.

\medskip

\begin{proof}
We prove only the first part, with the second part being similar.

Consider an automaton $\automaton$ for $\C$. It has transitions of the form $q_1,q_2,\sigma_L^{s_1, s_2} \to q$. Let $\automaton'$ have the
same states and accepting states as $\automaton$, the alphabet $\sigma_{L'}^{s_1, s_2}$ with $L' \subseteq \Sigma'$, and the transition
table 
$\{q_1,q_2,\sigma_{L\,\restr\, \Sigma'}^{s_1, s_2} \to q \mid q_1,q_2,\sigma_L^{s_1, s_2} \to q \text{\textrm{} is a transition of } \automaton\}$. 
We claim that
$\automaton'$ is of required size and captures $\C\restr \Sigma'$. Indeed, take $\F' \in \C \restr \Sigma'$. Then there is
$\F \in \C$ such that $\F' = \F \restr \Sigma'$. As $\automaton$ captures $\C$, there is a code $\T$ of $\F$ such that $\automaton$ accepts $\T$.
From the definition of $\automaton'$ it follows that $\automaton'$ accepts $\T \restr\Sigma'$ which is a code of $\F'$.
And the other way round, if $\automaton'$ accepts $\T'$ in $\Sigma'$ via a run $f'$, then there is a run $f$ of $\automaton$ which accepts 
some extension $\T$ of $\T'$. As $\automaton$ captures $\C$, it follows that there is a database instance $\F \in \C$ such that $\T$ is a code of $\F$, and so 
$\T'$ is a code of $\F \restr \Sigma'$.  Thus $\decode{\T'} \in \C \restr \Sigma'$.
\end{proof}

\subsection*{Proof of Lemma \ref{lem:homdet}}
Recall the statement:

\medskip

 For any Datalog query $Q$ and Datalog views $\views$,  if $Q$ is monotonically determined
 over $\views$ then it is homomorphically determined over $\views$.

\medskip

\begin{proof}
Assume monotonic determinacy and consider  instances  $\inst_1, \inst_2$ with $\inst_1 \models Q$, and a homomorphism $h$ from 
$\views(\inst_1)$ into  $\views(\inst_2)$.

It follows that there is a CQ $Q_i$ that is an approximation of $Q$, and  a homomorphism $\alpha$ from $Q_i$ into $\inst_1$.
Note that $\alpha$ is also a homomorphism from $\views(Q_i)$ into $\views(\inst_1)$.
Thus $\alpha$ followed by $h$, denoted $h(\alpha)$, is a homomorphism from $\views(Q_i)$ into $\views(\inst_2)$.

We create an instance $\inst'$ such that $\views(Q_i) \subseteq \views(\inst')$, along with a
homomorphism $h'$ taking $\inst'$ into $\inst_2$. We will construct
 $\inst'$ as the union of a set of facts $S_F$
obtained by chasing each fact $F$ in $\views(Q_i)$ with the  inverse of the
view definitions 
(see also the proof of Lemma \ref{lem:pipeline}). More precisely,
consider a fact $V(a_1 \ldots a_n)$ in $\views(Q_i)$. Then
$V(h(\alpha)(a_1) \ldots h(\alpha)(a_n))$ is in $\views(\inst')$, and
thus there is some CQ approximation $\rho$ of $Q_V$ 
such that $\rho(h(\alpha)(a_1) \ldots h(\alpha)(a_n))$ holds
in $\inst'$.
We let $S_F$ be obtained from $\rho(h(\alpha)(a_1) \ldots h(\alpha)(a_n))$
by replacing each $h(\alpha)(a_i)$  with
$a_i$ and each existentially quantified variable with a fresh null.
One can easily check  that
$\inst'$ is as required.
Thus by monotonic determinacy we have  $\inst' \models Q$. But since there is a homomorphism of $\inst'$ into
$\inst_2$, we conclude that $\inst_2 \models Q$ as required.
\end{proof}

\subsection*{Jointly-annotated terms}

We recall the definition of the backward mapping query $Q_\automaton = (\Pi_\automaton, \goal_\automaton)$:

For every transition of the form $q_1, q_2, \sigma_{L}^{s_1,s_2} \to q$ with
$L = \{\coderel^{R^1}_{\vec{n}_1}, \dots, \coderel^{R^m}_{\vec{n}_m}\}$ we create
a rule
\begin{equation} \label{rule-i}
P_q(x_1, \dots, x_k) \datalogarrow \\
\bigwedge_{i = 1}^{k} \Adompred(x_i) \land P_{q_1}(x_1^1, \dots, x_k^1)\land P_{q_2}(x_1^2, \dots, x_k^2)    \land \!\!
\bigwedge_{i \in \dom(s_1)} \!\!\! x_i = x^1_{s_1(i)}\land \!\!\! \bigwedge_{i \in \dom(s_2)} \!\!\! x_i = x^2_{s_2(i)} \land  \bigwedge_{l = 1}^m R^l(\vec{x}_{\vec{n}_l})
\end{equation}
where $j$ ranges over  $1$ and $2$,  $x_i^j$ are fresh variables for indices $j \in \{1,2\}$ and $i \in \{1, \dots, k\}$, and for $\vec{n}_l = (n_l^1, \dots, n_l^d)$
we have $\vec{x}_{\vec{n}_l} = (x_{n_l^1}, \dots, x_{n_l^d})$.
For initial transitions of the form $\sigma_{L} \to q$ with $L = \{\coderel^{R^1}_{\vec{n}_1}, \dots, \coderel^{R^m}_{\vec{n}_m}\}$ we have rules
\begin{equation} \label{rule-ii}
P_q(x_1, \dots, x_k) \datalogarrow\bigwedge_{i = 1}^{k} \Adompred(x_i) \land \bigwedge_{l = 1}^m R^l(\vec{x}_{\vec{n}_l}).
\end{equation}

For accepting states $q$ we add the rules $\goal_{\automaton}(x_1, \dots, x_{k}) \datalogarrow P_q(x_1, \dots, x_k)$
for the goal predicate $\goal_{\automaton}$. We also add a standard set of rules which, when evaluated on any data instance $\inst$
under fixed-point semantics, guarantee that the interpretation of the 
IDB $\Adompred(x)$ is the active domain of $\inst$.

\myparagraph{Proof terms and annotated proof terms}
To show correctness of the backward mapping construction (Proposition \ref{prop:backward})
we will need the notion of a ``proof certificate''
for backward mappings of an automaton.

When a Datalog query $Q = (\Pi, \goal)$ holds for a tuple $\vec d$ in an instance
$\inst$, there is a derivation that witnesses this, which has a tree-like structure.
A \emph{proof term} for $\inst \models Q(\vec d)$ is a labelled finite tree in which  every 
node $\node$ is labelled with a ground fact $\factof(\node)$ over the predicates mentioned in $\Pi$, and
every non-leaf node $\node$ is additionally labelled with a rule $\ruleof(\node)$ of $\Pi$ such that:
\begin{compactitem}
\item If $\node$ is the root, $\factof(\node)= \goalpred(\vec d)$
\item If $\node$ is a leaf then $\factof(\node)$ is a fact over the extensional predicates
of $Q$, and this fact holds in $\inst$
\item If $\node$ is not a leaf,  let $\inst_{\node}$ be the instance consisting
of $\factof(\node)$ and all facts $\factof(c)$ for $c$ a child of $\node$. 
Then there is a map $h_{\node}$ from the variables in the body 
of $\ruleof(\node)$ into the active domain of $\inst_{\node}$ that maps the facts
in the body of $\ruleof(\node)$ onto the facts of $\inst_{\node}$, and
 maps the head of $\ruleof(\node)$ to $\factof(\node)$.
\end{compactitem}

It is well-known \cite{AHV} and easy to see that proof terms represent a semantics for Datalog:
 $\inst \models Q(\vec d)$ exactly when there is a proof term that
witnesses this.

We now give a notion of a witness for acceptance of an automaton running over codes.
 A \emph{jointly-annotated term} for automaton $\automaton$, instance $\inst$,
and $k$-tuple $\vec{a}$ is a pair $(\T, \vec{b})$ where
\begin{itemize}
\item[--] $\T$ is a tree code accepted by $\automaton$;
\item[--] the map $\vec{b}$ assigns  each vertex of $\T$ to a $k$-tuple of elements from $\inst$, with
 the root of $\T$ mapped to $\vec{a}$;
\end{itemize}
which satisfy the following condition: {\textbf if}
$t_v = \sigma^{s_1, s_2}_L(t_{v_1}, t_{v_2})$ with
$L = \{\coderel^{R^1}_{\vec{n}_1}, \dots, \coderel^{R^m}_{\vec{n}_m}\}$,
$\vec{b}(v) = (b_1, \dots, b_k)$, $\vec{b}(v_j) = (b^j_1, \dots, b^j_k)$ for $j = 1,2$
{\textbf then}
\begin{equation}\label{rule-iii}
\inst\models
\bigwedge_{j = 1}^2 \hspace{4mm} \bigwedge_{i \in \dom(s_j)} \!\!\! b_i = b^j_{s_j(i)} \,\land \bigwedge_{l = 1}^m R^l(\vec{b}_{\vec{n}_l}).
\end{equation}

We also require that 

\begin{equation}\label{rule-iv}
\mbox{
{\textbf if} $t = \sigma_L$ is a leaf symbol in $T$ with
$L = \{R^1_{\vec{n}_1}, \dots, R^m_{\vec{n}_m}\}$
 {\textbf then }
the atoms $R^l(\vec{b}_{\vec{n}_l})$ are in $\inst$ for $l = 1, \dots, m$.}
\end{equation}

In other words, 
$\vec{b}$ can be considered as a homomorphism from $\decode T$ into $\inst$.

We  now verify the key  property of a jointly-annotated term:
\begin{proposition}\label{jointly-annotated-terms}
 For each data instance $\inst$, $\inst, \Pi_\automaton \models \goal_\automaton(\vec a)$ if
and only if there is a jointly-annotated term for
$\automaton$,$\inst$, and $\vec{a}$.\\
\end{proposition}

\begin{proof}
We prove the two directions of the if and only if separately.

($\Rightarrow$)
Take a proof term $\proofterm$ that witnesses $\inst, \Pi_\automaton \models \goal_\automaton(\vec a)$. 
We transform $\proofterm$ into a jointly-annotated term $(\T, \vec{b})$ on the set of all vertices of 
$\proofterm$
with $\factof(\node)$ being an IDB. Note that this gives us a binary tree since all rule bodies in 
$\Pi_\automaton$ have either $0$
or $2$ IDBs by assumption.
For each vertex $v$ we take some ordering $(u_1^\node, \dots u_k^\node)$ of elements in $\inst_\node$ 
without duplicates;
we use fresh dummy elements to fill up the tuple if $\inst_\node$ has less then $k$ elements. Now 
we define  unary
labels of $\T$ by setting $\coderel^R_{n_1, \dots, n_m}(\node) \in \T$  iff 
$R(u^\node_{n_1}, \dots, u^\node_{ n_m}) \in \inst_\node$.
We define edge labels $s$ between a parent $\node$ and its child $w$ by setting $s(n) = m$ if $u^\node_n$ is 
the same
element as $u^w_m$; it should be clear that $s$ is a partial bijection. This constitutes the definition of  
$\T$.
It remains to define $\vec b$ by setting $\vec b (\node)$ to be $(u_1^\node, \dots u_k^\node)$.
We can create an accepting run
  $f$ by setting
$f(\node)$ to be the state of the automaton $q$ such that $\factof(\node)$ is labelled by and IDB $P_q$.

($\Leftarrow$) It is easy to show by induction that if $v$ is a vertex of a jointly annotated-term
$(\T, \vec{b})$ for $\automaton$, $\inst$, and $\vec{a}$ and  $f$ is an accepting
run for $\automaton$ on $\T$ with $f(v) = q$, then $\inst, \Pi_\automaton \models P_q(\vec b(v))$. It follows that
$\inst, \Pi_\automaton \models \goal_\automaton(\vec a)$.

Indeed, if $v$ is a leaf, then then $\inst, \Pi_\automaton \models P_q(\vec b(v))$ by the rule 
(\ref{rule-ii}) because its body holds due to condition (\ref{rule-iv}) and the fact that for all $u$ in $\vec{b}(v)$
we have $\inst, \Pi_\automaton \models \Adompred (u)$.

If $v$ has children $v_1$ and $v_2$, then there must be $q_1$ and $q_2$ such that 
$f(v_1) = q_1$, $f(v_2) = q_2$, production $(q_1, q_2, \sigma^{s_1, s_2}_L \to q)$ is a transition of $\automaton$,
and the vertex label of $v$ is $L$ while edge labels between $v$, $v_1$ and $v_2$ are $s_1$ and $s_2$.

We claim that $\inst, \Pi_\automaton \models  P_q(\vec b(v))$ can be inferred by the rule (\ref{rule-i}) for this production
under assignment $\{(x_1, \dots, x_k):= \vec{b}(v), (x^1_1, \dots, x^1_k):= \vec{b}(v_1), (x^2_1, \dots, x^2_k):= \vec{b}(v_2) \}$. 
Indeed, we have  $\inst, \Pi_\automaton \models \Adompred (u)$ for
all elements in the body of the rule, we have  
$\inst, \Pi_\automaton \models P_{q_1}(\vec b(v_1))$ and
$\inst, \Pi_\automaton \models P_{q_2}(\vec b(v_2))$ by the induction hypothesis, and the rest of the rule by (\ref{rule-iii}). 

\end{proof}

\subsection*{Proof of Proposition \ref{prop:backward}}

Recall the statement:

\medskip

Let $Q$ be homomorphically determined over $\views$ and $\automaton$ be any automaton 
working on $k$-codes such that
$\{\views(Q_i) \mid i \in \omega\} \subseteq \decode{L(\automaton)} \subseteq \{\D \mid \views(Q_i)\text{\textrm{} maps into } \D \text{\textrm{} for some } i \in \omega\}$. More precisely,
we require that
\begin{compactitem}
\item[(1)] for each CQ approximation $Q_i$ of $Q$ there is a code $\T$ such that $\decode\T = \views(Q_i)$ and $\T$ is accepted by $\automaton$;
\item[(2)] for each tree code $\T$ accepted by $\automaton$ there is a CQ approximation $Q_i$ of $Q$ and a homomorphism from
$\views(Q_i)$ into $\decode\T$.
\end{compactitem}

Then
for each data instance $\inst$ we have $\inst \models Q$ iff 
$\views(\inst) \models Q_{\automaton}(\vec{a})$ for some $\vec{a} \in \adom(\inst)^k$.

\medskip

\begin{proof}
Suppose $\Pi$ is a Datalog program containing
intensional predicate  $A$,
 $\inst$ is an instance for the extensional (input) signature of
$\Pi$,  and $\vec a$ is a tuple of elements from $\inst$. Below we write 
\[
\inst, \Pi \models A(\vec a)
\]
to indicate that the least fixpoint of $\Pi$ on $\inst$ contains
$A(\vec a)$.

$(\Rightarrow)$  Suppose that  $\inst \models Q$.  Then there is an approximation $Q_i$ of $Q$ and a homomorphism $h$
from $Q_i$ into $\inst$, which is also a homomorphism from $\views(Q_i)$ to $\views(\inst)$. 
As $Q_i$ is an approximation of $Q$, by the first inclusion for $L(\automaton)$, $\automaton$ must accept some code 
$\T$ of $\views(Q_i)$. Choose an arbitrary element $e_0$ from $\adom(\views(\inst))$.  
For a vertex $v$ of $\T$ we define $\vec{b}(v)$ to be the tuple $(e_1, \dots, e_k)$ of elements of $\views(\inst)$
where each $e_i$ is defined as follows:
$$
e_i = \begin{cases} h([v,i]), \mbox{ if } [v,i] \in \adom(\inst) & \\ 
                    e_0, \mbox{otherwise.} 
       \end{cases}
$$
We claim that $(\T, \vec{b})$ is a jointly-annotated term for $\automaton$, $\views(\inst)$ and the $\vec{b}$-image 
of the root of $\T$. Indeed, if equation (\ref{rule-iii}) contains an equality $[v,i] = [u,j]$, it follows that
$[v,i]$ and $[u,j]$ are indeed equivalent. The R-atoms of equations (\ref{rule-iii}) and (\ref{rule-iv}) hold because
$h$ is a homomorphism, and also because they are never applied to dummies.
 It follows (by Proposition~\ref{jointly-annotated-terms}) that 
$\views(\inst) \models Q_\automaton(\vec{a})$ for some $\vec{a}$.

$(\Leftarrow)$ Suppose that $\views(\inst), \Pi_{\automaton} \models \goal_{\automaton}(\vec{a})$. 
Let $(\T, \vec{b})$ be a jointly-annotated term for the inference of $\goal_\automaton(\vec{a})$ for 
$\views(\inst), \Pi_{\automaton}$ and $\vec{a}$ (which exists by Proposition \ref{jointly-annotated-terms}),
and $f$ be an accepting run of $\automaton$ on $\T$.
Thus, by the second inclusion for $L(\automaton)$, there must be a homomorphism $h$ from 
$\views(Q_i)$ for some $i$ into $\decode T$. 
Note that by Proposition~\ref{jointly-annotated-terms},
we know that $\vec b$ can be considered as a homomorphism from $\decode T$ into $\views(\inst)$.
By composing $h$ with $\vec b$,
we obtain a homomorphism $g$ from
$\views(Q_i)$ into $\views(\inst)$.
Now we have a data instance $\inst' = Q_i$ such that $\inst' \models Q$ and 
a homomorphism $g$ from $\views(\inst')$ into $\views(\inst)$.
Therefore, as $Q$ is homomorphically determined by $\views$, we have $\inst \models Q$.
\end{proof}

\newpage
\section*{Proofs for Section \ref{sec:rewrite}: rewritability results}
\subsection*{Proof of the last part  of Theorem \ref{thm:frgd-rewriting}}
Recall that Theorem \ref{thm:frgd-rewriting} stated that
if $Q$ is in MDL,   $\views$ are a collection of $\fgdatalog$ views,
and $Q$ monotonically determined by $\views$
then $Q$ has a rewriting in MDL.
We sketch how to modify the prior argument for this claim.

A tree decomposition is \emph{frontier-one} if the intersection of any
two neighboring bags has at most one element. It is clear that approximations
of MDL queries have such decompositions, provide that we now
allow decompositions that have arbitrary outdegree, not necessarily
binary. We can further normalize so that in each bag other than
the root, the element that is shared with its parent (if such
exists) has the first local name in the code.

When we apply frontier-guarded  views, we annotate
the bags of the tree decomposition with view predicates, but we do not change the intersection
of neighboring bags. And when we project such a decomposition onto the view predicates,
we do not change this intersection either.
Thus in the proof of Theorem \ref{thm:frgd-rewriting}, we can consider
an automaton $\automaton$ that enforces the frontier-one restriction.

We can modify the backward mapping for frontier-one decompositions so that
it produces an MDL query; our modification will have only unary intensional predicates $P_q$ for
each state $q$ of the automaton, corresponding only to the element coded in the frontier. 

More formally,
for every transition of the form $q_1 \ldots q_r \sigma_{L}^{s_1,\ldots s_r } \to q$ with
$L = \{\coderel^{R^1}_{\vec{n}_1}, \dots, \coderel^{R^m}_{\vec{n}_m}\}$,
we know that for the $i^{th}$ child node, the label $L$
contains at most one equality of a local name $n_i$ with the
first local  name of the child.

We create a rule of the form:
\begin{align*}
P_q(x_1) \datalogarrow
\bigwedge_{i = 1}^{k} \Adompred(x_i) \wedge
\bigwedge_{i=1}^{r} P_{q_i}(x_{n_i}) \wedge 
\bigwedge_{l = 1}^m R^l(\vec{x}_{\vec{n}_l})
\end{align*}

Similar modifications are applied to the leaf rules.

\subsection*{Proof of Theorem \ref{thm:moncq-rewriting}}
Recall the statement:

\medskip

Suppose $Q$ is a normalized Monadic Datalog
query and $\views$ is a collection of Monadic Datalog and 
CQ views. 
If $Q$ is monotonically determined by $\views$, then
$Q$ is rewritable over $\views$ in Datalog. The size of the rewriting is at most double-exponential in $K = O(|Q|^{|\views|})$ (``of required size'' below).

\medskip

\begin{proof}
We first argue that without any loss of generality 
we can assume that all CQ views are connected.
If $V$ is a CQ view which is not connected, then it can be replaced by
a few connected CQs. For example, the disconnected view
$V(\vec{x},\vec{y}) = Q_1(\vec{x})\land Q_2(\vec{y})$ can be replaced by the free-variable-connected
views $V_1(\vec{x}) = Q_1(\vec{x})\land \exists \vec{y}\, Q_2(\vec{y})$ and
$V_2(\vec{y}) = (~ \exists \vec{x} \,Q_1(\vec{x})~ )\land  Q_2(\vec{y})$. Indeed, given $V$, we can restore
$V_1$ and $V_2$ as its projections on $\vec{x}$ and $\vec{y}$ respectively.
And the other way round, given $V_1$ and $V_2$, we
can restore $V$ as their product since
$V_1(\vec{x})\land V_2(\vec{y}) = Q_1(\vec{x})\land (\exists \vec{y}\, Q_2(\vec{y}))\land(\exists \vec{x} \,Q_1(\vec{x}))\land  Q_2(\vec{y})$
is equivalent in first-order logic to $Q_1(\vec{x})\land Q_2(\vec{y}) = V(\vec{x},\vec{y})$.

We need to show that there is an automaton $\automaton$ such that
\begin{align*}
\{\views(Q_i) \mid i \in \omega\} \subseteq \decode{L(\automaton)} \subseteq
\{\jnst \mid \views(Q_i)\text{\textrm{} maps into } \jnst \text{\textrm{} for some } i \in \omega\}
\end{align*}

Consider the class $\C$ of canonical databases of CQ approximations of $Q$.
By Lemma \ref{lem:boundingtw:2} (applied with the maximal radius $r$ of the CQ views in $\views$ where  $r= O(|V|)$), the treewidth of the class $\views(\C) = \{ \views(\F) \mid \F \in \C\}$
of view images of $\C$ is also bounded by some $K = O(|Q|^{|\views|})$.
We can strengthen Proposition~\ref{prop:cqappr:regular} to show that
for any treewidth $K$ greater than or equal to the maximal number of variables
in the rules of $Q$, the class  
$\C$ of approximations is $K$-regular and there is an 
NTA $\automaton^{\mathsf{base}}$ of at most exponential size in $K$ that
captures $\C$.

Without any loss of generality we assume that the sets of IDBs of programs for different views are disjoint, and that their goal predicates are
identical with the view predicates. Denote by $\Pi_{\views}$ the union of all rules in Datalog queries in $\views$, including the rules for the CQ views.
By Proposition~\ref{prop:datalog}, there is an NTA $\automaton^{\Pi_\views}$ of required size which
recognizes all codes of $\{\F \mid \F \models \Pi_{\views}, \text{tw}(\F) \le K \}$.

We claim that the automaton $\automaton' = \automaton^{\mathsf{base}} \cap \automaton^{\Pi_\views}
$ satisfies
\begin{equation*}
\{\fpeval{\Pi_{\views}}{Q_i} \mid i \in \omega\} \subseteq \decode{L(\automaton')} \subseteq \{\F \mid \F \restr \Sigma \in \C, \F \models \Pi_{\views} \}
\end{equation*}

Now the automation $\automaton$ that is the projection of $\automaton'$ on the signature of view predicates
(which exists by Proposition~\ref{prop:restriction}) captures
$\mathbb{V} = \{ \F \restr \Sigma_\views \mid \F \restr \Sigma_{\baseschema}  \in \C, \tw(\F) \le K, \F \models \Pi_\views\}$
and so satisfies two conditions of Proposition \ref{prop:backward}.
Now applying  Proposition \ref{prop:backward}, we conclude that
$Q$ is Datalog rewritable over views, and that the rewriting is of required size.

Another observation  will be useful later (see proof of Theorem~\ref{thm:decidemondetrewritingmdlandcqs}) is
that
$\automaton$ captures $\mathbb{V} = \{ \F \restr \Sigma_\views \mid \F \restr \Sigma_{\baseschema}  \in \C, \tw(\F) \le K, \F \models \Pi_\views\}$.

Now applying  Proposition \ref{prop:backward}, we conclude that
$Q$ is Datalog rewritable over views, and that the rewriting is of required size.
\end{proof}

\newpage
\section*{Proofs for Section \ref{sec:decide}: decidability results on monotonic determinacy}
\subsection*{Proof of Lemma \ref{lem:pipeline}}
Recall the statement:

\medskip

$Q$ is monotonically determined over $\views$ if and only if
every test succeeds.

\medskip 

We first need a bit of infrastructure.
When a Datalog query $Q = (\Pi, \goal)$ holds for a tuple $\vec d$ in an instance
$\inst$, there is a derivation that witnesses this, which has a tree-like structure.
A \emph{proof term} for $\inst \models Q(\vec d)$ is a labelled finite tree in which  every 
node $\node$ is labelled with a ground fact $\factof(\node)$ over the predicates mentioned in $\Pi$, and
every non-leaf node $\node$ is additionally labelled with a rule $\ruleof(\node)$ of $\Pi$ such that:
\begin{compactitem}
\item If $\node$ is the root, $\factof(\node)= \goalpred(\vec d)$
\item If $\node$ is a leaf then $\factof(\node)$ is a fact over the extensional predicates
of $Q$, and this fact holds in $\inst$
\item If $\node$ is not a leaf,  let $\inst_{\node}$ be the instance consisting
of $\factof(\node)$ and all facts $\factof(c)$ for $c$ a child of $\node$. 
Then there is a map $h_{\node}$ from the variables in the body 
of $\ruleof(\node)$ into the active domain of $\inst_{\node}$ that maps the facts
in the body of $\ruleof(\node)$ onto the facts of $\inst_{\node}$, and
 maps the head of $\ruleof(\node)$ to $\factof(\node)$.
\end{compactitem}

It is well-known \cite{AHV} and easy to see that proof terms represent a semantics for Datalog:
 $\inst \models Q(\vec d)$ exactly when there is a proof term that
witnesses this.


We are now ready for the proof of the lemma.
\begin{proof}
We assume $Q$ is Boolean for simplicity.
In one direction, assume $Q$ is monotonically determined over $\views$, and consider
a test $(Q_i, D')$.  By virtue
of $(Q_i, D')$ being a test, we have  $\views(D) \subseteq \views(D')$.
Monotonic determinacy and $Q_i \models Q$ thus imply that $D' \models Q$.

In the other direction, assume every test succeeds, and consider instance
$\inst_1$ and $\inst_2$ with $\inst_1$ satisfying $Q$ and $\views(\inst_1) \subseteq \views(\inst_2)$.
As $\inst_1 \models Q$, there is a homomorphism $\alpha$ from some $Q_i$ into $\inst_1$.
 Since the views are preserved under homomorphism, $\alpha$ is also a homomorphism from
$\views(Q_i)$ into $\views(\inst_1)$.

We will now create a $D'$ such that $(Q_i, D')$
forms a test, along with an
 extension of $\alpha$ that is a homomorphism taking $D'$ into $\inst_2$. 
$D'$ will be the union of a set of facts 
$S_F$ (defined below)  for every fact $F$ from $\views(Q_i)$.
For a fact $F= V(\vec c)$ from $\views(Q_i)$
let $F'=\alpha(F)$. Note that $F'$ is in $\views(\inst_1)$.
By assumption,  $F'$ is also in $\views(\inst_2)$.
Thus there is a proof term $\tau_{F'}$  witnessing
that  $\inst_2 \models F'$. 
Moving top-down on $\tau_{F'}$, we form a proof term
for $F$.  The root of the term $\tau_{F'}$ is labelled with the fact $\goal_V(\alpha(c_1) \ldots \alpha(c_n))$ for the goal predicate
$\goal_V$ of the Datalog program $Q_V$. Since $\alpha$ is not injective,
$c_i$ may not be unique, but we  choose one such tuple $c_1 \ldots c_n$ and fix it for the transformation
of $\tau_{F'}$. This choice will impact the proof term that we create, but will
not impact the homomorphism extending $\alpha$.
We first
transform $\tau_{F'}$ by replacing any element $\alpha(c_i)$ 
occurring in $\tau_{F'}$ by $c_i$.
We then  continue our transformation by proceeding top-down on the partially-transformed
term. At the root of the term we do nothing more.
In the inductive step, we consider an intensional fact
$U(\vec d)$ in $\tau_F$ witnessed by a set of facts $J$ that are a substitution instance of some rule body $B$.
In $J$, we uniformly
replace any witness $w$ to an existentially quantified variable $x$ of $B$ by a fresh element $d_w$, and extend
the homomorphism to take $d_w$ to $w$. We set $S_F$ to be the union of all EDB facts
occurring in the proof term we have constructed for $F$.

It is easy to see that the union of the facts $S_F$ forms an appropriate $D'$ giving a test.
By assumption this test succeeds, so $D' \models Q$. But since $D'$ is homomorphically
embedded into $\inst_2$, this means that $\inst_2 \models Q$ as required.
\end{proof}

\subsection*{Proof of Theorem \ref{decidability:cq:equivalence}}
Recall the statement:

\medskip

If $Q$ is a CQ and $\views$ is a collection of Datalog views, then  the problem of monotonic determinacy
of $Q$ over $\views$ is decidable in $\twoexp$.

\medskip

\begin{proof}
Let $Q' = \views(Q)$ and let $Q'$ inherit all answer variables $\vec{x}$ from $Q$.
Let $Q'' = (\Pi, \goal)$ where $\Pi$ is obtained by taking all rules defining $\views$
and adding the rule $\goal(\vec{x}) \datalogarrow Q'$

It is easy to see that
the following statements are equivalent:
\begin{itemize}
\item[(1)] $Q$ is monotonically determined by $\views$;
\item[(2)] $Q'$ is a CQ rewriting of $Q$ in terms of $\views$;
\item[(3)] for all $\inst$, $\inst \models Q$ iff $\views(\inst) \models Q' $;
\item[(4)] $Q''$ is equivalent to $Q$;
\item[(5)] $Q''$ is contained in $Q$.
\end{itemize}

Indeed, (1) implies (2) by the proof of Proposition~ \ref{prop:cqquery},
 and all other
implications between adjacent statements are trivial.
It remains to note that the containment (5) can be decided in $\twoexp$
by Theorem 5.12 of \cite{chaudhuri1997equivalence}.
\end{proof}

\subsection*{Proof of Theorem  \ref{thm:decidemondetrewritingmdlandcqs}}
Recall the statement:

\medskip

Suppose $Q$ is in Monadic Datalog,
and $\views$ is a collection of CQ and Frontier-guarded Datalog views.
Then there is an algorithm that decides if  $Q$ is monotonically determined by $\views$ in $\threeexp$.

\medskip

\begin{proof} 
In this proof the words ``of required size'' mean ``doubly-exponential in $K$'' where $K$ is some integer defined below,
$\C$ stands for the class of all CQ approximations of $Q$, $\Sigma_\views$ is the view signature
and $\Sigma_{\baseschema}$ is the initial signature.

As in the proof of Theorem \ref{thm:moncq-rewriting}, we
we can assume that all CQ views are connected.

We have to check whether $Q$ holds on all tests. As observed in the proof of Theorem~\ref{thm:moncq-rewriting}, there is an integer $K = O(|Q|^{|\views|})$
such that both the treewidth of the view images of CQ approximations of $Q$ and the treewidth of the CQ approximations of the views in $\views$ are at most $K$. 
Let $\mathbb{V} = \{ \F \restr \Sigma_\views \mid \F \restr \Sigma_{\baseschema}  \in \C, 
\tw(\F) \le K, \F \models \Pi_\views\}$.
As argued in the proof of Theorem~\ref{thm:moncq-rewriting}, $\mathbb{V}$ is $K$-regular and captured by an NTA $\automaton_\views$ of required size.

We follow the template of 
 Theorem ~\ref{thm:moncq-rewriting}
We will check the equivalent condition that $Q$ holds on each element of the class
$ETEST(Q, \views)$,  which consists of all instances $D'$ which can be obtained from an instance in $\mathbb{V}$ by applying inverses of view definitions
while keeping the atoms of the view signature. Note that the treewidth of all database instances in $ETEST(Q, \views)$ is also bounded by $K$. By Proposition~\ref{prop:cqappr:regular},
for each view  $(V,Q_V)$ there exists an automaton $\automaton'_{V}$ which for each atom $V(\vec c)$ at a node $n$ in $\T$ checks whether
$n$  has a descendant $n'$ such that $n'$ contains $\vec{c}$ and the subcode of $\T$ rooted at $n'$ is a code of some CQ approximation of $Q_V$.
The automaton $\automaton_{ETEST}$ defined as the product of $\automaton_\views$ and $\automaton'_{V}$ for all views $V$ in $\views$ (thus  accepting the intersection of these languages)
captures $ETEST(Q, \views)$.

By Proposition~\ref{prop:datalog:notaccept}, there is an NTA $\automaton''$ of required size 
which recognizes those codes which do not satisfy $Q$. 
So to check if $Q$ is monotonically determined by $\views$ we construct the intersection 
of $\automaton_{ETEST}$ and $\automaton''$ (which is of required size) and check if it is empty. The latter check is linear in the size of the automaton.
It should be clear that the time complexity of this procedure is doubly exponential in $K$, and so triply exponential in the size of the input.
\end{proof}

\newpage
\section*{Proofs for Section \ref{sec:lower}: lower bounds on monotonic determinacy}
\subsection*{Proof of Proposition \ref{prop:easyhardness}}
Recall the statement:

Monotonic determinacy is
\begin{compactitem}
\item $\np$-hard for CQ queries and views 
\item $\Pi_2^p$-hard for UCQ queries and UCQ views
\item $\twoexp$-hard for CQ queries and MDL views
\item $\twoexp$-hard for MDL queries and a fixed atomic view
\item undecidable for Datalog queries and a fixed atomic view 
\end{compactitem}

The first three bullet items  will follow from a reduction from Datalog equivalence:

\begin{lemma} Let $Q$ and $Q_V$ be arbitrary Datalog queries. Then $Q$ is monotonically determined by $\views = \{(V, Q_V)\}$ iff $Q$ and $Q_V$ are equivalent.
\end{lemma}

\begin{proof} Let $Q = \bigvee_{i=0}^{\alpha} Q_i^1$ and $Q_V = \bigvee_{j = 0}^{\beta} Q_j^2$ with non-empty $Q_i^1$ and $Q_j^2$.

First we show that each $Q_i^1$ satisfies $Q_V$. Indeed, if $Q_V$ is not true on some $Q_i$, then there is a test built on $Q_i$ 
with no atoms. Clearly this test does not satisfy $Q$.

Then we show that each $Q_j^2$ satisfies $Q$. Fix some CQ approximation $Q^1_0$ of $Q$.
We claim that $(Q^1_0, Q^2_j)$ is a test for $Q$ and $\views$ for any $j \in \alpha$. Indeed, $Q_V$ evaluates to true on $Q^1_0$,
and then $V=1$ during the inverse step can be replaced by any $Q^2_j$. Thus $Q^2_j$ must satisfy $Q$.

\end{proof}

The first bullet item now follows from the $\np$-hardness of equivalence of CQs;
the second item follows from the $\Pi_2^p$ hardness of equivalence for UCQs \cite{sagivyann}, while
the third follows from the $\twoexp$-hardness of a CQ
and an MDL query  \cite{mdlcq}.

The results for fixed views follow from a reduction found in \cite{inverserules}:
\begin{lemma}
Let $Q_1$ and $Q_2$ be arbitrary Datalog queries. Consider the query 
$Q = Q_1 \land e \lor Q_2$ where $e$ is a fresh extensional predicate of arity $0$
and a set of views $\views$ which has views $P'$ for all extensional relations  $P$ occurring in
$Q$ except $e$. Then  $Q_1$ is contained in $Q_2$ iff $Q$ is monotonically 
determined by $\views$.
\end{lemma}

\begin{proof}
$(\Rightarrow)$ Note that the tests for $Q$ and $\views$ consist of all CQ approximations of $Q_1$ and $Q_2$. 
It follows that if $Q_1$ is contained in $Q_2$, then all tests pass.

$(\Leftarrow)$ We assume monotonic determinacy and show that $Q_1$ is contained in $Q_2$. 
Pick some CQ approximation $Q_i^1$ of $Q_1$.  Then it's easy to see that
$(Q_i^1\land e, Q_i^1)$ is a test for $Q$ and $\views$. By monotonic determinacy
it follows that $Q_i^1 \models Q$, and so either $Q_i^1 \models Q_1 \land e$ or $Q^1_i \models Q_2$.
The first option is impossible because $Q_i^1$ contains no $e$-atoms. Therefore, $Q^1_i \models Q_2$.
Thus $Q_1$ is contained in $Q_2$.
\end{proof}
The second to the last item now follows from  \cite{cosmadakis1988decidable}
and the last item from \cite{undeciddatalog}, noting that the lower bounds only require a single
extensional predicate.

\subsection*{Proof of Proposition \ref{prop:tpnotmondet}}
Recall the statement:

\medskip

$Q_{TP}$ is not monotonically determined by $\views_{TP}$ iff TP has a solution.

\medskip

We recall the definition of the query and views, giving names to the special views.

\begin{compactenum}
\item $Q_{\qstart} ~ \datalogarrow ~ A(x) \datalogwedge B(x)$
\item $A(x) ~ \datalogarrow ~ \xsucc(x,x') \datalogwedge A(x') \datalogwedge C(x')$
\item $A(x) ~ \datalogarrow ~ \xend(x)$
\item  $B(y) ~ \datalogarrow ~ \ysucc(y,y') \datalogwedge B(y') \datalogwedge D(y')$
\item  $B(y) ~ \datalogarrow ~ \yend(y)$
\item $Q_{\qhelper} ~ \datalogarrow ~ C(u) \datalogwedge \yproj(y,z) \datalogwedge \xproj(x,z)$
\item $Q_{\qhelper} ~ \datalogarrow ~ D(u) \datalogwedge \yproj(y,z) \datalogwedge \xproj(x,z)$\\
\item $Q_{\qverify} ~ \datalogarrow ~ \ha(z_1,z_2,y,x_1,x_2) \datalogwedge T_i(z_1) \datalogwedge T_j(z_2)$ \\
 $\mbox{ for all pairs}$ $(T_i, T_j)\notin HC$
\item  $Q_{\qverify} ~ \datalogarrow ~ \va(z_1,z_2,y_1,y_2,x) \datalogwedge T_i(z_1) \datalogwedge T_j(z_2)$ \\
$\mbox{ for all pairs}$ $(T_i, T_j)\notin VC$
\item  $Q_{\qverify} ~ \datalogarrow ~ \ysucc(o,y) \datalogwedge \ysucc(y,z) \datalogwedge \xsucc(o,x) \datalogwedge
\xproj(x,z) \datalogwedge T_i(z)$ \\
$\mbox{ for all }$ $T_i \notin IT$
\item $Q_{\qverify} ~ \datalogarrow ~ \yend(y) \datalogwedge \yproj(y,z) \datalogwedge T_i(z) \datalogwedge \xproj(x,z) \datalogwedge \xend(x)$ \\
$\mbox{ for all }$ $T_i \notin FT$
\end{compactenum}

The set of views $\views_{TP}$ consists of
\begin{compactitem}
\item[--] the \emph{grid-generating
view}
$$
\begin{array}{rcl}
S(x,y) & \datalogarrow & C(x) \datalogwedge D(y)\\
S(x,y) & \datalogarrow & \xproj(x,z) \datalogwedge T_i(z) \datalogwedge \yproj(y,z) \mbox{ for all } T_i \mbox{ in } Tiles;\\
\end{array}
$$
\item[--]
the \emph{atomic views} $V_\ysucc$, $V_\xsucc$, $V_\yend$, $V_\xend$ and $V_{T_i}$ for EDBs $\ysucc,\xsucc, \yend$, $\xend$ and each
$T_i$ in $Tiles$;
\item[--] the following \emph{special} views
$$
\begin{array}{crcl}
(SP1) & V^\qhelper_C(u,x,y,z) & \datalogarrow & C(u) \datalogwedge \xproj(x,z) \datalogwedge \yproj(y,z) \\
(SP2) & V^\qhelper_D(u,x,y,z) & \datalogarrow & D(u) \datalogwedge \xproj(x,z) \datalogwedge \yproj(y,z) \\
(SP3) &  V_{\ha}(z_1, z_2, y,x_1,x_2) & \datalogarrow & \ha (z_1, z_2, y,x_1,x_2)\\
(SP4) & V_{\va}(z_1, z_2, y_1,y_2,x) & \datalogarrow & \va (z_1, z_2, y_1,y_2,x)\\
(SP5) & V_{I}(o,x,y,z) & \datalogarrow & \xsucc(o,x)\datalogwedge \xproj(x,z) \datalogwedge \ysucc(o,y) \datalogwedge \yproj(y,z) \\
(SP6) & V_{F}(x,y,z) & \datalogarrow & \xproj(x,z) \datalogwedge \xend(x) \datalogwedge \yend(y) \datalogwedge \yproj(y,z) .
\end{array}
$$
\end{compactitem}

We are now ready to begin the proof of Proposition \ref{prop:tpnotmondet}.

\begin{proof}
Suppose that $T = (Q_i, \inst')$ is a test for $Q_{TP}$ and $\views_{TP}$. Following Gogacz and
Marcinkowski \cite{redspider},
we call $Q_i$ \emph{the Green instance} and $\inst'$ \emph{the Red instance} of the test.
We say that $T = (Q_i, \inst')$ is a \emph{main test} if its Green instance is generated from the $Q_{\qstart}$-atom.
Otherwise $T$ is said to be a \emph{side test}. Note that due to the choice of special and atomic views,
all side tests always pass. Also note that all special views are empty when applied to
an approximation of a  $Q_{\qstart}$-atom (see Figure~\ref{fig:unfold}, (a)).

$(\Rightarrow)$
Suppose that $Q_{TP}$ is not monotonically determined by $\views_{TP}$. Then there exists a test 
$T = (Q_i, \inst')$ for $Q_{TP}$ and $\views_{TP}$ that fails $Q_{TP}$. Note that $T$ can't be a side test.
Therefore $T$ must be a main test. Note that there are three kinds of main tests (see Figure~\ref{fig:unfold};
all tests are obtained from (b) by non-deterministic replacement of the $S$-atoms by their definitions):

1) a test in which the second rule of the $S$ view never fires.
In this case the Red instance contains the same $C$ and $D$
atoms as in the Green instance, and hence $Q_{\qstart}$ must
hold.

2) a test in which both rules of the $S$ view fire at least once. In this case, the Red
instance will contain both $C$ facts, $D$-facts, and also some $\xproj$-fact
that joins with some $\yproj$-fact, and thus using SP1-SP2 and $Q_{\qhelper}$ 
we see that $Q$ will hold on the Red instance

3) a test in the second rule of the $S$ view which  fires at least once,
but the first rule never fires. In this case 
the Red instance is isomorphic 
to a grid from the picture with some $T_i$-predicate at each point of the grid.

We claim that these $T_i$-predicates give rise to a correct tiling $\tau$. 
Indeed, as 8) and 9) do not set $Q_{\qverify}$ to $\true$ on $\inst'$, $\tau$ must respect horizontal
and vertical compatibility constraints. Similarly, due to  rules 10) and 11), $\tau$ should have a tile from $IT$ at $(1,1)$ and from $FT$ at $(n,m)$.

$(\Leftarrow)$ Suppose that there are integers $m$ and $n$ and
a tiling  of the $n \times m$ grid with a tile from $IT$ at $(1,1)$ and from $FT$ at $(n,m)$. Then this tiling
(when placed on the $n \times m$ grid in  Figure~\ref{figure:grid}) is $\inst'$ for some grid test of monotonically determinacy. Thus
$Q_{TP}$ is not monotonically determined by $\views_{TP}$.

\end{proof}

\newpage
\section*{Proofs for Section \ref{sec:nonrewrite}: non-rewritability results}
\subsection*{Proof of Fact \ref{fact:pebble-unravellings}}
Recall the statement:

\medskip

Let $k\geq 2$. Let $\inst$ be an instance and $U$ be any $k$-unravelling of $\inst$. Then the following hold:
\begin{compactenum}
\item $U\to \inst$ and $\inst \to_k U$.
\item For every instance $\inst'$, we have $\inst\to_k \inst'$ iff $U\to \inst'$.
\end{compactenum}

\medskip

For the first part, $U \to \inst$ by definition. To see $\inst \to_k U$, we form a strategy for the duplicator inductively,
preserving the invariant that the pebbles of the duplicator are contained in a single bag of the
tree decomposition.  The induction step is accomplished using the second property of
an unravelling.

We turn to the second part, fixing $\inst'$. 
If $U \to \inst'$ via some homomorphism $h$, we can apply
$h$ to the strategy witnessing $\inst \to_k U$ to see $\inst \to_k \inst'$.
Conversely, suppose $\inst \to_k \inst'$. Given $u \in U$  we know there is some some bag
of the tree decomposition containing $u$ with at most $k$ elements, and  $\Theta$ is a partial
isomorphism on this bag. Consider a play for Spoiler in the pebble game from $\inst$ to $\inst'$
going down the branch
of the tree decomposition to $u$.  In this play, once Spoiler moves a pebble off of an element,
he will never move back on to the element. Let $h(u)$ be the element in $\inst'$ corresponding
to $u$ in the response of the duplicator playing according to his winning strategy witnessing
$\inst  \to_k \inst'$. One can verify that $h(u)$ is a homomorphism. 

\subsection*{Proof of Lemma \ref{lemma:special-tiling}}
Recall the statement:

\medskip

There is a tiling instance $TP^*$ such that $\inst_{n,m}^{grid}$ can not be tiled with $TP^*$ for each $n,m \geq 1$
but for each $n,m \geq 3$ and each $k$ with $2 \leq k < \min\{n,m\}$ any  $k$-unravelling
of $\inst_{n,m}^{grid}$  can be tiled with ${TP^*}$.

\medskip

We can rephrase a tiling problem as a homomorphism problem.
For a tiling problem $TP=(Tiles, HC,VC,IT,FT)$, 
we denote by $\inst_{TP}$ the database instance over $\delta=\{\texttt H,V,I,F\}$ with domain $Tiles$ and 
facts ${\texttt H}(T,T')$ (resp. ${\texttt V}(T,T')$) for every $(T,T')\in HC$ (resp. $(T,T')\in VC$), 
and ${\texttt I}(T)$ (resp. ${\texttt F}(T)$) for every $T\in IT$ (resp. $T\in FT$). 
Then an instance can be tiled according to $TP$ exactly when it has a homomorphism
to $\inst_{TP}$.
We can thus rephrase the lemma as:

There is a tiling problem $TP^*$ such that $\inst_{n,m}^{grid}\not\to \inst_{TP^*}$ for each $n,m \geq 1$, 
but $\inst_{n,m}^{grid}\to_k \inst_{TP^*}$ for each $n,m\geq 3$ and each $k$ with $2 \leq k < \min\{n,m\}$. 

\medskip

Before going into the proof, we state a well-known characterization of winning strategies for the Duplicator in the existential pebble game:

\begin{fact}
\label{fact:win-duplicator}
Let $k\geq 2$ and let $\inst,\inst'$ be two instances over the same schema. 
The Duplicator has a winning strategy in the existential $k$-pebble game on $\inst$ and $\inst'$ if and only if there is a non-empty collection ${\mathcal H}$ of partial homomorphisms 
from $\inst$ to $\inst'$ with domain size $\leq k$ such that:
(1) if $f\in {\mathcal H}$ and $g\subseteq f$, then $g\in {\mathcal H}$, and (2) 
for each $f\in {\mathcal H}$ with domain size $<k$ and each $a\in \adom(\inst)$, there is $g\in {\mathcal H}$ with $f\subseteq g$ whose domain contains $a$. 
\end{fact}

\begin{proof}
Our proof is an adaptation of a construction from~\cite{atserias07:power}. 
It was shown in~\cite{atserias07:power} that if an instance $\inst$ has a core of treewidth strictly
bigger than $k$ with $k\geq 2$, then there exists 
an instance  $\inst^*$ such that $\inst \not\to \inst^*$ and $\inst \to_k \inst^*$. 
We could apply this result to each $\inst_{n,m}^{grid}$, where $n,m\geq 3$, and obtain $\inst^*_{n,m}$ such that $\inst_{n,m}^{grid}\not\to \inst^*_{n,m}$ and $\inst_{n,m}^{grid}\to_{k} \inst^*_{n,m}$, for $2\leq k<\min\{n,m\}$. 
By adapting the arguments in~\cite{atserias07:power}, we show that 
the family $\{\inst^*_{n,m}\}_{n,m\geq 3}$ can actually be collapsed into a 
single  instance $\inst_{TP^*}$ with the desired properties.

For $n,m\geq 1$, let $G_{n,m}$ be the $(n\times m)$-grid graph. That is, $\vertices(G_{n,m}):=\{(i,j): 1\leq i\leq n, 1\leq j\leq m\}$ and $\edges(G_{n,m}):=\{\{(i,j),(i',j')\}:|i-i'|+|j-j'|=1\}$. 
Observe that $G_{n,m}$ is precisely the Gaifman graph of the database instance
$\inst_{n,m}^{grid}$. 
Intuitively, a solution for our tiling problem on $G_{n,m}$ will describe a $0/1$ assignment to the edges of the grid $G_{n,m}$. 
In order to define $TP^*$, we consider the grid $G_{3,3}$.
Intuitively, we want to think of grid points within $G_{3,3}$ as ``grid point types'' that can be assigned to a grid point in some larger grid $G_{n,m}$. For example, the tile $(2,1)$ that lies in the center of the lower border represents the type of all elements that lie on the lower border of $G_{n,m}$, excluding the corner points. Our tiles will enhance each abstract grid point with a $0/1$ assignment to its incident edges.

For each vertex $u\in \vertices(G_{3,3})$, we denote by $d_u$ the degree of $u$ (note that $d_u\leq 4$) and fix an enumeration $e^u_1,\dots,e^u_{d_u}$ of all the edges in $G_{3,3}$ that are incident to $u$. 
The set of tiles $Tiles^*$ of $TP^*$ contains all the tuples $(u,b_1,\dots,b_{d_u})$ such that
\begin{enumerate}
\item $u\in \vertices(G_{3,3})$ and $b_1,\dots,b_{d_u}\in \{0,1\}$, 
\item $b_1+\cdots+b_{d_u}\equiv 0$ (mod $2$) if $u\neq (1,1)$, 
\item $b_1+\cdots+b_{d_u}\equiv 1$ (mod $2$) if $u=(1,1)$. 
\end{enumerate}
That is, we consider assignments where the number of edges set to $1$ is odd for the left-lower point but the number of edges set to $1$ is even elsewhere.

Let us denote $\pi_1:Tiles^*\to \vertices(G_{3,3})$ the first-coordinate projection. 
We define the set of initial and final tiles to be $IT^*:=\{t\in Tiles^*: \pi_1(t)=(1,1)\}$ and $FT^*:=\{t\in Tiles^*: \pi_1(t)=(3,3)\}$, respectively.

Our compatibility relation will ensure that the $0/1$ assignment to incident edges is consistent among adjacent nodes: 
if a grid point $n$ has the outgoing edge to its right set to $b\in\{0,1\}$ and $n'$ is the neighbor of $n$ to the right, then $n'$ has the incoming edge to its left set to $b$.

We first give the constraints for pairs of grid points that are assigned to distinct abstract grid points in $G_{3,3}$. 
For each edge $e=\{u,v\}=\{(i,j), (i+1,j)\}\in \edges(G_{3,3})$ with $1\leq i<3$ and $1\leq j\leq 3$, we add to the horizontal compatibility relation $HC^*$ the pair $((u,b_1,\dots,b_{d_u}),(v,b'_1,\dots,b'_{d_v}))$ iff 
$e=e^u_\ell=e^v_m$, for some $\ell,m$ and $b_\ell=b'_m$. 
Similarly, for each edge $e=\{u,v\}=\{(i,j), (i,j+1)\}\in \edges(G_{3,3})$ with $1\leq i\leq 3$ and $1\leq j<3$, we add to the vertical compatibility relation $VC^*$ the pair $((u,b_1,\dots,b_{d_u}),(v,b'_1,\dots,b'_{d_v}))$ iff 
$e=e^u_\ell=e^v_m$, for some $\ell,m$ and $b_\ell=b'_m$. 

We now give the consistency restrictions for pairs of grid points that are assigned the same abstract grid point. 
We add the following pairs to $HC^*$ and $VC^*$:
\begin{itemize}
\item For $u=(2,j)$ with $j\in \{1,3\}$, the pair $((u,b_1,\dots,b_{d_u}),(u,b'_1,\dots,b'_{d_u}))\in HC^*$ iff $e=\{(2,j),(3,j)\}$, $e'=\{(1,j),(2,j)\}$, $e=e^u_\ell$, $e'=e^u_m$, for some $\ell,m$, and $b_\ell=b'_m$, 
\item For $u=(i,2)$ with $i\in \{1,3\}$, the pair $((u,b_1,\dots,b_{d_u}),(u,b'_1,\dots,b'_{d_u}))\in VC^*$ iff $e=\{(i,2),(i,3)\}$, $e'=\{(i,1),(i,2)\}$, $e=e^u_\ell$, $e'=e^u_m$, for some $\ell,m$, and $b_\ell=b'_m$, 
\item For $u=(2,2)$, the pair $((u,b_1,\dots,b_{d_u}),(u,b'_1,\dots,b'_{d_u}))\in HC^*$ iff $e=\{(2,2),(3,2)\}$, $e'=\{(1,2),(2,2)\}$, $e=e^u_\ell$, $e'=e^u_m$, for some $\ell,m$, and $b_\ell=b'_m$; and 
the pair\\
 $((u,b_1,\dots,b_{d_u}),(u,b'_1,\dots,b'_{d_u}))\in VC^*$ iff $e=\{(2,2),(2,3)\}$, $e'=\{(2,1),(2,2)\}$, $e=e^u_\ell$, $e'=e^u_m$, for some $\ell,m$, and $b_\ell=b'_m$.
\end{itemize}

Let $n,m\geq 3$. We define a function $\Psi$ from $\vertices(G_{n,m})$ to $\vertices(G_{3,3})$ as follows. We let $\Psi((1,1))=(1,1)$, $\Psi((n,1))=(3,1)$, $\Psi((1,m))=(1,3)$ and $\Psi((n,m))=(3,3)$. 
For $1<i<n$ and $1<j<m$, we define $\Psi((i,j))=(2,2)$, $\Psi((1,j))=(1,2)$, $\Psi((n,j))=(3,2)$, $\Psi((i,1))=(2,1)$ and $\Psi((i,m))=(2,3)$. 
We can now enumerate incident edges of $a$ in $G_{n,m}$ according to the already-defined enumeration for $\Psi(a)$ in $G_{3,3}$. 
For each $a\in \vertices(G_{n,m})$, we define a bijection $\Delta_a$ from its incident edges in $G_{n,m}$ to the incident edges of $\Psi(a)$ in $G_{3,3}$ in the natural way: 
 if $e$ corresponds to the incident edge of $a$ to the 
``up'' direction in the grid $G_{n,m}$ then $\Delta_a(e)$ is also the incident edge of $\Psi(a)$ in the 
grid $G_{3,3}$ to the ``up'' direction; similarly for the ``right'', ``down'' and ``left" directions.
 Then for each $a\in \vertices(G_{n,m})$, we enumerate its incident edges as $e_1^a,\dots,e_{d_a}^a=\Delta^{-1}_a(e_1^{\Psi(a)}),\dots,\Delta^{-1}_a(e_{d_{\Psi(a)}}^{\Psi(a)})$, where 
$e_1^{\Psi(a)},\dots,e_{d_{\Psi(a)}}^{\Psi(a)}$ is the enumeration for $\Psi(a)$ already fixed in the construction of $TP^*$.

We now formalize the intuition that the parity and consistency conditions ensure that
a rectangular grid cannot be tiled:
\begin{claim}
$\inst_{n,m}^{grid}\not\to \inst_{TP^*}$, for every $n,m\geq 1$. 
\end{claim}

\begin{proof}
Note that $\inst_{n,m}^{grid}\not\to \inst_{TP^*}$ if $\min\{n,m\}\leq 2$. Towards a contradiction, suppose $\inst_{n,m}^{grid}\to \inst_{TP^*}$ for some $n,m\geq 3$, via a homomorphism $h$. 
By construction, we must have $\pi_1(h(a))=\Psi(a)$, for every $a$ in $\inst_{n,m}^{grid}$ and hence $h$ corresponds to a $0/1$ assignment of the edges of the Gaifman graph $G_{n,m}$ of $\inst_{n,m}^{grid}$. 
In particular, there exists a $0/1$ vector $(x_e)_{e\in E(G_{n,m})}$ such that for each $a\in \vertices(G_{n,m})$, we have $h(a)=(\Psi(a),x_{e_1^a},\dots,x_{e_{d_a}^a})$. 
Now we have 
\begin{align*}
\sum_{a\in \vertices(G_{n,m})} (x_{e_1^a}+\cdots+x_{e_{d_a}^a})=\\
(x_{e_1^{(1,1)}}+\cdots+x_{e_{d_{(1,1)}}^{(1,1)}})+\sum_{a\in \vertices(G_{n,m})\setminus \{(1,1)\}} (x_{e_1^a}+\cdots+x_{e_{d_a}^a})\\
&=1 \qquad (\text{mod $2$})
\end{align*}
But this is impossible as each edge $e\in \edges(G_{n,m})$ is counted exactly twice in $\sum_{a\in \vertices(G_{n,m})} (x_{e_1^a}+\cdots+x_{e_{d_a}^a})$; a contradiction. 
\renewcommand{\qedsymbol}{$\blacksquare$}
\end{proof}

While there is no total mapping from $\inst_{n,m}^{grid}$ to $\inst_{TP^*}$ that is a homomorphism, by considering partial mappings 
with domains that are not too large, we can easily satisfy the correct parity conditions, and hence we can define partial homomorphisms  from $\inst_{n,m}^{grid}$ to $\inst_{TP^*}$.
The next claim tells us that these partial homomorphisms can be chosen to be consistent.

 \begin{claim}
$\inst_{n,m}^{grid}\to_k \inst_{TP^*}$, for every $n,m\geq 3$ and $2\leq k<\min\{n,m\}$. 
\end{claim}
\begin{proof}
Let $P=(a_0,a_1,\dots, a_\ell)$ be a walk in $G_{n,m}$. For every edge $e\in \edges(G_{n,m})$, we define:
\begin{enumerate}
\item $x_e^P=1$ if $P$ visits $e$ an odd number of times.
\item $x_e^P=0$ if $P$ visits $e$ an even number of times.
\end{enumerate}
We also define $h^P(a):=(\Psi(a),x_{e_1^a}^P,\dots,x_{e_{d_a}^a}^P)$, for each $a\in \vertices(G_{n,m})$ (i.e., in the domain of $\inst_{n,m}^{grid}$).

Let ${\mathcal W}$ be the collection of all 
walks $P=(a_0,a_1,\dots, a_\ell)$ in $G_{n,m}$ with $a_0=(1,1)$ and $a_\ell\neq a_0$. 
We claim that for each $P=(a_0,a_1,\dots, a_\ell)\in {\mathcal W}$ 
and each $a\neq a_\ell$ in $\vertices(G_{n,m})$, the tuple $h^P(a)$ always belongs to the domain of $\inst_{TP^*}$.  
Note that 
$x^P_{e_1^a}+\cdots+x_{e_{d_a}^a}^P=|\{\text{$e$ in $P$}:\text{$e$ is incident to $a$}\}|$ (mod $2$). 
For $a\neq a_0$, we have $|\{\text{$e$ in $P$}:\text{$e$ is incident to $a$}\}|=2\cdot |\{i: \text{$0< i<\ell$ and $a_i=a$}\}|=0$ (mod $2$), 
and hence $h^P(a)=(\Psi(a),x_{e_1^a}^P,\dots,x_{e_{d_a}^a}^P)$ 
belongs to $\inst_{TP^*}$ (as $\Psi(a)\neq (1,1)$). 
On the other hand, for $a=a_0$, we have $|\{\text{$e$ in $P$}:\text{$e$ is incident to $a$}\}|=1+2\cdot |\{i: \text{$0< i<\ell$ and $a_i=a$}\}|=1$ (mod $2$), and hence $h^P(a)=(\Psi(a),x_{e_1^a}^P,\dots,x_{e_{d_a}^a}^P)$ belongs to $\inst_{TP^*}$ (as $\Psi(a)=(1,1)$). 

Thus we can define for each walk 
$P=(a_0,a_1,\dots, a_\ell)\in {\mathcal W}$ a partial mapping $h^P$ from $\inst_{n,m}^{grid}$ to $\inst_{TP^*}$ with domain $\vertices(G_{n,m})\setminus\{a_\ell\}$. 
By definition of $TP^*$ and since $h^P$ is defined from a $0/1$ vector $(x^P_e)_{e\in \edges(G_{n,m})}$, we have that $h^P$ is actually a partial homomorphism. 

We define a non-empty collection $\mathcal H$ of partial homomorphisms from $\inst_{n,m}^{grid}$ to $\inst_{TP^*}$ as follows. 
For $1\leq p\leq n$ and $1\leq q\leq m$, we 
denote by $C_{p,q}$ the \emph{$(p,q)$-cross} of $G_{n,m}$ defined as $C_{p,q}:=\{(p,j): 1\leq j\leq m\}\cup \{(i,q): 1\leq i\leq n\}$. 
For every non-empty subset $S\subseteq \vertices(G_{n,m})$ with $|S|\leq k$ (recall that $2\leq k<\min\{n,m\}$), and every 
walk $P=(a_0,\dots,a_\ell)\in {\mathcal W}$ such that there are $p,q$ with $a_\ell\in C_{p,q}$ and $C_{p,q}\cap S=\emptyset$, 
we add to $\mathcal H$ the restriction $h^P|_{S}$. 
We prove that $\mathcal H$ is a winning strategy for the Duplicator and 
then $\inst_{n,m}^{grid}\to_k \inst_{TP^*}$ as required. 
Condition (1) of Fact~\ref{fact:win-duplicator} holds by definition, so we focus on condition (2). 
Let $h^P|_{S}\in {\mathcal H}$ for some $S$ with $|S|<k$ and walk $P=(a_0,\dots,a_\ell)\in {\mathcal W}$ such that $a_\ell\in C_{p,q}$ and $C_{p,q}\cap S=\emptyset$ for some $p,q$. Let $a\in \vertices(G_{n,m})\setminus S$ and 
$S'=S\cup\{a\}$. Since $k<\min\{n,m\}$, there exist $p',q'$ such that $C_{p',q'}\cap S'=\emptyset$. Moreover, since $C_{p,q}$ is connected and $|C_{p,q}\cap C_{p',q'}|\geq 2$, there is a walk $P''=(a_\ell, a_{\ell+1}, \dots, a_{\ell+r})$ such that $a_{\ell+i}\in C_{p,q}$, for all $0\leq i\leq r$, $a_{\ell+r}\in C_{p',q'}$ and $a_{\ell+r}\neq a_0$. Let $P'=(a_0,\dots,a_{\ell+r})$ be the concatenation of $P$ and $P''$. 
Then $h^{P'}|_{S'}\in {\mathcal H}$. Finally, observe that $h^P(b)=h^{P'}(b)$, for every $b\in S$, since $x^P_e$ and $x^{P'}_e$ can only differ for edges $e=\{b',b''\}\subseteq C_{p,q}$ and $C_{p,q}\cap S=\emptyset$. 
It follows that $h^{P}|_S\subseteq h^{P'}|_{S'}$, and hence condition (2) holds.
\renewcommand{\qedsymbol}{$\blacksquare$}
\end{proof}

\end{proof}

\subsection*{Additional comments on non-Datalog-rewritable examples}
We mentioned in the conclusion of the paper that for the
example query $Q_{TP^*}$ in views produced in the proof of  Theorem \ref{thm:no-datalog-mdl-ucq}
there is a rewriting in a slightly larger language,
stratified Datalog. The details of stratified Datalog will not concern us here,
except that it includes positive Boolean combinations
of  Datalog queries and relational algebra queries. We will show
that the example has a rewriting that is such a Boolean combination.
We now explain this. In fact, what we show
is that for every tiling problem $TP$ for which rectangular grids can not be tiled,
the query $Q_{TP}$ from Theorem \ref{thm:undec} has a rewriting that is a positive
Boolean combination of Datalog  queries and relational algebra queries.
In particular, this show that $Q_{TP}$ always has a separator in $\ptime$.

Denote by $Q^*_{\qstart}$ the query obtained from $Q_{\qstart}$ by
replacing $C$ and $D$ by the first and second projections of $S$, respectively.
Let $Q^*_\qverify$ be obtained from $Q^\qverify$ by using the views.
That is, by replacing:
\begin{itemize}
\item CQ $\ha$ by view $V_{\ha}$ and similarly for $\va$,
\item relations $T_i$ by the corresponding atomic views 
\item rewriting rules corresponding to  the second to last bullet item
 as $V_I(o,x, y,z), V_{T_i}(z)$, and similarly rewriting rules corresponding
to the final bullet item using $V_F$.
\end{itemize}

Let $\producttest$ be a query that tests whether $S$ is the product of its
projections. $\producttest$ can be expressed in relational algebra, hence in stratified
Datalog.

Consider  the query $R$ formed by existentially quantifying
\begin{align*}
V^{\qhelper}_C \vee V^{\qhelper}_D \vee Q^*_\qverify \vee (Q^*_\qstart \wedge \producttest)
\end{align*}

Clearly $R$ is a positive Boolean combination of Datalog queries
and the relational algebra query $\producttest$. We claim that $R$ is a rewriting of $q$.

In one direction, suppose $Q$ returns true on $\inst$ and let $\vinst$ be the view
image.
We do a case analysis depending on which of the top-level disjuncts holds.
If $Q_{\qhelper}$ holds on $\inst$ then $V^{\qhelper}_C$ or $V^{\qhelper}_D$ is non-empty, and thus $R$ holds in $\vinst$.
If  $Q_{\qverify}$ holds on $\inst$ then $Q^*_\qverify$ holds on $\vinst$
and hence we conclude again that $R$ holds on $\vinst$.
Finally, suppose $Q_{\qstart}$ holds on $\inst$.
If $\producttest$ fails, we know one of $C$ or $D$ is empty. But then $Q_{\qstart}$ cannot hold,
a contradiction to our assumption. 
Thus $\producttest$ must hold. From this, it is easy to see that $Q^*_{\qstart}$ holds.
This completes the proof of this direction.

Conversely suppose that $R$ holds on the view image $\vinst$. 
Again we do a case analysis on the top-level disjuncts.
If $V^{\qhelper}_C$ or $V^{\qhelper}_D$ is nonempty on $\vinst$, then $Q_{\qhelper}$ holds on
$\inst$ and hence $Q$ holds on $\inst$.
If $Q^*_\qverify$ holds on $\vinst$, then $Q_\qverify$ holds on $\inst$, 
and again we conclude that $Q$ holds on $\inst$.
Finally, suppose $Q^*_\qstart \wedge \producttest$ holds on $\vinst$, and
suppose that none of the disjuncts of $Q$ hold.
Note that since $Q_{\qhelper}$ fails, $V^{\qhelper}_C$ and $V^{\qhelper}_D$ must
be empty. Thus we have two possibilities for $S$.
There is the ``projection case'', where either one of $C$ or $D$ is empty,  and all the $S$ atoms are generated by the
second rule. The alternative is the ``product case'', where both  $C$ and $D$ are both nonempty and
all the atoms of $S$ are generated by the first rule.

We claim that we must be in the ``product case'' for $S$ above.
If we are in the projection case, then
every pair must be associated with a tile. Further, since $Q_{\qverify}$ and $Q_{\qhelper}$ fail,
we have a tiling of a rectangular grid, contradicting the hypothesis that there is no tiling.
Since we have argued that we are in the product case, it follows that
 $Q_\qstart$ holds on $\inst$ and
thus $Q$ holds in $\inst$ as required.

\subsection*{Proof of Theorem \ref{thm:nocomputable}}
Recall the statement:

\medskip

There is no integer-valued function $F$ such that
for all $Q, \views$ such that $\views$ and $Q$ are in Datalog and
$Q$ is monotonically determined over $V$, there is a separator
of $Q$ over $\views$
that runs in time 
$F(\views(I))$.

\medskip

We now give  the proof of Theorem  \ref{thm:nocomputable}.
We assume the opposite, aiming for a contradiction.
We use the following fact, which is a consequence of the time
hierarchy theorem:

\medskip

For any computable
function $F$ there is a deterministic Turing machine $M_F$ which halts on all of its inputs, and such
that no Turing machine running in time $F$ can decide the same language as $M_F$.

\medskip

Fix such a machine $M$ for $F$. 

Let $\Sigma_{\tminput}$ be the input alphabet
of $M$, and $\Sigma_{M}$ be a suitable alphabet for encoding configurations of 
$M$. 

We consider a base signature with relations $\tmsucc(x,y), U_a(x): a \in \Sigma_{\tminput}$ for
the input signature of $M$ along with symbols $\tmsucc'(x,y), U'_a(x): a \in \Sigma_M$ for
the configuration signature of $M$.

A \emph{pre-run-string} is a string in the
regular language formed by intersecting
\[
\sigma_{\inputbegin}  ~ (\Sigma_{\tminput})^* ~ \sigma_{\inputend} (\Sigma_{M}^*  \configsep )^+  \sigma_{\runend}
\]
with a regular expression  enforcing that the last maximal segment
of $\Sigma_M$ strings that does not contain $\configsep$ encodes a halting state.
Above:
\begin{itemize}
\item $\sigma_{\inputbegin}$ is a marker designating the beginning of the input while
\item $\sigma_{\inputend}$ designates the end of the input;
\item  $\configsep$ is a marker indicating the separator between configurations, while
\item $\sigma_{\runend}$ 
marks the end of the run.
\end{itemize}
A \emph{well-shaped string} will consist of an initial letter with a special
symbol $\sigma_{\inputbegin}$ and ending with $\sigma_{\inputend}$, 
followed by
a  code for a run of $M$, ending with a special symbol $\sigma_{\runend}$. A string
is \emph{badly-shaped} if it is not well-shaped.
 It is easy to
see that a badly-shaped string $w$ has at least one of the following \emph{bad
properties}: $w$ is not a pre-run string,  $w$ contains a sub-string
$\configsep  c_i \configsep c_{i+1};$ where $c_{i+1}$ does not encode a next configuration after $c_i$,
 $w$ contains a string $\sigma_{\inputbegin} w_{in} \sigma_{\inputend} c_1 \configsep$ such that $c_1$ does
not encode initial configuration of $M$ with input $w_{in}$.

A \emph{pre-run instance} will be a relational encoding of a homomorphic image of a pre-run string using
the relations  $\tmsucc(x,y), U_a(x): a \in \Sigma_{\tminput}$ for the coding of the initial
segment, symbols $\tmsucc'(x,y), U'_a(x): a \in \Sigma_M$ for the remaining part of the run,
and additional symbols for the separators. That is, in the relational encoding we allow the same element
to represent different places in the string.
A \emph{well-shaped string instance} will be a relational encoding of a homomorphic image of a well-shaped
string, again using the relations $\tmsucc(x,y), U_a: a \in {\Sigma}_{\tminput}$ for
 the initial segment and the
primed copies for the remaining segments. We define a badly-shaped string instance analogously.

A standard argument shows 
\begin{proposition} \label{prop:badlyshaped} There
is a Datalog query whose approximations are (up to isomorphism)
exactly the badly-shaped string instances.
\end{proposition}

Note that if we had enforced that codings were \emph{alternating}, with every other configuration
reversed, then we could use 
a  PDA to detect bad properties on a string and a context-free path query 
to detect it on the encoding. With the power of general Datalog, no alternation
is needed.

Our views $\views$ will include:

\begin{compactitem}

\item the \emph{input views},  with one binary view  returning exactly $\tmsucc(x,y)$, and for
each $a \in \Sigma_{input}$ a unary view
returning $U_a(x)$.

\item a nullary view $V^{\badlyshaped}$ 
which returns $\true$ whenever the instance contains
a badly-shaped string instance. That is,
$V^{\badlyshaped}$ returns $\true$ when the input
contains the homomorphic image of a relational encoding of a string starting with the symbol 
$\sigma_\inputbegin$ and ending with the symbol $\sigma_{\inputend}$
which has one of the bad properties. By Proposition \ref{prop:badlyshaped},
a Datalog view with this property exists.

\item a  unary view $V^{\prerun}(x)$ which holds for $x$ if there
is a subinstance that is  a pre-run instance 
in which the occurrence of $\sigma_{\inputend}$ corresponds to $x$.
\end{compactitem}

Our query $Q$ will be the sentence obtained from 
$V^{\badlyshaped}$ disjoined with $Q^{\tmaccept}$, where
$Q^{\tmaccept}$ returns true exactly when we detect a relational encoding
of a pre-run string that
ends in an accept state.

We now argue that $Q$ is monotonically  determined over $\views$.

Consider instances $\inst_1$ and $\inst_2$ with $\views(\inst_1) \subseteq  \views(\inst_2)$ and $Q(\inst_1)$ being true.

$Q(\inst_1)$ could be true because $V^{\badlyshaped}$ holds,
in this case, $Q(\inst_2)$ also holds since $V^{\badlyshaped}$ is one of the views.
So we can assume that  $V^{\badlyshaped}$ does not
hold in $\inst_1$ or $\inst_2$,
since if it does hold then $\inst_2$ satisfies $Q$.

$Q(\inst_1)$ could also be true because $V^{\badlyshaped}$ fails
but $Q^{\tmaccept}$ holds. We know there is a relational encoding
of some string
\[
\sigma_\inputbegin w \sigma_\inputend w_0 \configsep w_1 \configsep \ldots \configsep w_n \sigma_{\runend}
\]
witnessing
that  $Q^{\tmaccept}$ holds  in $\inst_1$. Let $x$ be the element
corresponding to the label $\sigma_{\inputend}$ in this encoding.
Note that $V^{\prerun}$ must  hold of 
$x$  in $\inst_1$, hence in $\inst_2$.
The latter must be witnessed  via a 
relational encoding of some string
of the form 
\begin{align*}
\sigma_{\inputbegin} w'_0 \sigma_{\inputend}  w'_1 \configsep \ldots \configsep w'_k
\end{align*}
 with $\sigma_{\inputbegin} w'_0 \sigma_{\inputend}$
relationally encoded in the 
unprimed signature,  the $w'_i$ encoded in the  primed signature, with
the element labelled by $\sigma_{\inputend}$ corresponding to $x$.  Note
that by the definition of pre-run, $w'_k$ must include a halting state.

Since we have views for all of the input signature elements, and $\views(\inst_1) \subseteq  \views(\inst_2)$,
we  know that we also have an encoding of a string $\sigma_{\inputbegin} w_0 \sigma_{\inputend}$
in $\inst_2$, with the encoding done in the unprimed signature, with the

We now consider the string 
\[
s= \sigma_\inputbegin ~ w_0 ~ \sigma_{\inputend} ~ w'_1 \configsep \ldots \configsep w'_k ~ \sigma_{\runend}
\]
$s$ begins with the input string, and ends with a halting state. Note that
since a relational encoding of $\sigma_\inputbegin ~ w_0 ~ \sigma_{\inputend}$ lies in $\inst_1$,
the encoding of $s$ must lie in $\inst_2$, due to the input views.
Since
$V^{\badlyshaped}$ is false in $\inst_2$, we know that in $\inst_2$:

\begin{compactitem}
\item For every relational encoding of a string of form:

\[
\sigma_{\inputbegin} ~ w_0  ~ \sigma_{\inputend} ~ w'_1 \configsep
\]

with $\sigma_{\inputbegin} ~ w_0  ~ \sigma_{\inputend}$ encoded  in the unprimed signature
and $w'_1$ 
encodes a state  with
tape configuration $w_0$ and state the initial
state of $M$, under the transition relation of $M$.

\item  For every relational encoding a string of the form:

\[
w'_1 \configsep w'_2 \configsep
\]

with the encoding being in the primed signature, 
$w'_2$ must encode a state that is a successor in the transition
relation of $M$ of the state encoded by $w'_1$.

\end{compactitem}

From this we infer that 
$s$  is an encoding of a run of $M$ on $w_0$, ending at a halting state.

But since $M$ is deterministic, $s$ must be the same as
\[
\sigma_{\inputbegin} ~ w_0 ~ \sigma_{\inputend}  w_1 \configsep \ldots \configsep w_n \sigma_{\runend}
\]
which ends in an acceptance state.

Since a relational encoding of $s$ lies in $\inst_2$,
 we can conclude that  $Q$ holds in $\inst_2$. This completes the argument
for monotonic determinacy of $Q$ with respect to $\views$.

Now, suppose $Q$ has a separator $R$ that runs in time
$F$. Then $R$ will allow us to check in time $F$ whether
$M$ accepts or rejects on its input, a contradiction.
Thus we have completed the proof of Theorem \ref{thm:nocomputable}.

\end{document}